\documentclass[journal, onecolumn]{IEEEtran}
\usepackage{amsmath,amsfonts}
\usepackage{array}
\usepackage{textcomp}
\usepackage{stfloats}
\usepackage{url}
\usepackage{verbatim}
\usepackage{graphicx}
\usepackage{caption}
\usepackage{subcaption}
\usepackage{cite}
\usepackage{dsfont}
\usepackage{color}
\usepackage{amssymb}
\usepackage{amsthm}
\usepackage{algorithmic}
\usepackage[ruled,vlined]{algorithm2e}
\usepackage{booktabs}
\usepackage{tabularx}      
\usepackage{enumitem} 
\usepackage{multirow}
\usepackage{autobreak}
\usepackage{anyfontsize} 
\usepackage[colorlinks,
linkcolor=blue,
anchorcolor=blue,
citecolor=blue,
]{hyperref}
\theoremstyle{plain}
\newtheorem{theorem}{Theorem}
\newtheorem{lemma}{Lemma}
\newtheorem{proposition}{Proposition}

\newtheorem{assumption}{H}
\theoremstyle{definition}
\newtheorem{definition}{Definition}
\newtheorem{remark}{Remark}

\begin{document}

\title{Exponential-Family Tensor Completion via Nonconvex Dual Total-Variation Regularization}

\author{Wenfei Cao,  Yang Chen, Qibin Zhao,  Jinglai Li, Andrzej Cichocki~\IEEEmembership{Life Fellow,~IEEE}
\thanks{Manuscript received  September XXX, 2025; revised XXX 17, 2026. Corresponding author: Jinglai Li (e-mail: j.li.10@bham.ac.uk).}

\thanks{W.-F. Cao and Y. Chen  are with the School of Mathematics and Statistics, Shaanxi Normal University, Chang'an District, Xi'an 710119, China.}
\thanks{Q.-B. Zhao is with the Tensor Learning Team in RIKEN AIP, 1-4-1 Nihonbashi, Chuo-ku, Tokyo 103-0027, Japan.}

\thanks{J.-L. Li and W.-F. Cao are with the School of Mathematics, University of Birmingham, Birmingham B15 2TT, United Kingdom.}

\thanks{A. Cichocki is with Systems Research Institute of Polish Academy of Science, Poland and RIKEN AIP, Japan. 
}
}


\maketitle
\begin{abstract}%
With the emergence of various tensor data, tensor completion from partial measurements has attracted widespread attention in data science and signal processing. Total Variation (TV) has been widely used as an effective regularization technique for tensor completion; however, theoretical studies on TV regularization in this context remain limited. In this work, we present a rigorous theoretical analysis of TV regularization for tensor completion. Specifically, we consider tensor completion under exponential-family noise, which generalizes the standard settings such as Gaussian and Poisson tensor completion. To handle exponential-family tensor completion, we propose a family of dual-TV (DTV) regularizers based on the transformed L1 function, which simultaneously capture sparsity and low-rank structures in the gradient tensor. Moreover, 
we establish the theoretical upper bounds on 
the recovery error of the proposed estimator.
In certain cases, these upper bounds can attain the convergence order of  $\mathcal{O}\big( n_3 r_t\big(\max_{k} s_k^2\big) \log\big((n_1+n_2)n_3\big) /n \big)$,  and the minimax lower bound analysis is further presented to show that  the upper-bounds  can approach the lower bound with the gap of  order $\mathcal{O}(\max_k s_k^2/max(n_1, n_2))$ up to a logarithmic factor.
Finally,  multiple groups of experiments on synthetic, image and video tensor data sets are conducted to  support our theoretical results and demonstrate the effectiveness of our method.
\end{abstract}

\noindent%
\emph{Keywords.} Tensor completion, Tensor-structure modeling,  Total Variation, Exponential-family noise, Theoretical guarantee.

\bigskip
\noindent \textbf{Nomenclature}
\begin{description}[leftmargin=7em, style=sameline]
    \small
    \item[$m, M$:] \( m := \min\{n_1, n_2\},\; M := \max\{n_1, n_2\} \)
    \item[$\boldsymbol{X}^H$:] The conjugate transpose of \( \boldsymbol{X} \in \mathbb{C}^{n_1 \times n_2} \)
    \item[$\boldsymbol{X}^\dagger$:] The Moore-Penrose inverse of \( \boldsymbol{X} \)
    \item[$\boldsymbol{D}_{n_k}$:] 1st-order difference matrix for mode \( k \)
    \item[$\|\boldsymbol{X}\|_*$:] The nuclear norm of \( \boldsymbol{X} \)
    \item[$\|\boldsymbol{X}\|$:] The spectral norm of \( \boldsymbol{X} \)
    \item[$\mathcal{X}_o$:] The \( o \)-th entry of \( \mathcal{X} \in \mathbb{C}^{n_1 \times n_2 \times \dots \times n_d} \)
    \item[$\mathcal{X}^{(i_3,\dots,i_d)}$:] The face slice \( \mathcal{X}(:,:,i_3,\dots,i_d) \) of \( \mathcal{X} \)
    \item[$\mathcal{X}_L$ \text{or} $\mathcal{L}(\mathcal{X})$:] Transformation of tensor \( \mathcal{X} \) w.r.t. \( \mathcal{L} \)
    \item[$\mathcal{L}^{-1}(\mathcal{Y})$:] Inverse transformation of tensor \( \mathcal{Y} \) w.r.t. \( \mathcal{L} \)
    \item[$\mathcal{X}^H$:] Conjugate transpose of tensor \( \mathcal{X} \) w.r.t. \( \mathcal{L} \)
    \item[$\nabla_k(\mathcal{X})$:] Gradient tensor of \( \mathcal{X} \) along the \( k \)-th mode
    \item[$\nabla \Phi(\mathcal{X})$:] Gradient of function \( \Phi(\mathcal{X}) \) w.r.t. \( \mathcal{X} \)
    \item[$\mathcal{A} *_{\mathcal{L}} \mathcal{B}$:] Tensor-tensor product between \( \mathcal{A} \) and \( \mathcal{B} \) w.r.t. \( \mathcal{L} \)
    \item[$\mathcal{A} \Delta \mathcal{B}$:] Face-wise product between \( \mathcal{A} \) and \( \mathcal{B} \)
    \item[$\langle \mathcal{X}, \mathcal{Y} \rangle$:] The inner product between \( \mathcal{X} \) and \( \mathcal{Y} \)
    \item[$\|\mathcal{X}\|_F$:] Frobenius norm of tensor \( \mathcal{X} \)
     \item[$\|\mathcal{X}\|_\infty$:] Infinity norm of tensor \( \mathcal{X} \)
  \item[$\|\mathcal{X}\|_{\circledast, \mathcal{L}}$:] Nuclear norm of tensor \( \mathcal{X} \) w.r.t. \( \mathcal{L} \)
    \item[$(r_1, \dots, r_{n_d})$:] Transformed multi-rank of tensor \( \mathcal{X} \)
    \item[$\tilde{r}$ \text{or} $\mathrm{rank}_s(\cdot)$:] Sum of transformed multi-rank of tensor \( \mathcal{X} \)
    \item[$r_t$ or $\mathrm{rank}_t(\cdot)$:] Transformed tubal-rank of tensor \( \mathcal{X} \)
    \item[$a \preceq b$ \text{or} $a = \mathcal{O}(b)$:] \( a \leq C b \) for some positive constant \( C \)
    \item[$a \asymp  b$:] \( a = \mathcal{O}(b) \) and \( b = \mathcal{O}(a) \).
\end{description}

\section{Introduction}
Tensors, as a generalization of vectors and matrices, arise in a variety of applications including recommender systems~\cite{Bi2018}, social networks~\cite{Nickel2011}, data mining~\cite{Papalexakis2016,Kolda2009}, machine learning~\cite{Cichocki2016},  and computer vision~\cite{Panagakis2021}. An important reason for the wide applicability is the effective representation of those data using tensor structures. 
For example, the data in recommender systems can be naturally described as a three-way tensor of user $\times$ item $\times$ context where each entry indicates the user-item interaction under a particular context. 
However, due to the hardware limitation of data acquisition/transmission or the inaccessibility of some private items, we have to collect some incomplete observations of overall tensor data in many practical applications. Therefore, tensor completion~\cite{Song2019} is a fundamental and essential task for the subsequent information-processing procedures. 
\subsection{Notation and Problem Setup}
\paragraph{Notation}To begin with, let us introduce some necessary notations.  $x$, $\boldsymbol{X}$ and $\mathcal{X}$ stands for  scalar, matrix and tensor, respectively.  For any positive integers $n$, $[n]:=\{1,2,\cdots,n\}$.   $\boldsymbol{1}_{x} = 1$ if $x \neq 0$; otherwise $\boldsymbol{1}_{x} = 0$.  For a third-order tensor $\mathcal{X} \in \mathbb{C}^{n_1\times n_2 \times n_3}$, we denote its $(i,j,k)$-th entry as $\mathcal{X}_{i,j,k}$, and employ Matlab notations $\mathcal{X}(i,:,:)$, $\mathcal{X}(:,j,:)$, and $\mathcal{X}(:,:,k)$ to stand for its $i$-th horizontal, lateral, and frontal slice, respectively. More often, the $i$-th frontal slice of tensor $\mathcal{X}$ is denoted as $\mathcal{X}^{(i)}$ for simplicity. The tubes of $\mathcal{X}$ are denoted as $\mathcal{X}(i,j,:)$. More notations can be found in the \textbf{Nomenclature}.
\paragraph{Tensor Completion Problem}
Tensor completion  refers to recovering the  true tensor $\mathcal{X} \in \mathbb{R}^{n_1 \times n_2  \times \cdots \times n_d}$ from a limited number of measurements corrupted by certain additive noise as follows:
\begin{equation*}
\mathcal{Y}(w_t) = \mathcal{X}(w_t) + \mathcal{E}(w_t), \:  w_t \in \Omega:=\left\{ w_1, w_2, \cdots, w_n\right\},
\end{equation*}
where $\Omega$ is the observed index set, $n$ is the number of the effective observed entries, and $\mathcal{E}(w_t)$ denotes i.i.d. additive noise random variables. It is known that tensor completion is a severely ill-posed inverse problem
in the field of data science and statistical signal processing.

\subsection{Related Work on Total-Variation Regularization}
Different types of total variations (TV) formulations are proposed to regularize the solution space to make the tensor completion problem tractable.  The classical isotropic TV~\cite{rudin1992nonlinear} is employed to handle image inpainting~\cite{getreuer2012total} that can be considered as a matrix completion problem:
\begin{equation}
\label{intro.eq1}
\mathop{\min}_{ \boldsymbol{X} \in \mathbb{R}^{n_1 \times n_2}}  \mathrm{TV}(\boldsymbol{X} )+ \frac{\lambda}{2} \sum_{w_t \in \Omega} \left( \boldsymbol{Y}(w_t)  -\boldsymbol{X}(w_t) \right)^{2},
\end{equation}
where the isotropic TV is of the form:
\begin{equation*}
\mathrm{TV}(\boldsymbol{X} )  = \sum_{i,j} \sqrt{ |\nabla_1 (\boldsymbol{X})_{i,j}|^2 +  |\nabla_2 (\boldsymbol{X})_{i,j}|^2 },
\end{equation*} 
with $\nabla_1$ and $\nabla_2$ (see Definition \ref{gradient})  being the horizontal and vertical gradient operator of image $\boldsymbol{X}$,  respectively.   The classical TV is extended to the vector TV (VTV\cite{getreuer2012total,chambolle2016a}) which is applied color image inpainting,
formulated as a 3rd-order tensor completion problem:
\begin{equation}
\label{intro.eq2}
\mathop{\min}_{ \mathcal{X} \in \mathbb{R}^{n_1 \times n_2 \times 3}}  \mathrm{VTV}( \mathcal{X} ) + \frac{\lambda}{2} \sum_{w_t \in \Omega}  \left( \mathcal{Y}(w_t)  -\mathcal{X}(w_t) \right)^{2}.
\end{equation}
Here, the VTV term is define as the following:
\begin{equation*}
\mathrm{VTV}(\mathcal{X}) =  \sum_{i,j} \sqrt{ \sum_{k=1}^{3} \Big(|\nabla_1 (\mathcal{X}^{(k)})_{i,j}|^2 +  |\nabla_2 (\mathcal{X}^{(k)})_{i,j}|^2 \Big)}, 
\end{equation*} with $\mathcal{X}^{(k)}$ being the $k$-th frontal slice of $\mathcal{X}$. The anisotropic TV (ATV\cite{cai2022approx})  is also investigated in the multi-dimensional image inpainting task:
\begin{equation}
\label{intro.eq3}
\mathop{\min}_{\mathcal{X}} \mathrm{ATV}(\mathcal{X}),~ s.t. \:   \frac{1}{n}  \sum_{w_t \in \Omega} \left( \mathcal{Y}(w_t)  - \mathcal{X}(w_t) \right)^{2}\le \eta,
\end{equation}
where  the ATV term is of the form:
\begin{equation*}
\mathrm{ATV}(\mathcal{X}) =  \sum_{k}  \sum_{i_1, i_2, i_3}   \vert \nabla_k(\mathcal{X})_{i_1,i_2,i_3} \vert,
\end{equation*} and $\eta$ is a tuning parameter related to the noise level. 

In addition to the smoothness captured by various types of TV,  the low-rankness of the underlying tensor is also a useful intrinsic structure for tensor completion.  Considering the joint modeling of tensor low-rankness and smoothness by imposing the matrix nuclear norm and weighted ATV (WATV) on each unfolding matrix of the underlying tensor, an exact tensor completion model \cite{Li2017a} is proposed as follows:
\begin{equation}
\label{intro.eq4}
\begin{split}
&\mathop{\min}_{\mathcal{X}  \in \mathbb{R}^{n_1 \times n_2 \times n_3} }   \mathrm{WATV}(\mathcal{X})+ \lambda_{x}  \| \mathcal{X} \|_{\mathrm{SNN}} \\
&~~~~s.t.\: \: \mathcal{Y}(w_t) = \mathcal{X}(w_t),~w_t \in \Omega,
\end{split}
\end{equation}
where  $\mathrm{WATV}(\mathcal{X}) =\sum_{k=1}^{3} \beta_k \| \nabla_1(\mathcal{X}_{(k)}) ||_{\ell_1}$, and $\| \mathcal{X} \|_{\mathrm{SNN}} =  \frac{1}{3} \sum_{k=1}^{3} \| \mathcal{X}_{(k)} \|_{*}$ with $\mathcal{X}_{(k)}$ being the mode-$k$ unfolding matrix. Later, considering robustness to gross outliers, a robust tensor completion model \cite{Qiu2021b} was introduced by employing the transformed tensor nuclear norm as follows:
\begin{equation}
\label{intro.eq5}
\mathop{\min}_{\mathcal{X},\mathcal{S} }    \phi(\mathcal{X}) + \lambda_{s} \| \mathcal{S}\|_{\ell_1}  + \frac{\rho}{2}\sum_{w_t \in \Omega} \big(\mathcal{Y} (w_t)- \mathcal{X}(w_t) -\mathcal{S}(w_t)\big)^{2},
\end{equation}
where $ \phi(\mathcal{X}) =    \mathrm{ATV}( \mathcal{X})+\lambda_{x} \| \mathcal{X} \|_{\circledast, \mathcal{L}}$ encodes both tensor low-rankness and smoothness, and the mathematical form of $\| \mathcal{X} \|_{\circledast, \mathcal{L}}$ can be found in Definition~\ref{def9}.  Recently, a novel tensor TV variant, called  tensor correlated total variation (t-CTV~\cite{Wang2023a}), is proposed for exact tensor completion by fusing tensor low-rankness and smoothness simultaneously:
\begin{equation}
\label{intro.eq6}
\begin{split}
\mathop{\min}_{\mathcal{X}}  \mathrm{t\text{-}CTV}(\mathcal{X}),~s.t.\: \mathcal{Y}(w_t) = \mathcal{X}(w_t),~w_t \in \Omega,
\end{split}
\end{equation}
where $ \mathrm{t\text{-}CTV}(\mathcal{X}) =\frac{1}{3} \sum_{k=1}^{3}  \|\nabla_{k} (\mathcal{X}) \|_{\circledast,  \mathcal{L}}$\footnote{For simplicity, we here present a special case of the t-CTV.}.
In contrast to the composite regularization of tensor nuclear norm and tensor TV in models~\eqref{intro.eq4} and~\eqref{intro.eq5},  t-CTV considers modeling low-rankness of gradient tensors instead of its sparsity, 
which leads to a byproduct that the tough issue of tuning parameter $\lambda_x$ can be avoided.
Thanks to this merit, t-CTV is widely utilized to solve other tensor recovery problems, e.g., Poisson tensor completion problem by formulating the following model\cite{Feng2024a}:
\begin{equation}
\label{intro.eq7}
\begin{split}
\mathop{\min}_{\mathcal{X}}  \| \mathcal{X} \|_{\mathrm{t\text{-}CTV}} + \lambda \Phi_{\Omega, \mathcal{Y}}(\mathcal{X}) , ~s.t. \:  \alpha_{l} \le \mathcal{X}_{i,j,k} \le  \alpha_u, 
\end{split}
\end{equation}
where $ \Phi_{\Omega, \mathcal{Y}}(\mathcal{X})$ indicates the Poisson loss on the observed index set $\Omega$, and $\alpha_{l}, \alpha_u$ are the lower and upper bounds of box constraint, respectively. In addition,  various TVs  are combined with different types of tensor decompositions to solve tensor completion problem with  satisfactory performance. Such as TV with matrix factorization\cite{Zeng2024a}, TV with CP and Tensor Train (TT)  decomposition \cite{Yokota2016a,  Ko2020a},  TV with Tucker decomposition~\cite{Li2017a}, etc. For more methodologies about tensor completion, see related works in Section~\ref{sect2}.

\begin{figure}[htbp]
\centering
\includegraphics[width=0.65\textwidth]{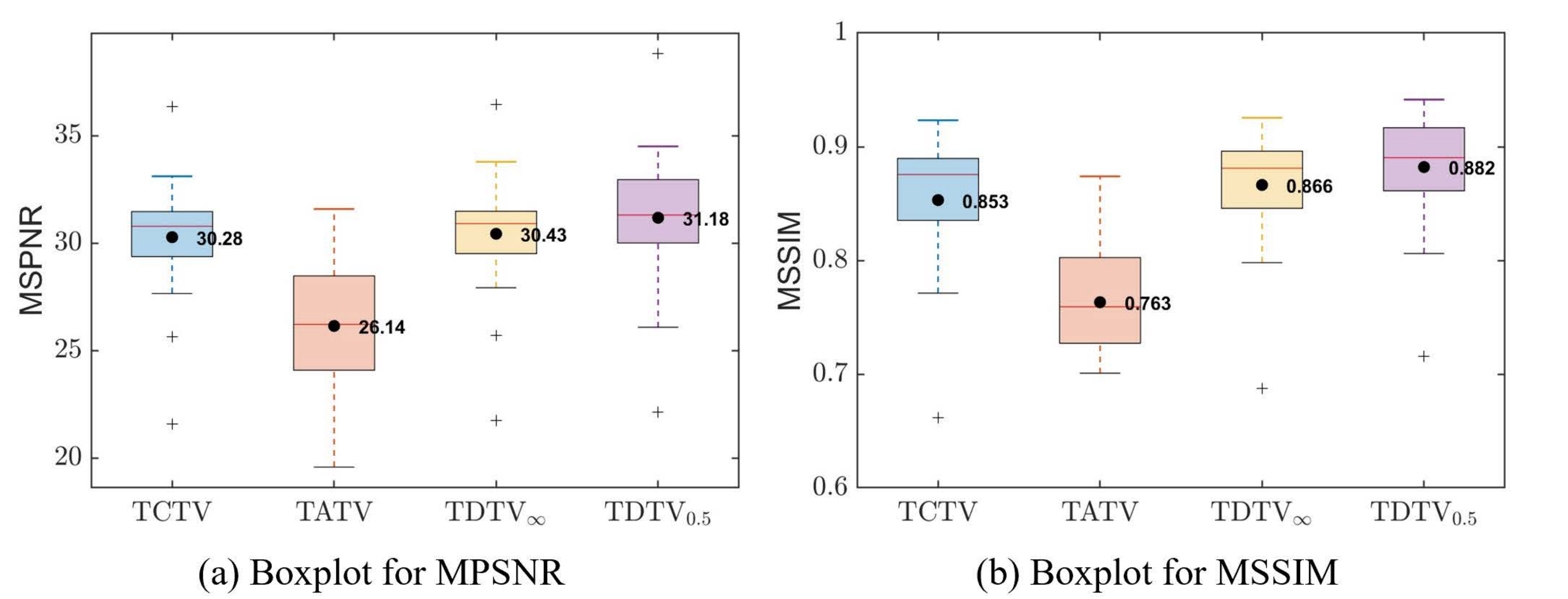}
\caption{Boxplots of (a) MPSNR and (b) MSSIM over 20 videos for Gaussian tensor completion with a sampling ratio of 0.2. TCTV exploits the low-rank prior, TATV the sparsity prior, while TDTV$_{\infty}$ and TDTV$_{0.5}$ jointly exploit both.}
\label{intro:fig1}
\end{figure}
\subsection{Our Motivation and Contribution}
In this work, we aim to advance current research on TV-regularized tensor completion from several perspectives. First, we extend existing methods to handle a broader range of noise beyond the noise-free and Gaussian settings. Second, in many practical problems, gradient tensors exhibit coefficients concentrated around zero and rapidly decaying singular values (see Fig.~\ref{fig0} (c)–(d) for an illustration). This motivates modeling both sparsity and low-rankness in the gradient domain. Moreover, the tensor singular values of gradient tensors may exhibit diverse decay patterns across different slices, indicating that the underlying low-rank structures can vary significantly, as illustrated in Fig.~\ref{fig0}(d). This motivates the use of a more flexible regularization scheme. In particular, we introduce a nonconvex tensor regularization that can adaptively characterize the tensor singular values. In Fig.~\ref{intro:fig1}, we compare our methods with two main baselines: TCTV, which promotes low-rankness, and TATV, which promotes sparsity. The details of both the proposed methods and the comparative approaches will be provided later. Here, we briefly note that the convex TDTV$_{\infty}$ model achieves improved performance over both baselines, while the nonconvex variant TDTV$_{0.5}$ further improves upon TDTV$_\infty$. 

Our main technical contributions can be summarized as follows.
\begin{itemize}
	\item We propose a family of nonconvex dual tensor TV regularizers (DTV) that simultaneously and adaptively capture the sparsity and low-rankness of visual tensors in the gradient domain. Based on this regularization, we formulate a unified empirical risk minimization model, denoted by TDTV$_a$ ($0<a\le\infty$), for tensor completion under exponential-family noise, including Gaussian, Poisson, and one-bit noise as special cases.
	\item We establish an upper bound on the recovery error of the proposed TDTV$_a$ model. For a special case, we further derive a corresponding minimax lower bound, showing that the upper bound  can approximate the lower bound  with a gap of order $\mathcal{O}(\max_k s_k^2/M)$ up to a logarithmic factor. Notably, our theoretical analysis is developed for the nonconvex regularization setting, whereas the existing results~\cite{cai2022approx, Wang2023a, Feng2024a} are primarily limited to convex models. Further details on the proof techniques are provided in Section~\ref{discuss}.
	\item Within the framework of the alternating direction method of multipliers (ADMM), we develop an effective optimization algorithm to solve the proposed model. Furthermore, multiple sets of experiments on synthetic tensors, videos, and multispectral images are conducted to demonstrate the validity of our recovery theory and the effectiveness of our method.
\end{itemize}

The rest of the paper is organized as follows.  Section 2 reviews related works. In Section 3, we introduce necessary preliminaries.  Section 4 presents the proposed  tensor completion  method and Section 5 develops an effective algorithm.   The numerical experiments are provided in Section 6 and this paper is concluded  in the last section.
\section{Related Works}\label{sect2}
\subsection{Continuous-Valued Tensor Completion}
Low-rankness is a widely used prior knowledge that characterizes the global correlation of intrinsic  information in a tensor.  
For example, the sum of matrix nuclear norms (SNN)~\cite{Liu2012} is proposed as a convex surrogate of Tucker rank for tensor completion. Based on the well-known tensor~singular~value~decomposition (t-SVD)~\cite{Kilmer2011}, Zhang et al.~\cite{Zhang2016a} utilize a new tensor nuclear norm for exact tensor completion with theoretical guarantee. Since the discrete Fourier transform in the classic t-SVD  cannot sufficiently adapt to the target tensor data, some extension works~\cite{Kernfeld2015,Song2023a,Li2022a} are subsequently proposed for tensor completion using inverse linear transform, unitary transform and nonlinear neural-network transform. As another efficient  representation of low-rank tensor,  tensor decompositions with different  structures have also been widely exploited for tensor completion.  Ashraphijuo and Wang~\cite{Ashraphijuo2017} investigate the fundamental conditions  on the sampling pattern for low-CP-rank tensor completion. Zhang and Xia~\cite{Xia2018a} further explore the underlying Tucker low-rank structure from the noisy tensor observations with a comprehensive analysis on tensor SVD statistically and computationally. Later, the low-rank tensor-factorization methods~\cite{Zhou2017a} using tensor $t$-product are proposed for 3rd-order tensor completion. 
Similarly,  low-rank tensor ring (TR) decomposition~\cite{Qiu2024a} is also  employed to accomplish tensor completion.

In recent years, the learning-based methods are attracting increasing attention from domain scholars. For instance, under the t-SVD framework, the idea of linear-transform learning~\cite{Wu2024a,Liu2024a} is exploited for tensor completion and meanwhile the nonlinear-transform-learning~\cite{ Li2023a} strategy is also leveraged for tensor completion. 
Using the well-known algorithm unfolding technique, Mai et al. \cite{Mai2024a} propose a novel unfolding network under the ADMM algorithm framework by considering the attention mechanism to better preserve the structure of the original tensor and designing implicit regularizers to compensate for modeling inaccuracies. 
\subsection{Discrete-Valued Tensor Completion}
Recent decades have witnessed many instances of discrete-valued tensors, such as binary tensors, in which all tensor entries are binary indicators. For example, click/no-click action in recommender system~\cite{Bi2018} and presence/absence of edges in multi-relational social networks~\cite{Nickel2011}. These binary tensors are often missing, noisy and high-dimensional, thus necessitating one-bit tensor completion. Currently, the existing works for one-bit tensor completion are very limited. Aidini et al.~\cite{Aidini2018a} proposes a seminal method for one-bit tensor completion by considering the low-rankness of matricizations of the underlying tensor. In the CP decomposition framework, Ghadermarzy et al.~\cite{Ghadermarzy2019a} construct the max-qnorm and atomic M-norm as robust proxies of the CP rank for one-bit tensor completion,  and theoretically prove that the proposed methods achieve near-optimal sample complexity. 
Hou et al.~\cite{Hou2021a}  come up with an improved method for one-bit tensor completion using the tensor nuclear-norm induced by general inverse linear transform and provide a solid theoretical recovery guarantee. On the basis of this work,  Cao et al.~\cite{Cao2024a} propose  tensor max-norm and provide a  rigorous theoretical recovery guarantee for one-bit tensor completion with tensor nuclear-norm and tensor max-norm composite regularization. Recently, Wang and Li~\cite{Wang2020a} develop a probabilistic tensor decomposition method for binary tensor data with statistical optimality guarantee. Likewise, Poisson tensor completion, as an instance of discrete-valued tensor completion,  has also attracted increasing attention in recent years. For example,  Zhang et al. \cite{Zhang2022a} extend Poisson matrix completion to the 3rd-order tensor case using the t-SVD induced tensor nuclear norm and provide a rigorous upper-and-lower bound analysis. Later, this work is extended in~\cite{Feng2024a} for Poisson tensor completion using more powerful tensor t-CTV regularization function. 
\section{Preliminaries}
In this section,  some preliminaries are reviewed. 
\begin{definition}[Transformed $L_1$ function~\cite{zhang2018aTL1}]The transformed $L_1$ (TL1) function is defined as:
\begin{equation}\label{TL1}
\text{TL1}_a(x) = \frac{(a+1)x}{a+x},~\text{with the parameter}~a \in (0, \infty).
\end{equation}
\end{definition}
This function is used to promote the sparsity of vector $\boldsymbol{x}$ and the low-rankness of matrix $\boldsymbol{X}$ by the vector $\text{T}\ell_1$ norm and the matrix $\text{TL}_1$ norm, which are defined respectively:
\begin{align}
\| \boldsymbol{x} \|_{\text{T}\ell_1^{a}}:= \sum_{i=1}^{n_1} \text{TL1}_a(|x_i|),~\| \boldsymbol{X} \|_{\text{TL}_1^{a}} : = \sum_{j=1}^{m} \text{TL1}_a(\sigma_{j}(\boldsymbol{X})). \notag 
\end{align}Also, this function is continuous with respect to the internal parameter $a$, and it has two useful limits:
\begin{equation}
\label{TL1:1}
\lim \limits_{a\rightarrow 0+} \text{TL1}_a(x) = \boldsymbol{1}_{x},~\lim \limits_{a \rightarrow \infty} \text{TL1}_a(x) = x,
\end{equation}
which will be used in Remark~\ref{TCTV}.

\begin{definition}[transform and inverse transform of tensor~\cite{Qin2022a}]
Given a transform $\mathcal{L}$, the transform of a high-order tensor $\mathcal{X}$ is defined as:
$$
\mathcal{X}_{\mathcal{L}}:=\mathcal{L}(\mathcal{X}) = \mathcal{X} \times_3 U_{n_3}  \times_4  \cdots   \times_d U_{n_d} ,
$$
where $\times_j$ denotes mode-$j$ product between tensor and matrix,  $U_{n_j}$ is transform matrix of size $n_{j} \times n_{j}, j=3,\cdots,n_d$, such as the discrete~Fourier~transform~(DFT) and discrete~consine~transform~(DCT) matrices.  Its inverse transform is given by:
$$
\mathcal{L}^{-1}(\mathcal{X}) = \mathcal{X} \times_3 U_{n_3}^{-1}  \times_4  \cdots   \times_dU_{n_d} ^{-1},
$$
satisfying $\mathcal{L}^{-1}(\mathcal{L}(\mathcal{X}) ) = \mathcal{X}$. 
\end{definition}
Here, the transformation matrices $\{ U_{n_j}\}_{j=3}^{d}$ of $\mathcal{L}$ are assumed to satisfy 
$$
(U_{n_d}^{H} \otimes \cdots \otimes U_{n_1}^{H}  )\cdot(U_{n_d} \otimes \cdots \otimes U_{n_1} ) = \ell\cdot I_{n_3\cdots n_d},
$$
where $\otimes$ denotes Kronecker product, $I$ denotes identity matrix and $\ell>0$ is a specified scale factor corresponding to the transformation, e.g.,  $\ell=\prod_{j=3}^d n_j$ for DFT matrix $F_{n_j}$ since  $F_{n_j}^{H}F_{n_j}=n_jI_{n_j}$, and $\ell = 1$ for  DCT matrix $C_{n_j}$ since $C_{n_j}^{H}C_{n_j}=I_{n_j},~j=3,\cdots,d$.

\begin{definition}[tensor-tensor product~\cite{Qin2022a}]
For two order-$d$ tensors $\mathcal{A} \in \mathbb{R}^{n_1 \times q \times n_3 \times \cdots \times n_d}$ and $\mathcal{B} \in \mathbb{R}^{q \times n_2 \times n_3 \times \cdots \times n_d}$,  the $\mathcal{L}$-based product is defined as:
$$
\mathcal{A} \ast_{\mathcal{L}} \mathcal{B} = \mathcal{L}^{-1}(\mathcal{L}(\mathcal{A}) \Delta \mathcal{L}(\mathcal{B})),
$$
where $\Delta$ denotes the face-wise product; that is $\mathcal{Z} = \mathcal{X} \Delta \mathcal{Y} \Leftrightarrow  \mathcal{Z}^{(i_3,\cdots,i_d)} = \mathcal{X}^{(i_3,\cdots,i_d)} \mathcal{Y}^{(i_3,\cdots,i_d)}$ for all face slices.
\end{definition}

In order to introduce tensor t-SVD, let us first recall the definition of conjugate transpose, identity tensor, unitary tensor, and $f$-diagonal tensor.
\begin{definition}[conjugate transpose~\cite{Qin2022a}]
	For any $\mathcal{X} \in \mathbb{C}^{n_1 \times n_2  \times n_3 \times \cdots \times n_d}$, its conjugate transpose with respect to $\mathcal{L}$, denoted by $\mathcal{X}^{H} \in \mathbb{C}^{n_2 \times n_1 \times n_3 \times \cdots \times n_d}$,  satisfies that
	\begin{displaymath}
		 ( \mathcal{X}^{H} )_{\mathcal{L}}  (:,:,i_3, \cdots, i_d) =  \big( \mathcal{X}_{\mathcal{L}}(:,:,i_3, \cdots, i_d) \big)^H,
	\end{displaymath}
	for all frontal slices.
\end{definition}

\begin{definition}[identity tensor~\cite{Qin2022a}]
	An order-$d$ tensor $\mathcal{I}_{n} \in \mathbb{C}^{n\times n \times n_3 \times \cdots \times n_d}$ is called as identity tensor if it satisfies $\mathcal{I}_{\mathcal{L}}(:,:,i_3,\cdots,i_d) = I_n$ for all  frontal slices.
\end{definition}

\begin{definition}[unitary tensor~\cite{Qin2022a}]
An order-$d$ tensor $\mathcal{U} \in \mathbb{C}^{ n\times n \times n_3 \cdots \times n_d}$ is unitary if $\mathcal{U}^{H} \ast_{\mathcal{L}} \mathcal{U} = \mathcal{U} \ast_{\mathcal{L}} \mathcal{U} ^{H} = \mathcal{I}_n$.
\end{definition}
\begin{definition}[$f$-diagonal tensor~\cite{Qin2022a}]
An order-$d$ tensor $\mathcal{U} \in \mathbb{C}^{ n\times n \times n_3 \cdots \times n_d}$ is $f$-diagonal  if all its face slices are diagonal.
\end{definition}

Based on the above definitions, we present tensor t-SVD with respect to $\mathcal{L}$ in the following proposition.
\begin{proposition}[transformed t-SVD~\cite{Qin2022a}]\label{prop1}
For any order-d tensor  $\mathcal{X} \in \mathbb{C}^{ n_1\times n_2 \times n_3 \cdots \times n_d}$, it can be decomposed as: 
\begin{equation*}
\mathcal{X} = \mathcal{U} \ast_{\mathcal{L}} \mathcal{S} \ast_{\mathcal{L}} \mathcal{V}^{H},
\end{equation*}
where $\mathcal{U} \in \mathbb{C}^{ n_1\times n_1 \times n_3 \cdots \times n_d}$ and $\mathcal{V} \in \mathbb{C}^{ n_2\times n_2 \times n_3 \cdots \times n_d}$ are unitary and $\mathcal{S} \in \mathbb{C}^{ n_1\times n_2\times n_3 \cdots \times n_d}$ is $f$-diagonal.
\end{proposition} 

Based on the transformed t-SVD above, we now define the transformed multi-rank and tubal rank of a high-order tensor. 
\begin{definition}[transformed multi-rank and tubal-rank~\cite{Qin2022a}]
For a tensor $\mathcal{X} \in \mathbb{C}^{ n_1\times n_2 \times n_3 \cdots \times n_d}$  with the transformed t-SVD $\mathcal{X} = \mathcal{U} \ast_{\mathcal{L}} \mathcal{S} \ast_{\mathcal{L}} \mathcal{V}^{H}$,  (a) its transformed multi-rank $\mathrm{rank}_m(\mathcal{X})$ is a vector $\boldsymbol{r} \in \mathbb{R}^{n_3n_3\cdots n_d}$ with the $i$-th entry being the rank of the $i$-th frontal slice of $\mathcal{X}_{\mathcal{L}}$, i.e., 
\begin{equation*}
\mathrm{rank}_m(\mathcal{X}) = (r_1, r_2, \cdots, r_{n_3n_4 \cdots n_d}), 
\end{equation*}
where  $r_i = \text{rank}(\mathcal{X}_{\mathcal{L}}(:,:, i))$ for $ i\in[n_3n_4\cdots n_d]$. (b) Its transformed tubal-rank, denoted as $r_t$,  is defined as the count of nonzero singular tubes of $\mathcal{S}$, i.e., 
\begin{equation*}
rank_{\text{t}}(\mathcal{X})= card\{ i: \mathcal{S}(i,i,:,\cdots,:) \neq 0\} = \max_{i} \{ r_i\},
\end{equation*}
where $card\{\cdot\}$ denotes the cardinality of a set.  (c) The sum of transformed multi-rank is $\widetilde{r} = r_1 + r_2 +\cdots +r_{n_3 n_4 \cdots n_d}$, and it is easy to see that $\widetilde{r} \le (n_3n_4 \cdots n_d) r_t$.
\end{definition}

For tackling tensor recovery, more useful concepts need to be introduced, such as tensor nuclear norm, tensor spectral norm, gradient tensor, and so on.
\begin{definition}[transformed tensor nuclear norm,  transformed tensor spectral norm~\cite{Qin2022a}]\label{def9}
For a tensor $\mathcal{X} \in \mathbb{C}^{ n_1\times n_2 \times n_3 \cdots \times n_d}$,  its transformed tensor nuclear norm with respect to $\mathcal{L}$ is defined as 
\begin{equation*}
\|\mathcal{X}\|_{\circledast, \mathcal{L}}:=\frac{1}{\ell} \sum_{i_3} \cdots \sum_{i_d} \| \mathcal{X}_{\mathcal{L}}(:,:,i_3,\cdots,i_d)\|_{*}.
\end{equation*}
The transformed tensor spectral norm with respect to $\mathcal{L}$ is defined as
\begin{equation*}
\|\mathcal{X}\|_{\mathcal{L}}:= \underset{i_3,\cdots,i_d}\max\| \mathcal{X}_{\mathcal{L}}(:,:,i_3,\cdots,i_d)\|.
\end{equation*}
\end{definition}
\begin{definition}[gradient tensor~\cite{Wang2023a, Liu2024b}]\label{gradient}
For $\mathcal{X} \in \mathbb{R}^{n_1 \times n_2 \times \cdots \times n_d}$, its gradient tensor along the $k$-th mode is defined as
\begin{equation}
\nabla_k (\mathcal{X}) = \mathcal{X} \times_k \boldsymbol{D}_{n_k}, ~ k=1,2,\cdots,d, 
\end{equation}
where $\boldsymbol{D}_{n_k} \in \mathbb{R}^{ (n_k -1) \times n_k}$ is the first-order difference matrix, i.e., 
$\boldsymbol{D}_{n_k}(i, i)=-1$, $\boldsymbol{D}_{n_k}(i, i+1)=1$ and the other entries  equal to 0.
\end{definition}
\begin{definition}[gradient tensor sparsity] For $\mathcal{X} \in \mathbb{R}^{n_1 \times n_2 \times \cdots \times n_d}$, its gradient-tensor sparsity on mode-$k$  is defined as
\begin{equation}
 s_k:=  \|  \nabla_k (\mathcal{X})  \|_{\ell_0}, 
\end{equation}
where $\| \cdot \|_{\ell_0}$ denotes the number of non-zero entries in a tensor.
\end{definition}
\section{Guaranteed Tensor Completion with Exponential-Family Noise}\label{sect3}
\subsection{Observation Model with Exponential Family Noise}
Gaussian noise and Laplacian noise are two typical noise settings in many works of  matrix and tensor recovery. However, the two noise settings are far from enough in practical applications. For example,  in the image and video completion, the observation noise usually follows Poisson noise in the dark-light scenario. Bernoulli and  Rademacher distributions are often exploited  in the one-bit matrix completion of signal communication.  In this paper, we shall consider a more general noise setting; that is the well-known exponential-family noise in the field of statistics and machine learning.
\begin{definition}[Exponential family distribution~\cite{Jean2015, Alaya2019}]
A random variable $Y$ is said to follow the natural exponential family distribution, if its probability density function characterized by the parameter $\zeta$ is given by:
\begin{equation}
\label{exponential}
Y|\zeta  \sim f_{h,F}(y|\zeta) = h(y)\exp(\zeta y -F(\zeta)),
\end{equation}
where $h$ is the base measure function and $F$ is called as the log-partition function.  $F(\zeta)$ is strictly convex with respect to the parameter $\zeta$.
\end{definition}

If $F $ is smooth enough, we have that $\mathbb{E}[Y]=F'(\zeta)$ and $\mathbb{V}ar[Y]=F''(\zeta)$, where $F'$ and $F''$ stand for the first and second derivative of $F$, respectively. The exponential family encompasses a wide range of standard distributions as follows:
\begin{itemize}
\item Normal, $\mathcal{N}(\mu, \sigma^2)$ (known $\sigma$), is typically used to model continuous data, with natural parameter $\zeta=\frac{\mu}{\sigma^2}$ and $F(\zeta)=\frac{\zeta^2\sigma^2}{2}$.
\item Gamma,  $\Gamma(\lambda,\alpha)$ (known $\alpha$), is often used to model positive valued continuous data, with natural parameter $\zeta=-\lambda$ and $F(\zeta)=-\alpha \log(-\zeta)$.
\item Negative binomial, $\mathcal{NB}(p,r)$ (known $r$), is often used to model over-dispersed count data, whose variance is larger than their mean, with natural parameter $\zeta=\log(1-p)$ and $F(\zeta)=-r\log(1-\exp(\zeta))$.
\item Binomial, $\mathcal{B}(p,N)$ (known $N$), is used to model number of successes in $N$ trials, with natural parameter $\zeta=\log(\frac{p}{1-p})$ and $F(\zeta)=N\log(1+\exp(\zeta))$.
\item Rademacher, $\mathcal{R}(p)$,  is a discrete probability distribution where a random variate $X$ takes $+1$ with probability $p$ and $-1$ otherwise, with natural parameter $\zeta=\frac{1}{2}\log(\frac{p}{1-p})$ and $F(\zeta)= \frac{1}{2}\log( 2+ \exp(-2\zeta) + \exp(2\zeta))$.
\item Poisson, $\mathcal{P}(\lambda)$, is used to model count data, with natural parameter $\zeta=\log(\lambda)$ and $F(\zeta)=\exp(\zeta)$.
\end{itemize}

The parameter configurations of the above distributions are listed in Table~\ref{expdist}. It is noteworthy that the  exponential, chi-squared, Rayleigh, Bernoulli and geometric distributions are special cases of the above distributions.
\begin{table}[htbp]
	\setlength{\abovecaptionskip}{0cm}
	\setlength{\belowcaptionskip}{0cm}
	\caption{Some common distributions in the Exponential family.}
	\renewcommand\arraystretch{1.2}
	\setlength{\tabcolsep}{4pt}
	\centering
        \scriptsize
\begin{tabular}{c|c|c} \toprule
		Distributions & $\zeta$  &  $F(\zeta)$  \\ \midrule
		Normal $\mathcal{N}(\mu,\sigma^2)$ &  $\frac{\mu}{\sigma^2}$ & $\frac{\zeta^2\sigma^2}{2}$\\
		Gamma $\Gamma(\lambda,\alpha)$  & $-\lambda$ & $-\alpha\log(-\zeta)$ \\
		Negative binomial $\mathcal{NB}(p,r)$ &  $\log(1-p)$ & $-r\log(1-\exp(\zeta))$\\
		 Binomial $\mathcal{B}(p,N)$  & $\log(\frac{p}{1-p})$   & $N\log(1+\exp(\zeta))$ \\
		 Rademacher $\mathcal{R}(p)$ &  $\frac{1}{2}\log(\frac{p}{1-p})$  & $\frac{\log( 2+ \exp(-2\zeta) + \exp(2\zeta))}{2}$\\
		  Poisson $\mathcal{P}(\lambda)$ & $\log(\lambda)$ & $\exp(\zeta)$ \\	
\bottomrule
\end{tabular}
\label{expdist}
\end{table}
 
Different from traditional tensor completion problem which directly observes a subset of the targeted tensor,  we access to the elements of a noisy tensor $\mathcal{Y}$ on the subset $\Omega$,  conditional  on the  targeted tensor $\mathcal{X}$:
\begin{equation}
	\label{observ0}
	\mathcal{Y}_{w_t}|\mathcal{X}_{w_t}   \sim f_{h,F}(\mathcal{Y}_{w_t}|\mathcal{X}_{w_t}) = h(\mathcal{Y}_{w_t})\exp(\mathcal{X}_{w_t}  \mathcal{Y}_{w_t} -F(\mathcal{X}_{w_t})),
\end{equation}
where $w_t \in \Omega$.  Here, $\Omega = \{ w_1, w_2, \cdots,w_n\}$ is an index set of i.i.d. random variables with probability distribution $\Pi=\{ \pi_{i_1, i_2,\cdots, i_d}\}$ on $[n_1] \times [n_2] \times \cdots \times [n_d]$, which satisfies
\begin{equation}\label{pidist}
	\mathbb{P} \big\{ w_t = (i_1, i_2,\cdots, i_d) \big\} = \pi_{i_1, i_2,\cdots, i_d},
\end{equation}
for all $t$ and $(i_1, i_2,\cdots, i_d)$. 

In the existing tensor completion studies, the measurements with Gaussian-noise, Poisson-noise, and  Rademacher-noise are the  special cases of  observation model~\eqref{observ0}. In addition,  the indices of the observation entries are assumed to follow a general  discrete probability distribution \eqref{pidist} instead of the common assumption of uniform distribution.

\begin{figure}[htbp]
\centering
\includegraphics[width=0.99\textwidth]{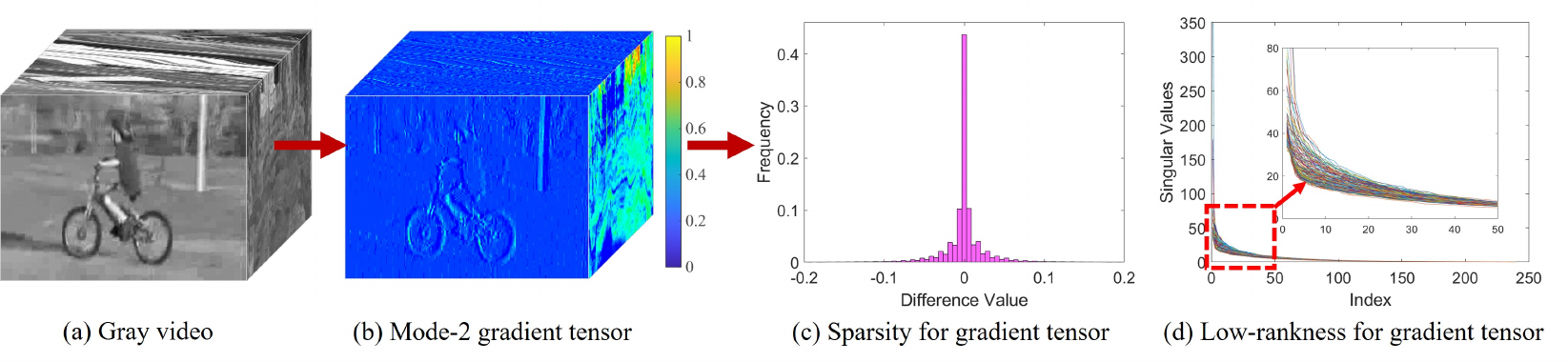}
\caption{Illustration of the simultaneous low-rankness and sparsity of the gradient tensor (see Definition~\ref{gradient}). (a) A grayscale video of size $240 \times 320 \times 170$; (b) the corresponding gradient tensor along mode-2; (c) the frequency histogram of all entries in the gradient tensor; (d) the decay curves of the tensor singular values of the gradient tensor under the transformed t-SVD with DCT.}
\label{fig0}
\end{figure}
\subsection{Tensor Completion with Nonconvex DTV Regularization} 
Tensor modeling is vitally important to tackle with many tensor inverse problems, such as tensor completion,  tensor robust PCA, and so on.  The reason is that the fine modeling of the targeted tensor can help to regularize the solution space of tensor inverse problems, hence producing a high-quality recovery tensor.  

In Fig.\ref{fig0} (c), we exhibit the histogram of the gradient tensor with respect to mode 2. From this histogram, it can be seen that the empirical distribution possesses most of the probability mass  around zero, indicating that the gradient tensor possesses approximate sparsity property. 
\begin{definition}[TATV$_a$] For $\mathcal{X} \in \mathbb{R}^{n_1 \times n_2 \times \cdots \times n_d}$,  let $\Gamma$ be the prior set consisting of mode directions. The tensor anisotropic total-variation (TATV) is defined as:
\begin{equation}
\label{reg1}
\text{TATV}_{a}(\mathcal{X}) :=\frac{1}{\gamma} \sum_{k \in \Gamma}  \| \nabla_k (\mathcal{X})  \|_{\text{T}\ell_1^a} ,
\end{equation}
where  $a$ is the internal parameter in the TL1 function,  $\gamma$ is the cardinality of set $\Gamma$, and $\| \mathcal{G}\|_{\text{T}\ell_1^a}=\sum_{i_1} \cdots \sum_{i_d} \text{TL1}_a( \vert \mathcal{G}(i_1, \cdots ,i_d)\vert)$. 
\end{definition}

Several recent works~\cite{Wang2023a,Liu2024b, Feng2024a,Liu2024c,Hou2024a,  KaiHuang2024} suggest that the low-rankness modelings in the gradient domain are also effective for tensor recovery. From Fig. \ref{fig0}(d), it can be observed that  the singular values of the gradient tensor decay dramatically toward a small value, indicating that the gradient tensors possess approximately low-rank property.
\begin{definition}[TCTV$_a$]\label{TCTVa}
For $\mathcal{X} \in \mathbb{R}^{n_1 \times n_2 \times \cdots \times n_d}$,  let $\Gamma$ be the prior set consisting of mode directions.  The tensor correlated total-variation (TCTV) is defined as:
\begin{equation}
\label{reg2}
\text{TCTV}_{a}(\mathcal{X}) :=\frac{1}{\gamma} \sum_{k \in \Gamma} \|\nabla_k (\mathcal{X}) \|_{\text{TL}_1^a,  \mathcal{L}},
\end{equation}
where  $a$ is the internal parameter in TL1 function,  $\gamma$ is the cardinality of set $\Gamma$,  and $\| \mathcal{G}\|_{\text{TL}_1^a, \mathcal{L}} = \sum_{i_3}\cdots \sum_{i_d} \| \mathcal{G}_{\mathcal{L}^{(i_3,\cdots,i_d)}\|_{\text{TL}_1^a}}$.
\end{definition}

\begin{remark}\label{TCTV}Since the TL1 function satisfies the property in \eqref{TL1:1}, when $a = \infty$, TATV$_a$ reduces to TATV \cite{cai2022approx, Liu2024c}, and TCTV$_a$ reduces to TCTV (or t-CTV) \cite{Wang2023a}. Moreover, because the TL1 function penalizes smaller values more heavily and drives them toward zero more aggressively than the $L_1$ function, TATV$_a$ promotes sparser solutions than TATV, while TCTV$_a$ encourages lower-rank solutions compared to TCTV.
\end{remark}


From the two sides above, we can find that the gradient tensors in the visual data possess both the approximate $\it{sparsity}$ and $\it{low\mbox{-}rankness}$ properties,  necessitating the joint modeling of the two properties in the gradient tensor.  On the other hand,  by imposing the maximum likelihood estimate (MLE) criterion on the observations $\{\mathcal{Y}_{w_t}\}_{t=1}^{n}$ in~(\ref{observ0}), the loss function is easily derived as:
\begin{equation}
	\label{loss}
	\Phi_{\mathcal{Y}}(\mathcal{X}):=  \frac{1}{n} \sum_{t=1}^{n} \ell(\mathcal{X}_{w_t}, \mathcal{Y}_{w_t}),
\end{equation}
where $\ell(\mathcal{X}_{w_t}, \mathcal{Y}_{w_t})= F(\mathcal{X}_{w_t}) - \mathcal{Y}_{w_t}\mathcal{X}_{w_t}$ and the specific form of $F(\cdot)$  can be found in Table~\ref{expdist}.  Based on the two tensor TV regularizers~\eqref{reg1} and \eqref{reg2}, the loss function~\eqref{loss} and with  box constraint $\mathbb{B}(\alpha)$ below, we can formulate a  tensor  dual-TV regularization model (TDTV$_a$) with $0<a \le \infty$  as follows:
\begin{eqnarray}
	\label{model1}
	\begin{aligned}
		&\hat{\mathcal{X}} = \arg\min_{\mathcal{X}} \Phi_{\mathcal{Y}}(\mathcal{X})   +\lambda_g \cdot \text{TCTV}_{a}(\mathcal{X})+ \lambda_{h}  \cdot\text{TATV}_{a} (\mathcal{X}),\\
		&s.t. ~\mathcal{X} \in \big\{ \mathcal{X} \in \mathbb{R}^{n_1 \times n_2\times \cdots \times n_d}: \|\mathcal{X}\|_{\infty} \leq \alpha  \big\} =:\mathbb{B}(\alpha),
	\end{aligned}
\end{eqnarray}
where $\lambda_g$ and $\lambda_h$ are two regularization parameters. 

\begin{remark}
The proposed tensor recovery model \eqref{model1} is quite extensive. For example, the proposed recovery model adapts exact tensor completion~\cite{Wang2023a} with t-CTV regularization to  the Gaussian noise case, and also includes  the existing Poisson tensor completion~\cite{Feng2024a} as the special case from the viewpoint of optimization. 
\end{remark}
\subsection{Upper Bound Analysis}
 For ease of the statement,  3rd-order tensors are taken as an example in the following to conduct a theoretical analysis for the upper bound on the recovery error of the estimator in model~\eqref{model1}. Some necessary assumptions are first presented as follows.
\begin{assumption}
	The function $x \mapsto F(x)$ is twice differentiable and strongly convex on $[-\alpha,\alpha]$, so that there exists constant $\underline{\sigma}_{\alpha}$, $\overline{\sigma}_{\alpha}>0$ satisfying:
	\begin{displaymath}
	\underline{\sigma}_{\alpha}^2 \leq F^{''}(x) \leq \overline{\sigma}_{\alpha}^2,
	\end{displaymath}
	for any $x \in  [-\alpha, \alpha]$.
\end{assumption}

For the sampling distribution, one needs to ensure that each entry has a sampling probability, which is lower bounded by a strictly positive constant, that is:
\begin{assumption}
	There exists a constant $\mu \geq 1$ such that, for any $n_1>0$, $n_2>0$, and $n_3>0$,
	\begin{equation}
		\min\limits_{(i,j,k)\in [n_1]\times[n_2]\times[n_3]} \pi_{i,j,k} \geq \frac{1}{\mu n_1n_2n_3}.
  \end{equation}
\end{assumption}
Denote by $R_{ik}=\sum_{j=1}^{n_2} \pi_{i,j,k}$ (resp. $C_{jk}=\sum_{i=1}^{n_1}\pi_{i,j,k}$)  the probability of sampling a coefficient from the  horizontal slice $i$ (resp. the lateral slice $j$). 
The following assumption requires that no horizontal slice nor lateral slice should be sampled far more frequently than the others.
\begin{assumption}
There exists a constant $ \nu \geq1$, such that
	\begin{equation}
		\max \limits_{i,j,k} (R_{ik},C_{jk}) \leq \frac{\nu }{mn_3}.
	\end{equation}
\end{assumption}
\begin{remark}
In the classical case of the uniform sampling, $\mu=\nu=1$ holds.
\end{remark}

For the observation distribution, it is assumed to be sub-exponential.
\begin{assumption}
There exist a constant $\delta_{\alpha}>0$ such that for all $x \in [-\alpha, \alpha]$ and $Y  \sim  f_{h,F}(y|x) $:
\begin{displaymath}
\mathbb{E} \left [ \exp \left ( |Y-F'(x)| /\delta_{\alpha} \right ) \right ] \leq e, 
\end{displaymath}
where $ f_{h,F}(y|x) $ is defined in~\eqref{exponential}.
\end{assumption}

In the following, we first present a useful definition to characterize our main result. Then,  Theorem~\ref{thm1} establishes an upper bound on the mean squared error (MSE) for our estimator  $\hat{\mathcal{X}}$ in~(\ref{model1}). .
\begin{definition}[inverse scaling  factor] The inverse scaling factor of $\boldsymbol{D}_{k}^{\dagger}$ is defined as 
\begin{equation*}
\rho_{k} := \max_{j \in [n_k] } \| \boldsymbol{D}_{n_k}^{\dagger}(j, :) \|_2~\text{for}~k=1,2,3, 
\end{equation*}
where  $\boldsymbol{D}_{n_k}^{\dagger}$ is Moore-Penrose inverse of  $\boldsymbol{D}_{n_k}$(Def.~\ref{gradient}).
\end{definition}

For the notational convenience, we define, 
\begin{align}
\Sigma_{\xi}:= \frac{1}{n} \sum_{o \in \Omega}  \xi_o \mathcal{E}_{o},~~\Xi_{\xi}:= \mathbb{E}\big[\max_k  \|  \Sigma_{\xi} \times_k (\boldsymbol{D}_{n_k}^{\dagger})^{T} \|_{\mathcal{L}}\big], \notag
\end{align}
where $\{\xi_o\}_{ o \in \Omega} $ is an independent Rademecher random variable series and $\mathcal{E}_{o}$ is a basis tensor whose value is 1 at the index $o$ and  0 otherwise. 
\begin{theorem}
\label{thm1}
Assume that H1, H2, H3 and H4 hold,  the true tensor $\mathcal{X}^{*}$ is  of transformed multi-rank $(r_1, r_2, r_3)$ and its gradient tensors in three modes have the sparsity of $s_1$, $s_2$ and $s_3$, respectively. Suppose that the sampling number $n$ and the regularization parameters $\lambda_g$, $\lambda_h$  satisfy: 
\begin{align}
&n> \nu^{-1} m n_3 \log\left( (n_1+n_2)n_3 \right)  \Big( \frac{\rho}{\max_{k}  \| \boldsymbol{D}_{n_k}^{\dagger}\|  } \Big)^2 C_m,  \label{n}\\
&\lambda_g \ge \frac{2}{3} \frac{a+\alpha \sqrt{N}}{1+a} \max_{k} \Big\| \nabla \Phi_{\mathcal{Y}} \left (\mathcal{ X}^* \right ) \times_k  (\boldsymbol{D}_{n_k}^{\dagger})^T \Big \|_{\mathcal{L}}, \\
&\lambda_h \ge C_h\Big(\frac{a+\alpha}{a+\alpha \sqrt{N}} + \frac{a+\alpha}{a}\Big) \lambda_g \sqrt{\frac{\widetilde{r}}{\ell}}, \label{lambda_h1}
\end{align}
where $\rho=\max_{k} \{ \rho_k\}$,   $N=n_1n_2n_3$, $\widetilde{r} =r_1 + \cdots +r_{n_3}$,  and
$$C_m=\max\Big\{  \frac{ 1 }{9}, ~2\big(\frac{\delta_{\alpha}}{\overline{\sigma}_{\alpha} }\big)^2 \log^2\big (\frac{2\rho \delta_{\alpha}\sqrt{\mu mn_3}}{\underline{\sigma}_{\alpha} } \big) \Big\},~~~C_h\ge 1.$$
Then,  the MSE of estimator $\hat{\mathcal{X}}$, with a large probability, has the following upper bound:
\begin{align}
\frac{\| \hat{\mathcal{X}} - \mathcal{X}^{*}\|_{F}^{2}}{n_1n_2n_3}   \precsim  \Xi_1(a)  \frac{ \mu^2 \widetilde{r} \max_{k} s_k}{\ell} + \Xi_2(a) \mu^2 \lambda_{h}^2 \max_{k} s_k^2 +\Xi_3(a)\frac{\mu^2 (n_1n_2 +n_2n_3 + n_1n_3) \log(2n_1n_2n_3)}{n}, \label{thm1:eq1} 
\end{align}
where $\Xi_1(a)$, $\Xi_2(a)$ and $\Xi_3(a)$ have the following form:
\begin{align}
&\Xi_1(a)=\frac{\lambda_g^2 }{\sigma_{\alpha}^{4}} \big( \frac{1+a}{a +\alpha \sqrt{N}}+ \frac{1+a}{a} \big)^2+ \alpha^2   \big( \frac{2a + \alpha \sqrt{N}}{a}\big)^2 \Xi_{\xi}^2, \notag \\
&\Xi_2(a)= \frac{1}{\sigma_{\alpha}^4} \big( \frac{1+a}{a}\big)^2+\alpha^2\big(\frac{a + \alpha \sqrt{N}}{a} \big)^2\frac{\Xi_{\xi}^2}{\lambda_g^2}, \notag \\
&\Xi_3(a)=\alpha^2 + \frac{\delta_{\alpha}^2}{\sigma_{\alpha}^4} + \alpha^2 \delta_{\alpha}^2\big(\frac{a + \alpha \sqrt{N}}{1+a} \big)^2 \frac{\Xi_{\xi}^2}{\lambda_g^2}.  \notag 
\end{align}
In particular, if $\lambda_g$,  $\lambda_h$ and $a$ are specified as:
\begin{align} 
&\lambda_g = \frac{2 c_{\alpha} \overline{\sigma}_{\alpha} }{3}  \frac{a+\alpha \sqrt{N}}{1+a} \max_{k} \| \boldsymbol{D}_{n_k}^{\dagger}\|  \sqrt{\frac{2\nu \ell \log((n_1+n_2)n_3}{nmn_3}}, \label{lambdag}\\
&\lambda_h = C_h \big (\frac{a+\alpha}{a+\alpha \sqrt{N}} + \frac{a+\alpha}{a}\big)\lambda_g \sqrt{\frac{\widetilde{r}}{\ell}}, \label{lambdah}\\
&a^{-1} = \mathcal{O}\big( (\alpha \sqrt{N})^{-1} \big),  \label{a_bound}
\end{align}
where  $c_{\alpha}$ depends on $\delta_{\alpha}$, then the upper bound reduces to 
\begin{align}
\frac{\| \hat{\mathcal{X}} - \mathcal{X}^{*}\|_{F}^{2}}{n_1n_2n_3}  \precsim  C_1  \cdot \frac{\mu^2 (n_1n_2 +n_2n_3 + n_1n_3) \log(2n_1n_2n_3)}{n} +C_2   \cdot \frac{\nu \mu^2  \left( \widetilde{r} \max_{k} s_k^2  \log((n_1+n_2)n_3)\right)}{n} \frac{ \max_k \| \boldsymbol{D}_{n_k}^{\dagger} \|^2 }{mn_3},  \label{thm1:eq2} 
\end{align}
where $C_1$ and $C_2$ are  coefficients in the following form:
\begin{align}
C_1 = \alpha^2+ \frac{\alpha^2\delta_{\alpha}^2}{c_{\alpha}^{2}\overline{\sigma}_{\alpha}^{2} } + \frac{\delta_{\alpha}^{2} }{\underline{\sigma}_{\alpha}^{4}},~C_2=\alpha^2 + \frac{c_{\alpha}^{2}\overline{\sigma}_{\alpha}^{2}}{\underline{\sigma}_{\alpha}^{4}}. \notag
\end{align}
\end{theorem}
\begin{proof}
The proof can be found in Appendix B.
\end{proof}
\begin{remark}When $a \to \infty$, the upper bounds in \eqref{thm1:eq1} and \eqref{thm1:eq2} correspond to the model \eqref{model1} with convex regularization, i.e.,
\[
\text{TDTV}_{\infty}(\mathcal{X}) = \lambda_g \cdot \text{TCTV}(\mathcal{X}) + \lambda_h \cdot \text{TATV}(\mathcal{X}).
\]
When $0 < a < \infty$, the upper bounds in \eqref{thm1:eq1}   apply to model \eqref{model1} under nonconvex regularization induced by the TL1 function. This nonconvex regularization is a key contribution of this work.
In this regard, we note that several recent works \cite{Qiu2021b,Wang2023a,Feng2024a,Song2023a,Hou2021a,Cao2024a,Zhang2022a,Hou2024a} study recovery guarantees for convex regularization or constrained models, and their results and proof techniques are not directly applicable to the nonconvex setting considered in this work. It is noted that for $0 < a < \infty$, the values of two regularization parameters $\lambda_g$ and $\lambda_h$ are related to the  internal parameter $a$.  In addition, when $a^{-1} = \mathcal{O}\big( (\alpha \sqrt{N})^{-1} \big)$,  it is proved that  the upper bound with nonconvex regularization is consistent to that with convex regularization ($a \to \infty$); see \eqref{thm1:eq2}.  However, when  $a^{-1} \neq \mathcal{O}\big( (\alpha \sqrt{N})^{-1} \big)$,  the upper bound in  \eqref{thm1:eq1}  becomes increasingly loose  as $a$ approaches zero, revealing a limitation of our current result that we leave for future work.
\end{remark}

\begin{remark}Theorem~\ref{thm1} provides a general upper bound for the recovery error which can be further simplified under certain conditions. In the following, we show that  the second term of the upper bound in \eqref{thm1:eq2}  dominates this bound. Recall that $\widetilde{r} = r_1 + \cdots + r_{n_3} \le n_3 r_t$. If we further assume that $r_t$ and $s_k$ satisfy
\begin{equation}\label{rs}
\frac{\max\{ n_1 n_2, \, n_1 n_3, \, n_2 n_3 \}}{n_3} \precsim r_t \big( \max_{k} s_{k}^{2} \big),
\end{equation}
then the upper bound in \eqref{thm1:eq2} reduces to
\begin{align}
\frac{\| \hat{\mathcal{X}} - \mathcal{X}^{*}\|_{F}^{2}}{n_1 n_2 n_3}
\precsim \max(C_1, C_2) \cdot \frac{ \max_k \| \boldsymbol{D}_{n_k}^{\dagger} \|^2 }{m n_3}\cdot \frac{\nu \mu^2 \left( n_3 r_t \max_{k} s_k^2 \, \log\big((n_1+n_2)n_3\big)\right)}{n}.
\label{coro:eq1}
\end{align}
It is worth noting that the sparsity $s_k$ $(k=1,2,3)$ is defined with respect to the signal length $N = n_1 n_2 n_3$, rather than the mode dimension $n_k$. Therefore, the condition in \eqref{rs} is practically reasonable. 
For example, consider a tensor $\mathcal{X} \in \mathbb{R}^{100 \times 100 \times 100}$ with transformed tubal-rank $r_t = 5$. If $1\%$ of the entries in the gradient tensors are nonzero, then $s_k = 10{,}000$. In this case, $r_t (\max_k s_k^2)$ is significantly larger than 
$\frac{\max\{ n_1 n_2, \, n_1 n_3, \, n_2 n_3 \}}{n_3} = 100$,
which illustrates that the condition \eqref{rs} can be satisfied in practical scenarios.
\end{remark}
\begin{remark}When $\frac{ \max_k \| \boldsymbol{D}_{n_k}^{\dagger} \|^2 }{m n_3} \asymp 1$, for instance when the target tensor is approximately cubic, i.e., $n_1 \approx n_2 \approx n_3$, the upper bound in \eqref{coro:eq1} reduces to
\begin{equation}
\label{coro:eq2}
\frac{\| \hat{\mathcal{X}} - \mathcal{X}^{*}\|_{F}^{2}}{n_1 n_2 n_3}
\precsim C_3 \frac{\nu \mu^2 \, n_3 r_t (\max_{k} s_k^2)\, \log\!\left((n_1+n_2)n_3\right)}{n},
\end{equation}
where $C_3 = \max(C_1, C_2)$ depends on $\alpha$.
\end{remark}
\begin{remark}In the proof of Theorem \ref{thm1}, we impose an assumption on the relationship between $\lambda_g$ and $\lambda_h$; see \eqref{lambda_h1} and its special case \eqref{lambdah}. Although this assumption may be somewhat strong, it defines a range of values for $\lambda_g$ and $\lambda_h$ that ensure recoverability, and in principle, these parameters can be set at the lower bound of this range to determine their order of magnitude. In practice, however,  since the theoretical expressions involve unknown constants, the values of $\lambda_g$ and $\lambda_h$ are selected empirically via grid search to achieve empirically optimal  performance.
\end{remark}


\subsection{Lower Bound Analysis}
We provide a lower bound analysis in the following, which shows that the upper bound in Theorem~\ref{thm1}, under the considered special case, can approximate the lower bound with the gap of  order $\mathcal{O}(\max_k s_k^2/M)$ up to a logarithmic factor when $\alpha$ is treated as a constant. Before approaching the result, let us first introduce the set $\mathbb{K}(r, \alpha)$ of tensors whose tubal-rank  is at most $r$:
$$ \mathbb{K}(r, \alpha) = \left\{ \mathcal{X}' \in \mathbb{R}^{n_1 \times n_2 \times n_3}: ~\text{rank}_{t}(\mathcal{X}')\le r, ~\| \mathcal{X}' \|_{\infty} \le \alpha  \right \}.$$
The infimum over all estimators $\hat{\mathcal{X}}$ that are measurable functions of the data $(w_i, \mathcal{Y}_{w_i} )_{i=1}^{n}$ is denoted as $\underset{ \hat{\mathcal{X}}}{\inf}$.
\begin{theorem}
\label{thm2}
For all $\alpha>0$,  $\iota \in(0, \frac{1}{8})$ and  $1\le r \le m$, there exists two constant $\tilde{c}>0$ and $\theta_{\alpha,r} >0$ such that,
\begin{equation}
\label{lowerbd}
\underset{\mathcal{\hat{X}}}{\inf} \underset{\dot{\mathcal{X}} \in\mathbb{K} (r,  \alpha) }{\sup}\mathbb{P}\left( \frac{\| \mathcal{\hat{X}} - \dot{\mathcal{X}}\|_{F}^{2} }{n_1n_2n_3} > \tilde{c} \min \left \{ \alpha^2, \frac{\iota r_tn_3M }{n\overline{\sigma}_{\alpha}^{2}} \right \}\right) \ge \theta_{\alpha,r} \notag
\end{equation}
where 
$$\theta_{\alpha,r} = \frac{1}{1+2^{- \frac{rn_3M}{32}}} \left(1- 2\iota  -4 \sqrt{\frac{\iota }{rn_3M}}\right).$$
\end{theorem}
\begin{proof} 
The proof can be found in Appendix C.
\end{proof}
\begin{remark}
Theorem~\ref{thm2} provides a lower bound of order $\mathcal{O}(\frac{r_tn_3M}{n\overline{\sigma}_{\alpha}^{2} })$. 
The order of the ratio between this lower bound and the upper bound in  \eqref{coro:eq2} is $ (\max_k s_k^2)\log((n_1+n_2)n_3) / (M\overline{\sigma}_{\alpha}^{2})$. So, it can be concluded that  our upper bound in \eqref{coro:eq2} approaches the lower bound with the gap of  order $\mathcal{O}(\max_k s_k^2/M)$ up to a logarithmic factor. 
\end{remark}

\subsection{Discussion on Methodologies and Proof Techniques}\label{discuss}
From the viewpoint of methodologies,  our work is closely related to some recent works \cite{Qiu2021b}, 
\cite{Wang2023a}, \cite{Feng2024a},  \cite{Liu2024b},  \cite{Liu2024c},  \cite{Hou2024a},  \cite{KaiHuang2024}. All these methods consider modeling local smoothness in visual tensors through tensor TV. However,  both the modeling details and the considered problems are  quite different. Specifically,  in \cite{Qiu2021b},  a composite regularization of  transformed tensor nuclear norm and TATV, i.e., 
$\| \mathcal{X} \|_{\text{TTNN}} +\lambda \cdot \mathrm{TATV}( \mathcal{X})$, is used to characterize the low-rankness and local-smoothness of a targeted tensor. In \cite{Wang2023a} and \cite{Feng2024a}, 
TCTV or its nonconvex variants are used to tackle the tensor completion problem, but these works do not consider more comprehensive exponential-family noise. In \cite{Liu2024b},   \cite{Hou2024a},  \cite{KaiHuang2024},  TCTV or its nonconvex variants are exploited to handle tensor compressive sensing or image denoising rather than tensor completion.  In \cite{Liu2024c},  for the spectral computed tomography problem,
the 3rd-order tensor $\mathcal{X}$ is decomposed into a low-rank background component  $\mathcal{B}$ and a sparse movement component $\mathcal{S}$, i.e., $\mathcal{X} = \mathcal{B} + \mathcal{S}$;  then the weighted TCTV (WTCTV) and weighted TATV (WTATV)  are imposed on $\mathcal{B}$ and $\mathcal{S}$ seperately, i.e.,
$\lambda_B\mathrm{WTCTV}(\mathcal{B})+ \lambda_S  \cdot \mathrm{WTATV}(\mathcal{S})$.
This formulation is completely different from our TDTV model, which imposes the joint regularization directly on $\mathcal{X}$.

As for the proof techniques,  the \textit{cover-number} of a hypothesis space of the solution in \cite{cai2022approx} is analyzed to establish an approximation property for TV minimization from incomplete observations.  From the viewpoint of optimization,  let us consider a variant of model~\eqref{model1} as follows:
\begin{eqnarray}\label{model1b}
\begin{aligned}
		&\min_{\mathcal{X}}   \gamma \cdot \text{TCTV}_{a}(\mathcal{X})+  \text{TATV}_{a} (\mathcal{X}),\\
		&s.t. ~~\Phi_{\mathcal{Y}}(\mathcal{X})  \le \eta,~ \mathcal{X} \in \mathbb{B}(\alpha).
	\end{aligned}
\end{eqnarray}
The hypothesis space  of the above model can be defined as:
\begin{equation}
\begin{split}
 \mathcal{M}:= \left\{ \mathcal{X} \in \mathbb{R}^{n_1 \times n_3 \times n_3}:  \text{TDTV}_{a}(\mathcal{X})  \le  \text{TDTV}_{a}(\mathcal{X}^{*}),~ \Phi_{\mathcal{Y}}(\mathcal{X})  \le \eta,~ \mathcal{X} \in \mathbb{B}(\alpha) \right \}, \nonumber
\end{split}
\end{equation}
where $\text{TDTV}_{a}(\mathcal{X}) = \gamma \cdot \text{TCTV}_{a}(\mathcal{X})+  \text{TATV}_{a} (\mathcal{X})$.  When $\gamma =0$ and $\Phi_{\mathcal{Y}}(\mathcal{X})$ is taken to be the least square loss,   the hypothesis space $\mathcal{M}$  is the same as that  considered in  \cite{cai2022approx} and its cover-number is also computed in  \cite{cai2022approx}. However, for the general $\mathcal{M}$ above, computing its cover-number is not straightforward. Different  from \cite{cai2022approx}, we consider the optimization model~\eqref{model1} instead of model~\eqref{model1b}, and establish an upper bound for the recovery error using tools from high-dimensional statistics \cite{buhlmann2011statistics}.
The \textit{incoherence} conditions on gradient tensors in \cite{Wang2023a}
are used to provide a theoretical  guarantee for exact tensor recovery. However, it cannot be easily extended as a tool for establishing a recovery theory for Exponential-family tensor completion. In \cite{Feng2024a},  the  \textit{incoherence} conditions are abandoned and instead a constraint on the regularization parameter is enforced to generate an upper bound for Poisson tensor completion. The proof approach of \cite{Feng2024a} is similar to that of \cite{Zhang2022a}, except for the use of the equivalence relation between the tensor Frobenius norm and the TCTV seminorm, i.e., 
$\mathrm{TCTV}(\mathcal{X}^{*})- \mathrm{TCTV}(\hat{\mathcal{X}} )\leq \frac{1}{3} \sum_{k=1}^{3} \sqrt{\frac{2 \widetilde{r} s_k}{n_1n_2}} \|\mathcal{X}^{*} - \hat{\mathcal{X}} \|_{F}$.  
 This property is specific to the DFT and does not generally hold for arbitrary inverse transforms.
In contrast, our framework is formulated under a general inverse transform $\mathcal{L}$ and is not restricted to the DFT setting. Moreover, our analysis concerns a more challenging setting by incorporating nonconvex regularization via the TL1 function and a more general exponential-family loss that encompasses Poisson loss. Consequently, our proof techniques differ from the existing works \cite{cai2022approx},\cite{Wang2023a}, \cite{Zhang2022a} and \cite{Feng2024a}. The main techniques of our work arise from  \cite{Jean2015} and \cite{huetter16}. In addition,  to handle the TV seminorm, we apply an important property of TV in   \cite{huetter16} and then we derive \textit{a key inequality} (see Lemma 3 in Appendix A),
which will be invoked many times in our proof. 
\section{Optimization Algorithm}
\subsection{Algorithm Framework}
An efficient solver will be developed for model~\eqref{model1}  resorting to the ADMM algorithm~\cite{Boyd2011a}. By introducing several splitting variables,  we have an equivalent formulation of~\eqref{model1} below:
\begin{equation}
\label{model2}
\begin{aligned}
&\underset{\mathcal{Z} \in \mathbb{B}(\alpha)}{\min} \Phi _{\mathcal{Y}}(\mathcal{Z}) +\sum_k \big( \lambda_g \| \mathcal{G}_k\|_{\text{TL}_1^a, \mathcal{L}} + \lambda_{h} \| \mathcal{H}_k\|_{\text{T}\ell_1^a}  \big) \\
&s.t.~ \mathcal{Z} =\mathcal{X}, \mathcal{G}_k = \nabla_k (\mathcal{X}), ~\mathcal{H}_k = \nabla_k (\mathcal{X}).\\  
\end{aligned}
\end{equation}
Intuitively, the advantage of converting the problem~(\ref{model1}) into the equivalent form~(\ref{model2}) is that we split the hard composite optimization on $\mathcal{X}$ into an easy optimization with respect to four block tensor-variables. 

To solve the optimization problem~(\ref{model2}), we consider its augmented Lagrangian function defined by
\begin{displaymath}
\begin{aligned}
&\mathcal{L}\big(\mathcal{X}, \mathcal{Z}, \{\mathcal{G}_k\}, \{\mathcal{H}_k\}, \Lambda_{\mathcal{Z}},  \{\Lambda_{\mathcal{G}_k}\},  \{\Lambda_{\mathcal{H}_k}\} \big) := \Phi_{\mathcal{Y}} (\mathcal{Z}) + \sum_k  \big( \lambda_g \| \mathcal{G}_k\|_{\text{TL}_1^a, \mathcal{L}} + \lambda_{h} \| \mathcal{H}_k\|_{\text{T}\ell_1^a} \big)+ \langle\Lambda_{\mathcal{Z}}, \mathcal{Z} - \mathcal{X}   \rangle   + \frac{\beta}{2}\| \mathcal{Z} - \mathcal{X}   \|_{F}^2 \\
&+\sum_{k} \Big( \langle\Lambda_{\mathcal{G}_k}, \mathcal{G}_k - \nabla_k ( \mathcal{X} )  \rangle   + \frac{\beta}{2}\|  \mathcal{G}_k - \nabla_k ( \mathcal{X} )  \|_{F}^2  + \langle\Lambda_{\mathcal{H}_k}, \mathcal{H}_k - \nabla_k ( \mathcal{X} )  \rangle  + \frac{\beta}{2}\|  \mathcal{H}_k - \nabla_k ( \mathcal{X} )  \|_{F}^2\Big),
\end{aligned}
\end{displaymath}
where  $\Lambda_{\mathcal{Z}}$,  $\{\Lambda_{\mathcal{G}_k}\}$ and $\{\Lambda_{\mathcal{H}_k}\}$ are the multiplier variables. Then, ADMM is used to solve the above optimization problem. The algorithm runs iteratively, and at the $t$-th iteration, we can update 
$\big(\mathcal{X}, \mathcal{Z},  \{\mathcal{G}_k\}, \{\mathcal{H}_k\}, \Lambda_{\mathcal{Z}}, \{\Lambda_{\mathcal{G}_k}\},  \{\Lambda_{\mathcal{H}_k}\} \big)$  by solving the following subproblems:
\begin{eqnarray}
	\label{subproblem}
	\mathcal{X}^{t+1} &=& \underset{\mathcal{X}}{\arg\min}  \mathcal{L}\big(\mathcal{X}, \mathcal{Z}^t, \{\mathcal{G}_k^{t}\}, \cdots, \{\Lambda_{\mathcal{G}_k}^t\},  \{\Lambda_{\mathcal{H}_k}^t\} \big)  \nonumber \\
		\mathcal{Z}^{t+1} &=& \underset{\mathcal{Z} \in\mathbb{B}(\alpha)}{\arg\min}  \mathcal{L}\big(\mathcal{X}^{t+1}, \mathcal{Z}, \{\mathcal{G}_k^{t}\},\cdots,  \{\Lambda_{\mathcal{G}_k}^t\},  \{\Lambda_{\mathcal{H}_k}^t\} \big)  \nonumber \\
	\mathcal{G}_k^{t+1} &=& \underset{\mathcal{G}_k}{\arg\min}  \mathcal{L}\big(\mathcal{X}^{t+1}, \mathcal{Z}^{t+1}, \{\mathcal{G}_k\},\cdots,  \{\Lambda_{\mathcal{G}_k}^t\},  \{\Lambda_{\mathcal{H}_k}^t\} \big)  \nonumber \\
	\mathcal{H}_k^{t+1} &=& \underset{\mathcal{H}_k}{\arg\min}  \mathcal{L}\big(\mathcal{X}^{t+1}, \cdots, \{\mathcal{G}_k^{t+1}\}, \{\mathcal{H}_k\}, \cdots,  \{\Lambda_{\mathcal{H}_k}^t\} \big)  \nonumber   \\
	\Lambda_{\mathcal{Z}}^{t+1} &=& \Lambda_{\mathcal{Z}}^t + \beta \cdot\big(  \mathcal{Z}^{t+1} - \mathcal{X}^{t+1}   \big)   \label{lamZ} \\
	\Lambda_{\mathcal{G}_k}^{t+1} &=& \Lambda_{\mathcal{G}_k}^{t} +\beta\cdot \big(   \mathcal{G}_k^{t+1} - \nabla_k \big ( \mathcal{X}^{t+1} \big) \big) \\
	\Lambda_{\mathcal{H}_k}^{t+1} &=& \Lambda_{\mathcal{H}_k}^{t} +  \beta \cdot \big(   \mathcal{H}_k^{t+1} - \nabla_k \big( \mathcal{X}^{t+1} \big) \big), \label{lamH} 
	\end{eqnarray}
where $\beta$ is the step-length parameter for multiplier variables.

\subsection{Solving Subproblems}
\subsubsection{Updating $\mathcal{X}^{t+1}$} For $\mathcal{X}$-subpoblem, it is not hard to find that seeking its solution amounts to solving the following linear system:
\begin{equation*}
\big( I + 2 \cdot  \sum_k \nabla_k^T \nabla_k   \big)(\mathcal{X})=\big (\mathcal{Z}^t + \frac{\Lambda_{\mathcal{Z}}^t}{\beta} \big) + \sum_k\nabla_k^T ( \mathcal{V}_k^{t}),
\end{equation*}
where $I$ represents the identify operator, $\nabla_k^T$ denotes the transpose of the $k$-mode gradient operator and $\mathcal{V}_k^{t}  =\frac{\Lambda_{\mathcal{G}_k}^t }{\beta} + \mathcal{G}_k^t +  \frac{\Lambda_{\mathcal{H}_k}^t }{\beta} + \mathcal{H}_k^t$,
Fortunately, the large linear-system above can be fast solved in the frequency domain through the multi-dimensional fast Fourier transform (FFT), avoiding the expensive computation to invert a large matrix in the linear system above. Specifically, denote by $\mathcal{D}_k$ the $k$-mode first-order difference tensor with respect to $\nabla_k$.  Then, the solution of the  large linear system above can be gotten as follows.
\begin{equation}
\label{x_problem}
\mathcal{X}^{t+1} = \mathcal{F}^{-1}\Big(\frac{ \mathcal{F}\big(\big (\mathcal{Z}^t + \frac{\Lambda_{\mathcal{Z}}^t}{\beta} \big) +  \sum_k \mathcal{U}_k^t \big) }{  \mathbf{1}+ 2 \cdot \sum_{k}  \mathcal{F}(\mathcal{D}_k)^* \odot \mathcal{F}(\mathcal{D}_k) }  \Big),
\end{equation}
where $ \mathcal{U}_k^t= \mathcal{F}(\mathcal{D}_k)^*\odot \mathcal{F}(\frac{\Lambda_{\mathcal{G}_k}^t }{\beta} + \mathcal{G}_k^t + \frac{\Lambda_{\mathcal{H}_k}^t }{\beta} + \mathcal{H}_k^t )$,  $\mathbf{1}$ is a tensor with all entries being 1, $\odot$ is component-wise multiplication, and the division is component-wise as well.  

\subsubsection{Updating $\mathcal{Z}^{t+1}$} The $\mathcal{Z}$-subproblem can be easily formulated as:
\begin{equation*}
\begin{split}
&\min_{\mathcal{Z} \in \mathbb{B}(\alpha)} \Phi_{\mathcal{Y}}(\mathcal{Z}) + \frac{\beta}{2}\| \mathcal{Z} - ( \mathcal{X}^{t+1} - \frac{\Lambda_{\mathcal{Z}}^{t}}{\beta}) \|_{F}^{2}  
\end{split}
\end{equation*}
Denote by $\mathrm{ProjNewton}_{\mathbb{B}(\alpha)}(\cdot)$ the projected Newton iterative procedure. The optimal solution of this model
can be sought by the following procedure:
\begin{equation}
\label{z_problem}
\mathcal{Z}^{t+1} = \mathrm{ProjNewton}_{\mathbb{B}(\alpha)} \big(  \mathcal{X}^{t+1} - \frac{\Lambda_{\mathcal{Z}}^{t}}{\beta} \big).
\end{equation}
More precisely, denote  $\Pi_{\mathbb{B}(\alpha)}(x)= \min(\alpha, |x|)\frac{x}{|x|}$ and $\Pi_{\mathbb{B}(\alpha)}(\mathcal{X})$ imposes the element-wise projection on $\mathcal{X}$ with the radius $\alpha$. By iterating the following equation until reaching the termination conditions, 
$$
\mathcal{Z}_{i+1}^{t}  \gets  \Pi_{\mathbb{D}(\alpha)} \big(\mathcal{Z}_{i}^{t}  - \mathbf{H}_{\mathcal{Z}_{i}^{t}}^{-1} \mathbf{g}_{\mathcal{Z}_{i}^{t}} \big),
$$
we can attain the optimal solution $\mathcal{Z}^{t+1}$.  Here, $\mathbf{H}_{\mathcal{Z}_{i}^{t}}$ stands for Hessian matrix of the objective function at the iteration $\mathcal{Z}_{i}^{t}$, and $\mathbf{g}_{\mathcal{Z}_{i}^{t}}$ represents the gradient of the objection function. It is noted that our Hessian matrix  $\mathbf{H}_{\mathcal{Z}_{i}^{t}}$ is a diagonal matrix and its inverse is easy to compute.  It is noted that for some special cases like the least squared loss or Poisson loss, the $Z$-subproblem has the closed-form solution and thus the projected Newton iterative procedure is not required.
\subsubsection{Updating $\mathcal{G}_k^{t+1}$} It is not hard to express the $\mathcal{G}_{k}$ subproblem as the following formulation:
\begin{displaymath}
\min_{\mathcal{G}_k} \frac{\lambda_g}{\beta}\| \mathcal{G}_k\|_{\text{TL}_1^a, \mathcal{L}} + \frac{1}{2}\| \mathcal{G}_k - \big (\nabla_k ( \mathcal{X}^{t+1} )  - \frac{\Lambda_{\mathcal{G}_k}^{t}}{\beta} \big)  \|_{F}^{2}.
\end{displaymath}
The optimal solution of this model can be obtained by tensor~singular~value~thresholding~(t-SVT$^a$)  with respect to  $\| \cdot\|_{\text{TL}_1^a, \mathcal{L}}$.  Specifically,  the solution is given by, 
\begin{equation}
\label{G_subproblem}
\mathcal{G}_{k}^{t+1} = \text{t-SVT}_{\lambda_g/\beta}^a \big(\nabla_k ( \mathcal{X}^{t+1} )  - \Lambda_{\mathcal{G}_k}^{t}/\beta\big),
\end{equation}
where details of $\text{t-SVT}_{\lambda_g/\beta}^a(\cdot)$ are provided in Appendix D.

\subsubsection{Updating $\mathcal{H}_k^{t+1}$} We can easily express the $\mathcal{H}_{k}$ subproblem as the following formulation:
\begin{displaymath}
\min_{\mathcal{H}_k} \frac{\lambda_h}{\beta} \| \mathcal{H}_k\|_{\text{T}\ell_1^a}+ \frac{1}{2}\| \mathcal{H}_k - \big (\nabla_k( \mathcal{X}^{t+1} )  - \frac{\Lambda_{\mathcal{G}_k}^{t}}{\beta} \big)  \|_{F}^{2}.
\end{displaymath}
Its optimal solution can be attained by the shrinkage operator corresponding to tensor $\text{T}\ell_1$ norm, i.e., 
\begin{equation}
\label{H_subproblem}
\mathcal{H}_{k}^{t+1} = \mathrm{Shrink}_{\lambda_h/\beta}^{a} \big(\nabla_k ( \mathcal{X}^{t+1} )  - \Lambda_{\mathcal{H}_k}^{t}/\beta\big),
\end{equation}
where details of $\mathrm{Shrink}_{\lambda_h/\beta}^{a}(\cdot)$ are provided in Appendix D.

\begin{algorithm}[!htbp]
\caption{Alternating Direction Method of Multipliers for Solving Optimization Problem~(\ref{model2})}
	Initialize $\mathcal{X}=rand([n_1,\cdots,n_d])$, $\mathcal{X}(\Omega)=\mathcal{Y}$;
	 $\mathcal{Z}=\mathcal{X}$,   $\mathcal{G}_k = \nabla_k (\mathcal{X}), ~\mathcal{H}_k = \nabla_k(\mathcal{X})$, $\Lambda_{\mathcal{Z}}=0$,   $\Lambda_{\mathcal{G}_k}=0$,  $\Lambda_{\mathcal{H}_k}=0$, $\beta=\text{1e-4}$, $\varrho=1.07$, $t=0$\; 
	\KwIn{$\mathcal{Y}$, $\Omega$, $\lambda_g$, $\lambda_h$;}
	\While{ relChgX (\ref{relchgX}) $ > 10^{-8}$ }{
		Update $\mathcal{X}^{t+1}$ by Eq.~\eqref{x_problem}\;
		Update $\mathcal{Z}^{t+1}$ by Eq.~\eqref{z_problem}\;
		Update $\mathcal{G}_{k}^{t+1}~(k=1,2,\cdots,d)$ by Eq.~\eqref{G_subproblem}\;
		Update $\mathcal{H}_{k}^{t+1}~(k=1,2,\cdots,d)$ by Eq.~\eqref{H_subproblem}\;
		Update $\Lambda_{\mathcal{Z}}^{t+1}$, $\Lambda_{\mathcal{G}_{k}}^{t+1} $, $\Lambda_{\mathcal{H}_k}^{t+1} $ by Eqs.~\eqref{lamZ}-\eqref{lamH} \;
		Update $\beta \gets \min(\varrho\cdot\beta,  1e5)$\;
		Update $t \gets t + 1$\;
	}
	\KwOut{$\mathcal{X}^{t_*+1}$}
\label{alg1}
\end{algorithm}

For the step-length parameter $\beta$ in the ADMM algorithm, we exploit the dynamic updating scheme, i.e. $\beta = \min\{ \varrho \cdot \beta,1e5 \}$ with $\varrho$ being taken to be 1.07. We now summarize the whole procedure for solving the problem~(\ref{model2}) in Algorithm~\ref{alg1}.
\subsection{Convergence Analysis}
Let us denote $\widetilde{\mathcal{Z}} :=\big( \mathcal{Z}, \{ \mathcal{G}_k\}, \{ \mathcal{H}_k\} \big)$ and $\mathcal{A}:= - \mathrm{diag}\big(I, \{\nabla_k \}, \{\nabla_k\}\big)$. The constraints in model~\eqref{model2}  can be equivalently rewritten as:
\begin{equation*}
\mathcal{A}(\mathcal{X}) + \widetilde{\mathcal{Z}} =0. 
\end{equation*}
Due to the separability of the linear constraints in our model~\eqref{model2}, we can equivalently convert our algorithm into the standard  ADMM form with two blocks  $\mathcal{X}$ and $\tilde{\mathcal{Z}}$. Thus,  the convergence of two-block ADMM proved in \cite{Boyd2011a} can be directly applied to our algorithm. 

\section{Numerical  Experiments}
In this section, we first conduct experiments on synthetic tensors with both low rankness and sparsity priors to verify our main theoretical results. Then, more experiments on real visual tensors are carried out to substantiate the effectiveness of the proposed method. All experiments are executed on a desktop computer with Intel Core i9-12900K (3.20GHz) and 64GB RAM.

\subsection{Experimental Settings} We will test the three specific noise in the exponential-family, i.e., Gaussian noise, Poisson noise and one-bit noise.  For a given ground-truth tensor $\mathcal{X} \in \mathbb{R}^{n_1 \times n_2 \times n_3}$, we can obtain its partial measurements contaminated by Gaussian noise, Poisson noise and one-bit noise, conforming to the following observation equations: 
\begin{align}
& \mathcal{Y}(w_t) = \mathcal{X}(w_t) + \mathcal{E}(w_t),  \label{observ1} \\
&\mathbb{P} \big(\mathcal{Y}(w_t) =k | \mathcal{X}(w_t) \big) = \frac{(\mathcal{X}(w_t)+c)^k}{k!} e^{-\big(\mathcal{X}(w_t )+c\big)}, ~k \in \mathbb{N} \label{observ2}\\
&\mathcal{Y}(w_t) =\begin{cases}
		1,  & \mbox{with~prob.}~\phi(\mathcal{X}(w_t))\\
		-1, & \mbox{with~prob.}~1 - \phi(\mathcal{X}(w_t))\\
		\end{cases}, \label{observ3}
\end{align}
where the sampling-index set $\Omega=\{w_1,w_2,\cdots,w_n\}$ is obtained by the uniform-sampling scheme on the set $[n_1] \times [n_2] \times [n_3]$,  all $\mathcal{E}(w_t)$ follow Gaussian-white-noise distribution of mean zero and the standard deviation $\sigma_{Gauss}$, and  $c$ is a pre-specified constant (e.g., $c =1$) avoiding the occurrence of zero or a small positive number in the logarithmic function. For one-bit measurements \eqref{observ3},  the function $\phi(x)$~\cite{Cai2013a} can be specified as $\Phi(\frac{x-\mu}{\sigma_{\phi}})$, where $\Phi$ denotes the cumulative distribution function of the standard Gaussian,  $\mu$ is the mean value and $\sigma_{\phi}$ is the standard deviation. The parameter $\mu $ is chosen to
be $0.5$ to center the normalized data (scaled to $[0,1]$) around zero. 
Parameter $\sigma_{\phi}$ is determined empirically, which will be discussed in the subsection~\ref{sensitive}.  The sampling ratio is computed as $SR =\text{card}(\Omega)/(n_1n_2n_3)$. 

The performance of the proposed TDTV$_a$ approach is influenced by the selection of transformation\cite{Kernfeld2015, Song2023a,Feng2024a} in the transformed t-SVD. Here, the Fast Fourier Transform (FFT)   is exploited in our experiments.
In the synthetic data experiments, the Mean Square Error (MSE)  and the Signal-to-Noise Ratio (SNR) are used as evaluation metrics. For visual tensor data, the Mean Peak~Signal\text{-}to\text{-}Noise Ratio (MPSNR),  and the Mean Structural~SIMilarity (MSSIM) along with the third or more dimensions are employed to evaluate the recovery performance, all of which tend to have  better performance when attaining  larger value.

\subsection{Synthetic Data}
We use a simple procedure to synthesize a tensor with transformed tubal-rank $r_t$ and intrinsic smooth structures. 
For a $d$-order tensor $\mathcal{X}^{*}$ of size $n_1 \times n_2 \times \cdots \times n_d$,  we generate its each frontal slices $\boldsymbol{X}^{j} \in \mathbb{R}^{n_1 \times n_2}, j=1,\cdots, n_3 \cdots n_d$ by the following pipelines:
\begin{itemize}
\item We randomly select $s_t~(s_t<n_1n_2)$  initial points, and the remaining points are allocated into these $s_t$ points by the nearest neighbor principle, which means that all these slices are divided into $s_t$ regions. 
\item Each region is assigned the same vector whose values are independently sampled from the uniform distribution $U(0,1)$, yielding the intermediate tensor $\mathcal{Z}$.
\item Then, we produce a  tensor $\mathcal{X}^{*}$ of low tubal-rank $r_t$ by imposing a transformed tubal-rank approximation (FFT) on the  intermediate  tensor $\mathcal{Z}$.
\item Finally, we output a scaled tensor $\mathcal{X}^{*}$ whose values lie in the interval $[0, \alpha]$.
\end{itemize}
Since $t$-SVD is orientation-dependent, this makes the tensor $\mathcal{X}^{*}$ possess  well smooth property in the first two directions. Such property imitates the spatial smoothness of visual data.
\subsubsection{Empirical convergence}\label{converg1}
Using the above procedure, we generate a tensor of size $70 \times 70 \times 70$ with the following parameter settings: the tubal-rank $r_t = 5$, the number of segmentation regions $s_t = 10$, and the normalized FFT transform. For Gaussian noise, the standard deviation is set to $\sigma_{\mathrm{Gauss}} = 0.03$ and the scale parameter to $\alpha = 1$; for Poisson noise, $c = 1$ and $\alpha = 100$; for one-bit noise, $\alpha = 1$. The sampling rate is fixed at $SR = 0.4$. The parameter $a$ in the TL1 function is set to $0.5$, and the parameter $\sigma_{\phi}$ in the function $\phi(x)$ is set to $0.14$, both of which are empirically selected via sensitivity analysis (see \ref{sensitive}). The relative change on variable $\mathcal{X}$ is selected as the termination rule:
\begin{equation}
\label{relchgX}
relChgX= \| \mathcal{X}^{t+1} - \mathcal{X}^{t}\|_F / max(1, \| \mathcal{X}^{t}\|_{F}).
\end{equation}
In the proposed model~\eqref{model1},  we exert the proposed TDTV$_{0.5}$ regularization only on the  mode-1 and mode-2 directions, which well characterizes local smoothness on the frontal slices of the synthetic tensor $\mathcal{X}^{*}$. 

In Fig.~\ref{converg},  we show the convergence curves of the proposed algorithms in terms of indices MSE and relChgX.  From subfigure (a), it can be seen that relChgX  gradually declines  as the iteration  runs,  and the algorithm is terminated when $\mathrm{relChgX}$ reaches the prespecified tolerance of $10^{-8}$.   Meanwhile, the subfigure (b) indicates that the MSE curves of three kinds of tensor completions attain stable values when the iterations arrive at suitably large values,  demonstrating the empirical convergence of our method.
\subsubsection{Parameter sensitive analysis}\label{sensitive}
For the three types of tensor completion, $a$ is the common tuning parameter, while for one-bit tensor completion, $\sigma_{\phi}$ serves as an additional tuning parameter. Both parameters are selected via grid search to minimize the MSE over the candidate values. Specifically, the grid for $a$ is $\{0.01, 0.05, 0.1, 0.5, 1, 10, 20, 50, 100, 200, 500\}$, and the grid for $\sigma_{\phi}$ is $\{0.08, 0.09, 0.1, 0.11, 0.12, 0.13, 0.14, 0.15,\allowbreak 0.16, 0.17, 0.18, 0.19, 0.2\}$. Other experimental settings are the same as those in Subsection~\ref{converg1}. In Fig.~\ref{param}, we present the sensitivity analysis with respect to these two parameters. Based on the results, we set $a = 0.5$ and $\sigma_{\phi} = 0.14$, where the MSE attains relatively small values. These settings are adopted in the subsequent experiments.

\begin{figure}[htbp]   
    \centering
    \captionsetup[subfigure]{aboveskip=2pt}
    \begin{subfigure}[b]{0.33\linewidth}
        \centering
        \includegraphics[width=\linewidth]{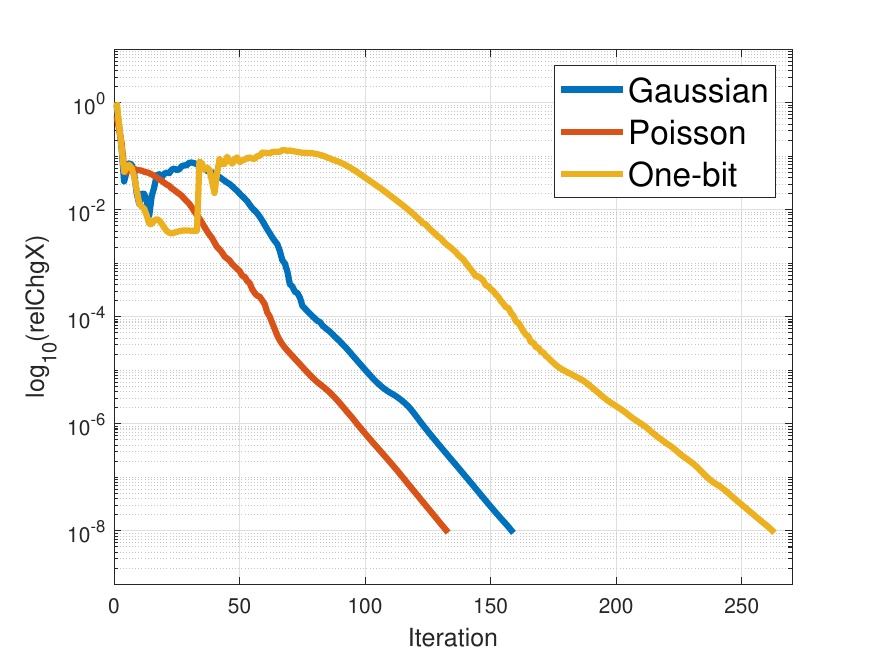}
        \caption{relChgX}
        \label{level.sub.2}
    \end{subfigure}
    \begin{subfigure}[b]{0.33\linewidth}
        \centering
        \includegraphics[width=\linewidth]{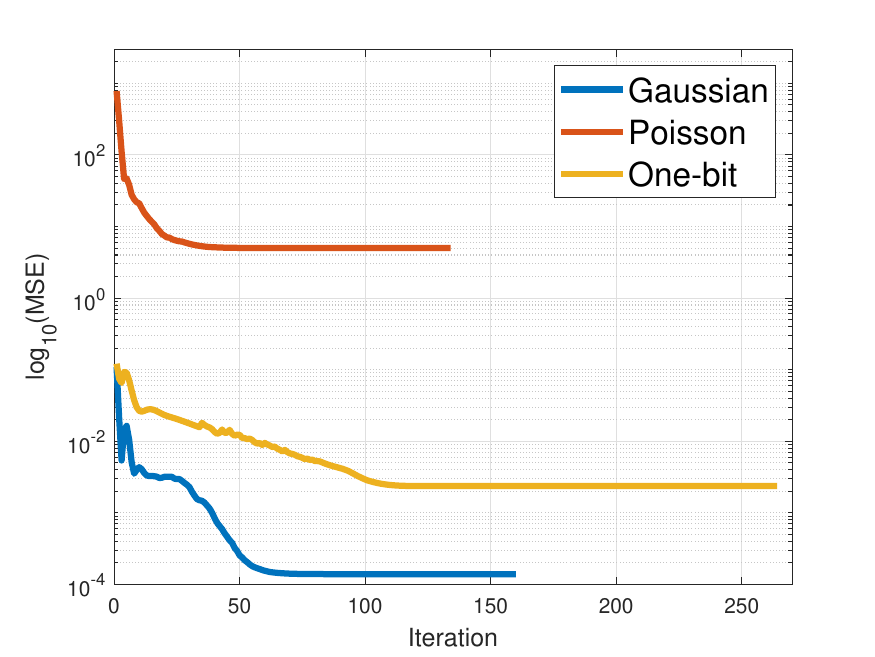}
        \caption{MSE}
        \label{level.sub.1}
    \end{subfigure}
    \caption{Convergence of Algorithm~\ref{alg1} for three types of noise.}
    \label{converg}
\end{figure}
\begin{figure}[htbp]   
    \centering
    \captionsetup[subfigure]{aboveskip=2pt}
    \begin{subfigure}[b]{0.32\linewidth}
        \centering
        \includegraphics[width=\linewidth]{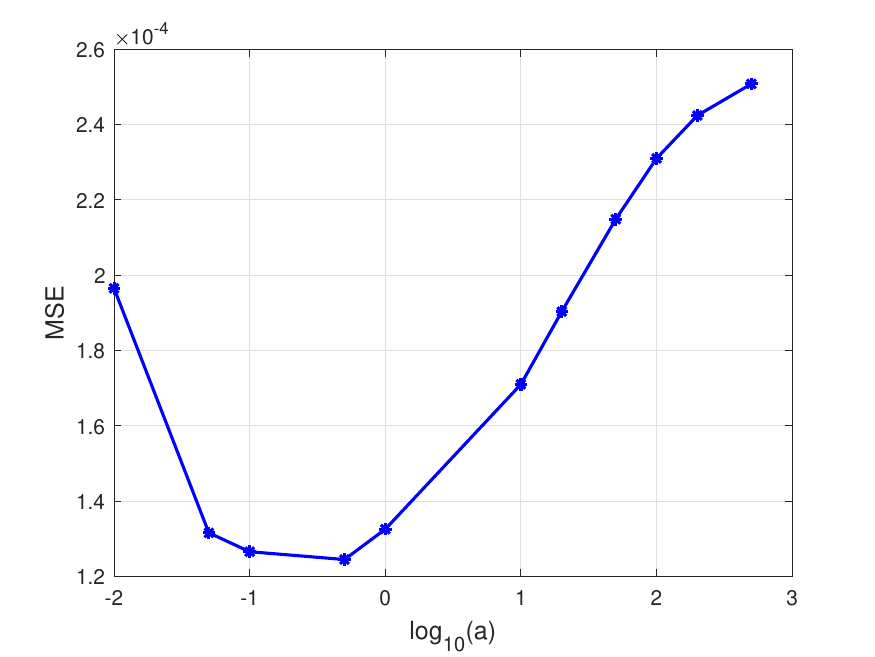}
        \caption{Gaussian Tensor Completion}
        \label{level.sub1}
    \end{subfigure}
    \begin{subfigure}[b]{0.32\linewidth}
        \centering
        \includegraphics[width=\linewidth]{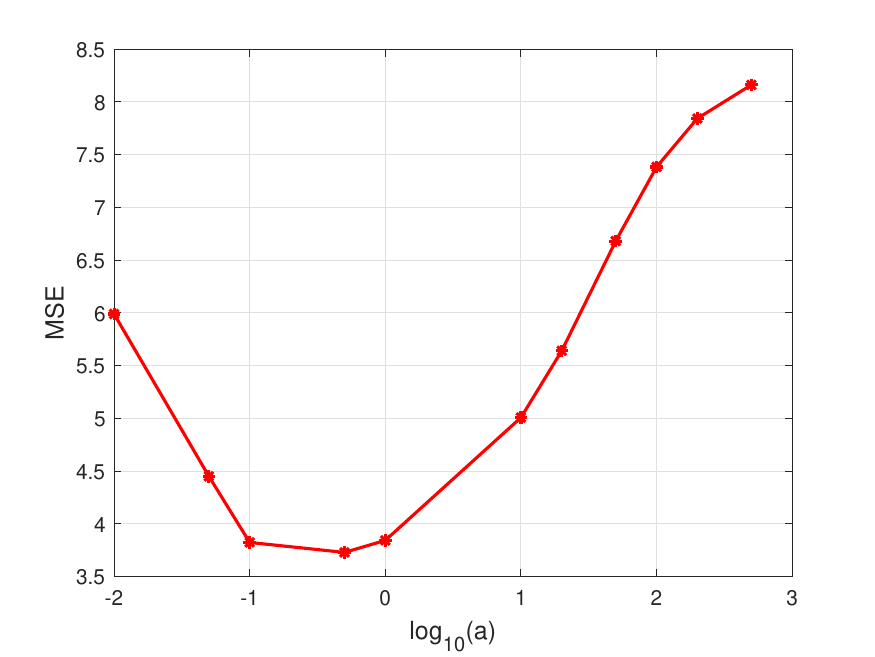}
        \caption{Poisson Tensor Completion}
        \label{level.sub2}
    \end{subfigure}
    \begin{subfigure}[b]{0.33\linewidth}
        \centering
        \includegraphics[width=\linewidth]{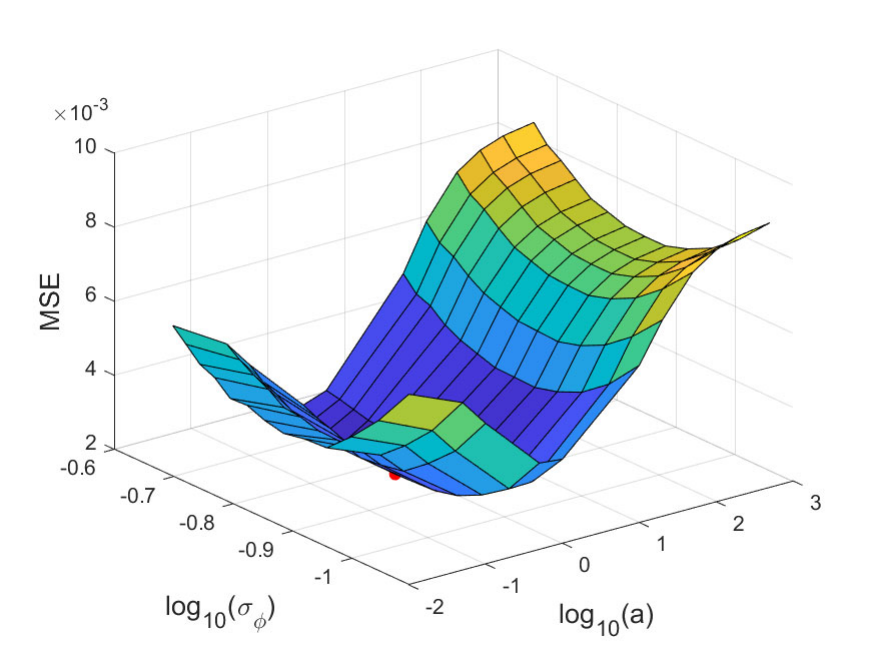}
        \caption{One-bit Tensor Completion}
        \label{level.sub3}
    \end{subfigure}
    \caption{Sensitivity analysis for the internal parameters ($a$ and $\sigma_{\phi}$) in the proposed model TDTV$_a$.}
    \label{param}
\end{figure}
\begin{figure}[htbp] 
    \centering
    \captionsetup[subfigure]{aboveskip=2pt}
    \begin{subfigure}[b]{0.32\linewidth}
        \centering
        \includegraphics[width=\linewidth]{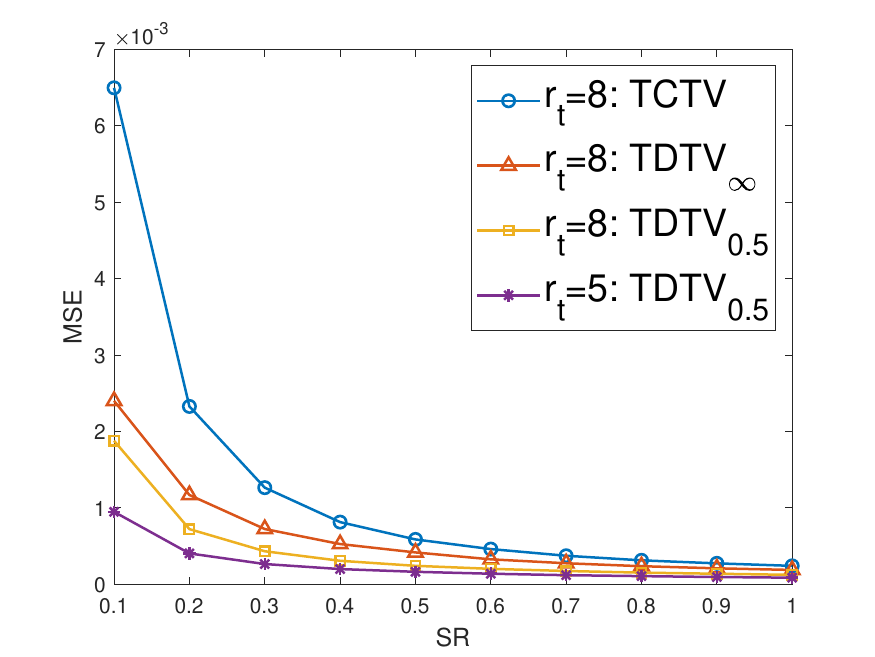}
        \caption{\small Gaussian Tensor Completion}
        \label{level.sub1}
    \end{subfigure}
    \begin{subfigure}[b]{0.32\linewidth}
        \centering
        \includegraphics[width=\linewidth]{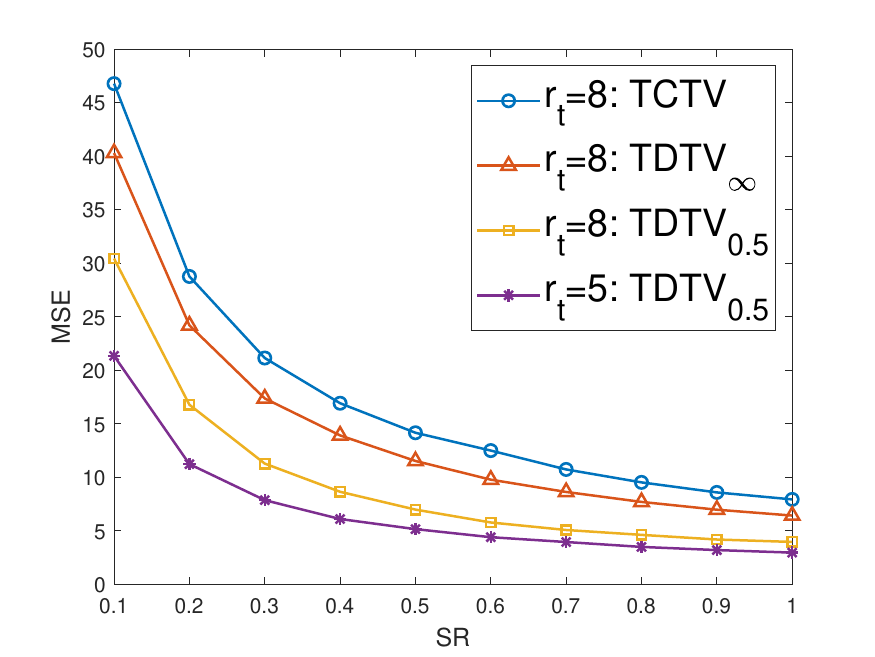}
        \caption{\small Poisson Tensor Completion}
        \label{level.sub2}
    \end{subfigure}
    \begin{subfigure}[b]{0.32\linewidth}
        \centering
        \includegraphics[width=\linewidth]{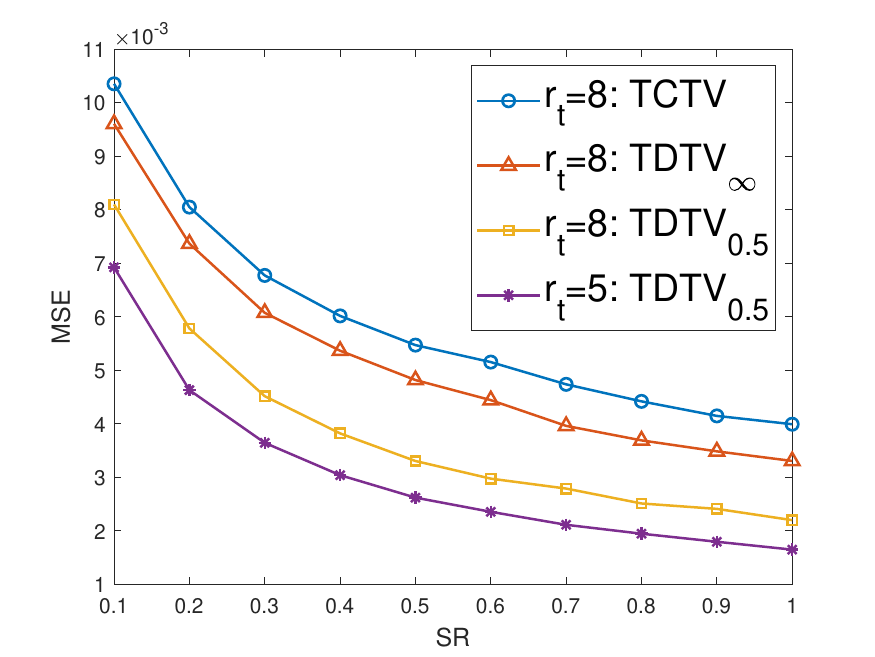}
        \caption{\small  One-bit Tensor Completion}
        \label{level.sub3}
    \end{subfigure}
    \caption{Recovery performance for three kinds of tensor completion with TCTV regularization and TDTV$_a$ ($a=\infty$,~$0.5$) regularization.}
    \label{recovery}
\end{figure}
\subsubsection{Recovery performance and running time} The original tensor $\mathcal{X}^{*}$ is generated in the same manner as in Subsection~\ref{converg1} except for tensor size and tubal-rank $r_t$.  The tensor size is set to $n_1 = n_2 = n_3 = 60$, and the tubal-rank is chosen as $r_t = 5$ and $8$, respectively.  For Gaussian, Poisson and one-bit tensor completion, their internal parameters are set as in subsection \ref{converg1}.To validate the proposed recovery theory, we examine the behavior of the MSE with respect to the sampling rate $SR$ and the tubal-rank $r_t$. Accordingly, $SR$ is selected from the set $\{0.1, 0.2, 0.3, 0.4, 0.5, 0.6, 0.7, 0.8, 0.9, 1\}$.

In Fig.~\ref{recovery}, we present the recovery curves for three tensor completion methods with TCTV (or TCTV$_{\infty}$, see Remark~\ref{TCTV}), TDTV$_{\infty}$, and TDTV$_{0.5}$ regularization, where each curve represents the mean MSE over ten runs. We note that TDTV$_{\infty}$ is actually the weighted combination of TCTV and TATV.  From the top two curves in each subfigure, TDTV$_{\infty}$ regularization performs slightly better than TCTV regularization, indicating that the joint modeling of low-rankness and sparsity in the gradient domain enhances recovery performance. Moreover, comparing the orange curve with triangle markers and the yellow curve with square markers, the nonconvex TDTV$_{0.5}$ regularization outperforms the convex TDTV$_{\infty}$ regularization. Finally, from the bottom two curves in each subfigure, the recovery performance of TDTV$_{0.5}$ improves as the tubal-rank $r_t$ decreases from 8 to 5 or as the sampling rate increases from 0.1 to 1, which empirically supports the proportional dependence of the MSE upper bound in Theorem~\ref{thm1} on $r_t$ and $n$, respectively.

\begin{table}[htbp]
	\caption{Running-time comparison of three tensor completion methods (TCTV, TDTV$_{\infty}$, and TDTV$_{0.5}$), where the mean elapsed time (in seconds) over 10 runs is reported, with standard deviations shown in parentheses.}
	\renewcommand\arraystretch{1.35}
	\setlength{\tabcolsep}{2.8pt}
	\centering
        \scriptsize
	\begin{tabular}{l | lll | lll | lll} \toprule
Tensor Size  & \multicolumn{3}{c|}{Gaussian Tensor Completion~} &\multicolumn{3}{c|}{Poisson Tensor Completion~} & \multicolumn{3}{c}{One-bit Tensor Completion}    \\ \cline{2-10}
    $(r_t=5,~s_t=10)$ & TCTV & TDTV$_{\infty}$ &  TDTV$_{0.5}$   & TCTV & TDTV$_{\infty}$ &  TDTV$_{0.5}$  & TCTV & TDTV$_{\infty}$ &  TDTV$_{0.5}$  \\ \midrule
$60\times60 \times 60$ &    5.521(0.029)~ & 6.613(0.034)~  &  11.915(0.064)~ &   5.416(0.020)~ & 6.278(0.043)~ & 10.381(0.055)~ &   7.250(0.019)~ & 8.152(0.029)~ &14.000(0.022)  \\
$60\times60 \times 90$ &   7.991(0.033)~& 9.632(0.054)~&  17.430(0.196)~& 7.920(0.039)~& 9.170(0.043)~&14.913(0.065) & 10.681(0.039) &   11.927(0.037)~&   20.870(0.049)  \\
$90\times90 \times 90$ & 17.746(0.301)~& 21.483(0.381)& 38.181(0.520)~  &17.354(0.043)~&20.095(0.112)~ &32.343(0.162)~&   24.482(0.087)~&    27.076(0.150)&   47.136(0.173) \\
$60\times60 \times 30 \times 30$ &    81.656(1.952)~ & 100.674(2.059)~ & 176.884(3.195)~ & 78.598(1.681)~& 90.379(1.225)~ &   143.330(1.976)~&   108.181(1.597)~&   125.864(1.327)~& 203.976(2.334) \\        
\bottomrule
\end{tabular}
\label{tab_time}
\end{table}
In Table~\ref{tab_time}, we report the running-time comparison for three tensor completion settings (Gaussian, Poisson, and one-bit), each using TCTV, TDTV$_{\infty}$, and TDTV$_{0.5}$ regularization. From this table, it can be seen that TDTV$_{\infty}$ uses slightly more time than TCTV  due to the incorporation of TATV regularizer. The nonconvex TDTV$_{0.5}$ requires the highest computational time, but achieves the best performance.
\subsubsection{Comparison with other methods} Three sizes of the original tensor are exploited, i.e., $70\times 70 \times 70$, $90\times 90\times 90$ and $70 \times 70 \times 90$.   The number of segmentation regions is set as $s_t= 0.1*n_1 + 2$ and the tubal-rank is set to be $r_t = 0.1*n_1$. The sampling rates are  $SR=0.2$ and $0.4$. For Gaussian, Poisson and one-bit tensor completion, their internal parameters are set as in subsection \ref{converg1}, e.g. $\sigma_{\phi}=0.14$.

For Gaussian tensor completion, we compare several relevant methods, i.e., TNN\cite{Qin2022a},  TNN+TV\cite{Qiu2021b} and t-CTV~\cite{Wang2023a}. It is noted that t-CTV corresponds to the noise-free (i.e., exact completion) setting of TCTV, while TNN+TV is also designed for exact tensor completion.  For  Poisson  tensor completion, we compare three methods, i.e., tensor nuclear norm (TNN)~\cite{Zhang2022a}, t-CTV~\cite{Wang2023a} and
TCTV  (or PCTV (FFT) in \cite{Feng2024a}).  For one-bit tensor completion, we compare TNN~\cite{Hou2021a} method. 
It is noted that TCTV, TATV, TDTV$_{\infty}$ and TDTV$_{0.5}$ are the special case of the proposed method TDTV$_{a}$.  
In Table~\ref{tab1}-\ref{tab3}, we show the mean  SNR value over 20 runs  for Gaussian, Poisson and one-bit tensor completion. Moreover,  the TCTV method is chosen as the benchmark to assess the recovery performance of other methods through the paired-sample test, where the confidence interval is computed by the test function $\mathtt{t.test()}$ in R programming. From the three tables, it can be found that the proposed TDTV$_{0.5}$ method achieves the best performance of all the competing methods in terms of SNR  due to its more delicate modeling. We can also see that  TDTV$_{\infty}$ outperforms TCTV and TATV,  suggesting that the combination of the two kinds of TV regularizations in our method can improve the recovery performance. In addition, it can be observed that the nonconvex method TDTV$_{0.5}$ can achieve significantly better recovery performance than the convex method TDTV$_{\infty}$.

\begin{table}[htpb]
	\setlength{\abovecaptionskip}{0.1cm}
	\setlength{\belowcaptionskip}{0cm}
	\caption{The mean SNR over 20 runs and the confidence interval  of confidence level $95\%$  for Gaussian tensor completion. The best and second-best values in each row are shown in  \textbf{bold} and \underline{underline}, respectively.}
	\renewcommand\arraystretch{1.35}
	\setlength{\tabcolsep}{4pt}
	\centering
        \scriptsize
\begin{tabular}{c|c|ccccccc} \toprule
Tensor Size & SR  & TNN & TNN+TV & TATV & t-CTV & TCTV  & TDTV$_{\infty}$  & TDTV$_{0.5}$ \\ \midrule
     {\fontsize{6pt}{4pt} \selectfont     \multirow{4}{*}{$70 \times 70 \times 70$}} & 0.2 &16.15  & 23.16 & 22.36& 23.65 &24.78 &\underline{25.30} &\bfseries{27.77} \\
& &{\fontsize{6.3pt}{5pt} \selectfont [-8.64, -8.61]} & {\fontsize{6.3pt}{5pt} \selectfont [-1.63, -1.6]} &{\fontsize{6.3pt}{5pt} \selectfont [-2.43, -2.4]} &  {\fontsize{6.3pt}{5pt} \selectfont [-1.14, -1.12]} & \multicolumn{1}{c}{---}  & {\fontsize{6.3pt}{5pt} \selectfont [0.51, 0.52]} & {\fontsize{6.3pt}{5pt} \selectfont [2.97, 3.01]}  \\
    & 0.4 &22.09  &25.75  &26.77 & 24.66  &27.60  &\underline{28.63} & \bfseries{31.37} \\
 &  & {\fontsize{6.3pt}{5pt} \selectfont [-5.52, -5.5]} & {\fontsize{6.3pt}{5pt} \selectfont [-1.85, -1.84]} & {\fontsize{6.3pt}{5pt} \selectfont [-0.84, -0.82]} & {\fontsize{6.3pt}{5pt} \selectfont [-2.95, -2.94]} & \multicolumn{1}{c}{---}  & {\fontsize{6.3pt}{5pt} \selectfont [1.02, 1.03]} & {\fontsize{6.3pt}{5pt} \selectfont [3.76, 3.79]}\\
   \hline
 {\fontsize{6pt}{4pt} \selectfont  \multirow{4}{*}{$90 \times 90 \times 90$}} & 0.2 & 17.96 & 24.81  &24.43  & 24.43 & 25.78 & \underline{26.65} & \bfseries{28.75} \\ 
                                                                   &      &{\fontsize{6.3pt}{5pt} \selectfont  [-7.83, -7.81]} &{\fontsize{6.3pt}{5pt} \selectfont [-0.98, -0.96]} & {\fontsize{6.3pt}{5pt} \selectfont [-1.36, -1.35]} &{\fontsize{6.3pt}{5pt} \selectfont  [-1.36, -1.35]} &\multicolumn{1}{c}{---}  & {\fontsize{6.3pt}{5pt} \selectfont  [0.86, 0.87]} &{\fontsize{6.3pt}{5pt} \selectfont  [2.95, 2.97] } \\
                                                                    &  0.4 &22.26  & 26.01  & 27.17  & 24.93  & 28.03  & \underline{28.95}  & \bfseries{31.56} \\
                                                                    &    & {\fontsize{6.3pt}{5pt} \selectfont [-5.78, -5.76]} &{\fontsize{6.3pt}{5pt} \selectfont [-2.03, -2.01]}  &{\fontsize{6.3pt}{5pt} \selectfont [-0.87, -0.85]} &	{\fontsize{6.3pt}{5pt} \selectfont [-3.1, -3.09] } &	 \multicolumn{1}{c}{---}  &{\fontsize{6.3pt}{5pt} \selectfont [0.92, 0.93]}  &{\fontsize{6.3pt}{5pt} \selectfont [3.52, 3.54]} \\ \hline                                                                       
{\fontsize{6pt}{4pt} \selectfont  \multirow{4}{*}{$70 \times 70 \times 90$}}  & 0.2 & 17.06  & 23.95  & 23.26  & 23.44  & 24.54 & \underline{25.41} & \bfseries{27.78} \\
                                                                                 &   &{\fontsize{6.3pt}{5pt} \selectfont [-7.49, -7.47]} &{\fontsize{6.3pt}{5pt} \selectfont  [-0.59, -0.57]} &{\fontsize{6.3pt}{5pt} \selectfont  [-1.28, -1.26]} &{\fontsize{6.3pt}{5pt} \selectfont [-1.1, -1.09]} & \multicolumn{1}{c}{---}  &{\fontsize{6.3pt}{5pt} \selectfont [0.86, 0.87]} &{\fontsize{6.3pt}{5pt} \selectfont  [3.23, 3.26]} \\
                                                                                & 0.4  &22.53  & 26.21 &  26.85  & 25.33  & 28.51  & \underline{29.17}  & \bfseries{31.57}  \\
                                                                               &   &{\fontsize{6.3pt}{5pt} \selectfont [-5.99, -5.97]}&{\fontsize{6.3pt}{5pt} \selectfont [-2.31, -2.29]}  &{\fontsize{6.3pt}{5pt} \selectfont [-1.66, -1.65]}&{\fontsize{6.3pt}{5pt} \selectfont [-3.18, -3.17]}  &	\multicolumn{1}{c}{---}  & {\fontsize{6.3pt}{5pt} \selectfont [0.66, 0.67]}  &	{\fontsize{6.3pt}{5pt} \selectfont [3.06, 3.07]} \\                                                                     
\bottomrule
	\end{tabular}
	\label{tab1}
\end{table}
\begin{table}[htpb]
	\setlength{\abovecaptionskip}{0.1cm}
	\setlength{\belowcaptionskip}{0cm}
	\caption{The mean SNR over 20 runs and the confidence interval of confidence level $95\%$  for Poisson tensor completion. The best and second-best values in each row are shown in  \textbf{bold} and \underline{underline}, respectively.}
	\renewcommand\arraystretch{1.35}
	\setlength{\tabcolsep}{4pt}
	\centering
        \scriptsize
	\begin{tabular}{c|c|cccccc} \toprule
		Tensor Size & SR   & TNN   & TATV  & t-CTV & TCTV & TDTV$_{\infty}$ & TDTV$_{0.5}$ \\ \midrule
            {\fontsize{6pt}{3pt} \selectfont  \multirow{4}{*}{$70 \times 70 \times 70$}} & 0.2 & 14.92 &21.61 &18.62 &21.74& \underline{22.71} &\bfseries{24.70}\\
                                                                                 &   &{\fontsize{6.3pt}{5pt} \selectfont [-6.83, -6.8]} &{\fontsize{6.3pt}{5pt} \selectfont  [-0.14, -0.11]} &{\fontsize{6.3pt}{5pt} \selectfont [-3.13, -3.1]}  &\multicolumn{1}{c}{---}  &{\fontsize{6.3pt}{5pt} \selectfont [0.96, 0.98]} &{\fontsize{6.3pt}{5pt} \selectfont  [2.93, 2.98]} \\ 
                                                                                & 0.4  &  18.63 & 24.41 & 18.65  &24.08 & \underline{25.31} & \bfseries{27.25} \\
                                                                                &  &{\fontsize{6.3pt}{5pt} \selectfont  [-5.47, -5.43]} &{\fontsize{6.3pt}{5pt} \selectfont [0.31, 0.33] }& {\fontsize{6.3pt}{5pt} \selectfont [-5.45, -5.43] } &\multicolumn{1}{c}{---}  &{\fontsize{6.3pt}{5pt} \selectfont [1.21, 1.24]}  &{\fontsize{6.3pt}{5pt} \selectfont [3.14, 3.18]}\\ \hline
{\fontsize{6pt}{3pt} \selectfont     \multirow{4}{*}{$90 \times 90 \times 90$}}& 0.2 & 14.69 & 21.94 &18.76 &22.02 &\underline{23.14} & \bfseries{24.82}\\
                                                                                &       &{\fontsize{6.3pt}{5pt} \selectfont  [-7.34, -7.32]}  &{\fontsize{6.3pt}{5pt} \selectfont  [-0.09, -0.07]} &{\fontsize{6.3pt}{5pt} \selectfont  [-3.28, -3.26]} &\multicolumn{1}{c}{---}   &{\fontsize{6.3pt}{5pt} \selectfont  [1.11, 1.13]} &{\fontsize{6.3pt}{5pt} \selectfont  [2.78, 2.81]}\\
                                                                                & 0.4  &18.42  &24.39  &18.55 & 24.11  & \underline{25.20} &\bfseries{ 27.08} \\
                                                                             &       & {\fontsize{6.3pt}{5pt} \selectfont [-5.7, -5.68]}& {\fontsize{6.3pt}{5pt} \selectfont  [0.27, 0.28]} &  {\fontsize{6.3pt}{5pt} \selectfont [-5.57, -5.55] } & \multicolumn{1}{c}{---}   & {\fontsize{6.3pt}{5pt} \selectfont [1.09, 1.1] } &{\fontsize{6.3pt}{5pt}  \selectfont  [2.95, 2.98]} \\  \hline
 {\fontsize{6pt}{3pt} \selectfont    \multirow{4}{*}{$70 \times 70 \times 90$}} & 0.2 & 14.07 &20.83&18.44 &21.48 &\underline{22.27} &\bfseries{23.95} \\
                                                                                &       &{\fontsize{6.3pt}{5pt} \selectfont   [-7.42, -7.4]} &{\fontsize{6.3pt}{5pt} \selectfont  [-0.65, -0.63]} &{\fontsize{6.3pt}{5pt} \selectfont   [-3.05, -3.03]} & \multicolumn{1}{c}{---}   &{\fontsize{6.3pt}{5pt} \selectfont   [0.78, 0.8]}  &{\fontsize{6.3pt}{5pt} \selectfont   [2.46, 2.49]} \\
                                                                                & 0.4  &17.99  &23.43 & 18.43  &23.95  &\underline{24.60}  &\bfseries{26.39} \\
                                                                             &   & {\fontsize{6.3pt}{5pt} \selectfont [-5.97, -5.95]} &{\fontsize{6.3pt}{5pt} \selectfont [-0.53, -0.51]} & {\fontsize{6.3pt}{5pt} \selectfont  [-5.53, -5.51]} & \multicolumn{1}{c}{---} & {\fontsize{6.3pt}{5pt} \selectfont
                                                                              [0.64, 0.66]} &{\fontsize{6.3pt}{5pt} \selectfont [2.43, 2.46]}\\
                                                                     
                                                                     \bottomrule
	\end{tabular}
	\label{tab2}
\end{table}
\begin{table}[htpb]
	\setlength{\abovecaptionskip}{0.1cm}
	\setlength{\belowcaptionskip}{0cm}
\caption{The mean SNR over 20 runs and the confidence interval of confidence level $95\%$ for one-bit tensor completion. The best and second-best values in each row are shown in  \textbf{bold} and \underline{underline}, respectively.}
	\renewcommand\arraystretch{1.35}
	\setlength{\tabcolsep}{4pt}
	\centering
        \scriptsize
	\begin{tabular}{c|c|ccccc} \toprule
		Tensor Size & SR  & TNN  & TATV   & TCTV & TDTV$_{\infty}$ & TDTV$_{0.5}$  \\ \midrule
       {\fontsize{6pt}{4pt} \selectfont         \multirow{4}{*}{$70 \times 70 \times 70$}} &0.2 & 10.39  & 15.07  & 15.95  & \underline{16.52}  &\bfseries{17.79} \\
       
                                                                                 &       &  {\fontsize{6.3pt}{5pt} \selectfont [-5.58, -5.53]}	&{\fontsize{6.3pt}{5pt} \selectfont [-0.89, -0.86]}&\multicolumn{1}{c}{---} &{\fontsize{6.3pt}{5pt} \selectfont [0.56, 0.59]}&{\fontsize{6.3pt}{5pt} \selectfont [1.83, 1.87]}\\
                                                                                & 0.4  & 12.48  &16.83  &17.74 & \underline{18.18}  & \bfseries{18.77}\\
                                                                                &    &  {\fontsize{6.3pt}{5pt} \selectfont  [-5.28, -5.24]} &  {\fontsize{6.3pt}{5pt} \selectfont [-0.94, -0.9]}&\multicolumn{1}{c}{---} 	&{\fontsize{6.3pt}{5pt} \selectfont [0.43, 0.45]}&{\fontsize{6.3pt}{5pt}  \selectfont [1.01, 1.04]}\\                                                                  \hline
  {\fontsize{6pt}{4pt} \selectfont              \multirow{4}{*}{$90 \times 90 \times 90$}} & 0.2 &11.47 &16.77 &17.40 &\underline{18.02} &\bfseries{18.83}\\
                                                                                &   &{\fontsize{6.3pt}{5pt} \selectfont  [-5.95, -5.91]}&{\fontsize{6.3pt}{5pt} \selectfont [-0.65, -0.61]}&  \multicolumn{1}{c}{---} &{\fontsize{6.3pt}{5pt} \selectfont [0.61, 0.64]}& {\fontsize{6.3pt}{5pt} \selectfont [1.41, 1.44]} \\
                                                                                & 0.4  &12.53 &17.36 &17.87 &\underline{18.53} &\bfseries{19.07} \\
                                                                                    && {\fontsize{6.3pt}{5pt} \selectfont [-5.36, -5.33]} &  {\fontsize{6.3pt}{5pt} \selectfont [-0.53, -0.5]} &  \multicolumn{1}{c}{---}&{\fontsize{6.3pt}{5pt} \selectfont [0.65, 0.67]} &  {\fontsize{6.3pt}{5pt} \selectfont [1.18, 1.21]}\\    \hline
    {\fontsize{6pt}{4pt} \selectfont  \multirow{4}{*}{$70 \times 70 \times 90$}} & 0.2 & 11.12  & 15.84  & 16.51  & \underline{17.05}  &\bfseries{18.07} \\
                                                                               & &{\fontsize{6.3pt}{5pt} \selectfont [-5.41, -5.37]} &{\fontsize{6.3pt}{5pt} \selectfont [-0.69, -0.66]} & \multicolumn{1}{c}{---} & {\fontsize{6.3pt}{5pt} \selectfont [0.52, 0.55]} & {\fontsize{6.3pt}{5pt} \selectfont [1.54, 1.57]} \\
                                                                                & 0.4  & 12.58  &17.17  &17.65  &\underline{18.31}  &\bfseries{18.92}  \\
                                                                               &&{\fontsize{6.3pt}{5pt} \selectfont  [-5.09, -5.05]} & {\fontsize{6.3pt}{5pt} \selectfont [-0.5, -0.47]}  & \multicolumn{1}{c}{---}   & {\fontsize{6.3pt}{5pt} \selectfont [0.65, 0.67]} & {\fontsize{6.3pt}{5pt} \selectfont  [1.25, 1.28]}\\
\bottomrule
	\end{tabular}
	\label{tab3}
\end{table}
\begin{table}[htpb]
	\caption{For Gaussian tensor completion, we report the mean MPSNR/MSSIM over 15 multispectral images (videos),  and  their confidence interval (ci.mpsnr/ci.mssim) of confidence level $95\%$ where the TCTV method is chosen as the benchmark. For each sampling rate, the best and second-best results are highlighted in \textbf{bold} and \underline{underline}, respectively.}
	\renewcommand\arraystretch{1.35}
	\setlength{\tabcolsep}{3pt}
	\centering
        \scriptsize
	\begin{tabular}{c|c|c|ccccccc} \toprule
Data&SR      & Index & TNN\cite{Qin2022a} & TNN+TV &  t-CTV   & TATV & TCTV  & TDTV$_{\infty}$ & TDTV$_{0.5}$    \\ \midrule
        \multirow{6}{*}{images} &  &MPSNR/MSSIM &28.29/0.5248  &30.82/0.7006  &28.79/0.4972    &30.54/0.8081  &34.68/0.8574  &\underline{35.47}/\underline{0.9060} & \bfseries{36.80}/\bfseries{0.9371} \\
           &0.15  & CI.mpsnr  &   [-6.80, -5.99]&  [-4.11, -3.61]&  [-6.45, -5.33]   & [-4.47, -3.82] &  \multicolumn{1}{c}{---} & [0.60, 0.98]& [1.66, 2.58] \\   
     &   & CI.mssim &   [-0.3542, -0.3111]&  [-0.1707, -0.1430] &   [-0.3924, -0.3281] &  [-0.0601, -0.0385] &  \multicolumn{1}{c}{---} & [0.0400, 0.0571] & [0.0685, 0.0909]\\   \cline{2-10} 
& &MPSNR/MSSIM &28.01/0.4796 &29.78/0.5804  &28.15/0.4604   &33.48/0.8765 &36.60/0.8957 &\underline{37.53}/\underline{0.9355} & \bfseries{38.98}/\bfseries{0.9575} \\
           & 0.3  & CI.mpsnr  &  [-9.15, -8.05] &   [-7.32, -6.33]&   [-9.10, -7.81]&  [-3.40, -2.84] &  \multicolumn{1}{c}{---} &   [0.73, 1.13]&  [1.95, 2.79]   \\   
     &   & CI.mssim &  [-0.4437, -0.3885] & [-0.3402, -0.2905]&   [-0.4688, -0.4019]   & [-0.0298, -0.0086] &  \multicolumn{1}{c}{---} &  [0.0328, 0.0469] &  [0.0517, 0.0718] \\   \midrule
\multirow{6}{*}{videos} &  &MPSNR/MSSIM &26.40/0.6329  & 27.54/0.7266   &27.26/0.6669   & 25.15/0.7323  & 30.13/0.8510  & \underline{30.31}/\underline{0.8656} & \bfseries{31.09}/\bfseries{0.8860}\\
           &0.15  & CI.mpsnr  &  [-4.26, -3.21] & [-2.96, -2.22]&  [-3.7, -2.06]  &  [-6.11, -3.85]&  \multicolumn{1}{c}{---} &  [0.06, 0.29] & [0.58, 1.33] \\   
     &   & CI.mssim &   [-0.2522, -0.1841] & [-0.1417, -0.1073] &   [-0.226, -0.1422]& [-0.1573, -0.0801] &  \multicolumn{1}{c}{---} &  [0.0055, 0.0236]&[0.0257, 0.0443] \\  \cline{2-10}  
&  &MPSNR/MSSIM &26.71/0.6247 &27.79/0.6967  &27.22/0.6468   &27.96/0.8239 &31.84/0.8790 &\underline{32.04}/\underline{0.8957} & \bfseries{32.68}/\bfseries{0.9105} \\
           &0.3  & CI.mpsnr  &    [-5.94, -4.33] &  [-4.86, -3.24]&   [-5.68, -3.57]&  [-5.08, -2.69] &  \multicolumn{1}{c}{---} &  [0.07, 0.33]&  [0.65,  1.03]   \\   
     &   & CI.mssim & [-0.2964, -0.2122] &[-0.2151, -0.1494] &  [-0.2801, -0.1842]  &  [-0.0827, -0.0273] &  \multicolumn{1}{c}{---} & [0.0075, 0.0261]& [0.0227, 0.0403] \\   
\bottomrule
\end{tabular}
\label{tab4}
\end{table}
\begin{table}[htpb]
	\caption{For Poisson tensor completion, we report the mean MPSNR/MSSIM over 15 multispectral images (videos) and  the confidence interval (ci.mpsnr/ci.mssim)  of confidence level $95\%$ where the TCTV method is chosen as the benchmark. For each sampling rate, the best and second-best results are highlighted in \textbf{bold} and \underline{underline}, respectively.}
	\renewcommand\arraystretch{1.35}
	\setlength{\tabcolsep}{3pt}
	\centering
        \scriptsize
	\begin{tabular}{c|c |c|cccccc} \toprule
Data & SR  & Index & TNN\cite{Zhang2022a} & t-CTV & TATV & TCTV &  TDTV$_{\infty}$ & TDTV$_{0.5}$    \\ \midrule
   \multirow{6}{*}{images}     & & MPSNR/MSSIM & 31.53/ 0.8133 &31.05/0.7047 &30.77/0.8588  &34.96/0.8837 &\underline{35.28}/\underline{0.9021} &\bfseries{36.76}\bfseries{0.9328} \\
 &0.15  & CI.mpsnr & [-3.83, -3.03]&[-4.62, -3.21] & [-4.53, -3.84]  & \multicolumn{1}{c}{---}  & [0.19, 0.45]&[1.31, 2.29] 	\\
 & &  CI.mssim & [-0.089, -0.052]	& [-0.2166, -0.1414]&[-0.0366, -0.0133] & \multicolumn{1}{c}{---}  	& [0.013, 0.0237] &	 [0.0377, 0.0604]	\\  \cline{2-9}
 & &  MPSNR/MSSIM  & 33.41/0.8479  &30.70/0.6770 &33.43/0.8862  &36.40/0.9001& \underline{36.67}/\underline{0.9124}& \bfseries{38.63}/\bfseries{0.9510}  \\
  & 0.3  &CI.mpsnr &[-3.43, -2.55]  & [-6.67, -4.74]&[-3.25, -2.7]  &\multicolumn{1}{c}{---}   &  [0.14, 0.39]  & [1.79, 2.67] \\
  &  & CI.mssim &[-0.0699, -0.0345] &   [-0.2721, -0.1743] &  [-0.0237, -0.0042]  &\multicolumn{1}{c}{---}   & [0.0075, 0.0171]  &  [0.0412, 0.0606] \\
     \midrule
 \multirow{6}{*}{videos}     & &MPSNR/MSSIM  & 27.01/0.7592 & 25.31/0.5903 & 24.55/0.7006  & 29.35/0.8280 & \underline{ 29.40}/\underline{0.8443} & \bfseries{29.98}/\bfseries{0.8673} \\ 
 & 0.15  &CI.mpsnr &   [-2.7, -1.97]&  [-5.18, -2.88] & [-5.83, -3.76]   & \multicolumn{1}{c}{---} & [-0.06, 0.17]&  [0.38, 0.89] \\
 &  & CI.mssim &[-0.0893, -0.0483]  &  [-0.301, -0.1743]& [-0.1655,  -0.1743]  & [-0.1656, -0.0893] &[0.0067, 0.026]  & [0.0258, 0.0529] \\ \cline{2-9}
 & &MPSNR/MSSIM & 28.71/	0.8051 & 24.95/0.5621 & 27.14/0.7904 & 30.80/0.8571  &  \underline{30.95}/\underline{0.8749} & \bfseries{31.47}/\bfseries{0.8909}   \\ 
  & 0.3&CI.mpsnr & [-2.39, -1.79]  & [-7.19, -4.52]  &  [-4.72, -2.6] &  \multicolumn{1}{c}{---} & [0.04, 0.25] & [0.43, 0.9]  \\
  & &CI.mssim &[-0.0677, -0.0363] & [-0.3631, -0.2268]& [-0.0968, -0.0367]  & \multicolumn{1}{c}{---} &  [0.0072, 0.0284] & [0.0214, 0.0462]  \\ 
  \bottomrule
\end{tabular}
\label{tab5}
\end{table}
\begin{table}[htpb]
	\caption{For one-bit tensor completion, we report the mean MPSNR/MSSIM over 15 multispectral images (videos) and  the confidence interval  (ci.mpsnr/ci.mssim)  of confidence level $95\%$ where the TCTV method is chosen as the benchmark. For each sampling rate, the best and second-best results are highlighted in \textbf{bold} and \underline{underline}, respectively.}
	\renewcommand\arraystretch{1.35}
	\setlength{\tabcolsep}{3pt}
	\centering
        \scriptsize
	\begin{tabular}{c|c|c|ccccc} \toprule
Data &SR &   Index   & TNN\cite{Hou2021a}  & TATV &TCTV & TDTV$_{\infty}$ & TDTV$_{0.5}$      \\ \midrule
     \multirow{6}{*}{images}    &\multirow{3}{*}{0.45}  & MPSNR/MSSIM &  25.71/0.5249  & 24.44/0.5282 & 28.03/0.5014  & \underline{28.94}/ \underline{0.5788}  & \bfseries{30.10}/\bfseries{0.6958}      \\
  &  & CI.mpsnr&[-2.95, -1.69]& [-4.24, -2.95] &  \multicolumn{1}{c}{---}  &  [0.72, 1.09]& [1.72, 2.42]\\
  &  & CI.mssim&  [-0.0494, 0.0964]&[-0.0187, 0.0723]&  \multicolumn{1}{c}{---}&  [0.0595, 0.0954]	& [0.1597, 0.2292]\\ \cline{2-8}
  & \multirow{3}{*}{0.6} &MPSNR/MSSIM  &  26.25/0.5491 & 25.15/0.5703  & 28.82/0.5651 & \underline{29.62}/\underline{0.6445} & \bfseries{30.72}/ \bfseries{0.7428}  \\ 
  &     & CI.mpsnr & [-3.11, -2.01]  &  [-4.28, -3.06] &\multicolumn{1}{c}{---} & [0.63, 0.98]&[1.47, 2.34]\\
  &    & CI.mssim & [-0.0745, 0.0425]  & [-0.0358, 0.0463] &\multicolumn{1}{c}{---} &[0.0581, 0.1009] &[0.1279, 0.2276] \\ \midrule
  \multirow{6}{*}{videos}    &\multirow{3}{*}{0.45} & MPSNR/MSSIM   &  24.01/0.6367 & 23.39/0.6241 & 26.41/0.7346 & \underline{26.45}/\underline{0.7399} & \bfseries{27.14}/\bfseries{0.7783}		 \\
  &  &  CI.mpsnr&  [-2.72, -2.08] & [-3.66, -2.37]	&\multicolumn{1}{c}{---} &[-0.01, 0.1] &[0.43, 1.04] \\
  &  &CI.mssim& [-0.1168, -0.079]	&[-0.1531, -0.0678]& \multicolumn{1}{c}{---}&[-0.0019, 0.0125]	&[0.0282, 0.0731]	\\ \cline{2-8}
  & \multirow{3}{*}{0.6} & MPSNR/MSSIM  & 24.52/0.6607   & 23.96/0.6504 & 26.90/0.7528 	&\underline{26.96}/\underline{0.7597}	&\bfseries{27.34}/\bfseries{0.7852} 	\\
  &  &CI.mpsnr  &[-2.7, -2.06] & [-3.63, -2.25] & \multicolumn{1}{c}{---} & [0, 0.11]& [0.24, 0.63] \\
  & & CI.mssim&[-0.1117, -0.0723] &[-0.1377, -0.067] &	\multicolumn{1}{c}{---}&[0.0031, 0.0107] &[0.0118, 0.0392]	\\ 
\bottomrule
\end{tabular}
\label{tab6}
\end{table}
\begin{figure}[htpb]
\centering
\includegraphics[width= 0.98\textwidth, height = 6cm]{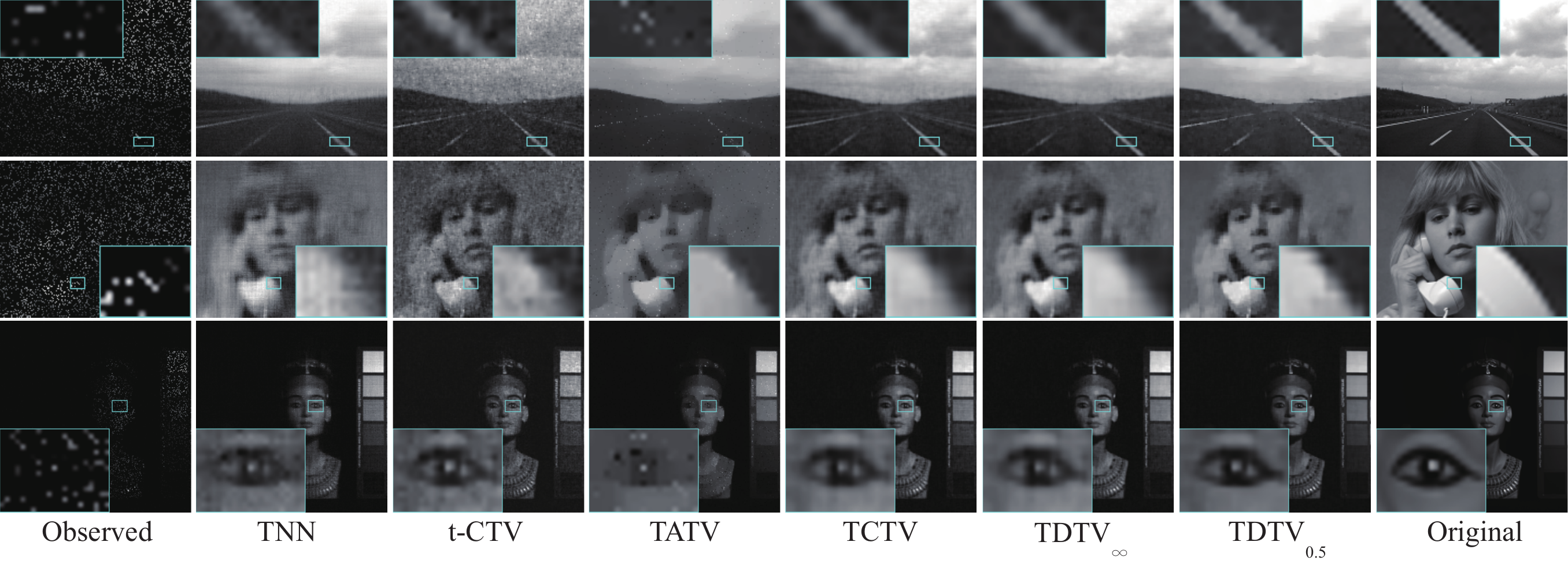}
\caption{Visual comparison for  the videos  `highway'  (frame $\#65$), `suzie'  (frame $\#35$) and the multispectral image `eygptian\_statue' (image $\#25$)  recovered by different Poisson tensor completion methods with the sampling rate of 0.15.}
\label{fig1}
\end{figure}
\subsection{Multispectral images and Videos}
Fifteen multispectral images\footnote{\url{https://cave.cs.columbia.edu/repository/Multispectral}} are  used to evaluate the effectiveness of the proposed method, including \textit{balloons}, \textit{beers}, \textit{chart\_and\_stuffed\_toy}, \textit{eygptian\_statue},  \textit{face}, \textit{fake\_and\_real\_tomatoes},  \textit{flowers}, \textit{lemon\_slices}, \textit{lemons},  \textit{sponges},  \textit{strawberries}, \textit{stuffed\_toys}, \textit{superballs}, \textit{sushi},  and \textit{yellowpeppers}, all being of size $256 \times 256 \times 31$.  Fifteen grayscale videos\footnote{\url{https://media.xiph.org/video/derf/}} are further used to evaluate the effectiveness of the proposed method, including \textit{akiyo}, \textit{bridge-far}, \textit{carphone}, \textit{container}, \textit{deadline}, \textit{foreman},  \textit{grandma}, \textit{hall\_monitor}, \textit{highway}, \textit{mad900}, \textit{miss}, \textit{news}, \textit{paris}, \textit{salesman}, and \textit{suzie}. The first 100 frames of each video are used in the experiments and the video tensor size is $176 \times 144 \times 100$.   For Gaussian and one-bit completion, all  images and videos are normalized to the interval $[0,1]$, while for Poisson completion, they are scaled to $[0,100]$. The relevant parameters are set as follows: $a = 0.5$ or $\infty$, $\sigma_{\mathrm{Gauss}} = 0.05$, and $\sigma_{\phi} = 0.19$.

Tables~\ref{tab4}-\ref{tab6} report the recovery results and the confidence intervals for three kinds of tensor completions, and Figure~\ref{fig1} displays the visual comparison of results recovered by different Poisson-completion methods.  From these tables, we observe that  TDTV$_{\infty}$ and its nonconvex extension TDTV$_{0.5}$ outperform the other methods in terms of the MPSNR and MSSIM metrics. In addition, Fig.~\ref{fig1} shows that both methods produce clearer recovered video frames and images with sharper structural details.  This observation illustrates that the joint modeling of low-rankness and sparsity in gradient domain is beneficial for visual tensor recovery, and the nonconvex regularization extension can further improve the recovery performance.

As shown in Fig.~\ref{fig2}, Gaussian tensor completion is used as an example  to illustrate the sensitivity of two regularization parameters $\lambda_g$ and $\lambda_h$ in our model. It can be seen that the recovery performance of our specific models (TDTV$_{\infty}$ and TDTV$_{0.5}$) remains relatively stable even when these regularization parameters vary over a slightly wide range.

\begin{figure}[htpb]
\centering
\includegraphics[width= 0.99\textwidth, height = 3.5cm]{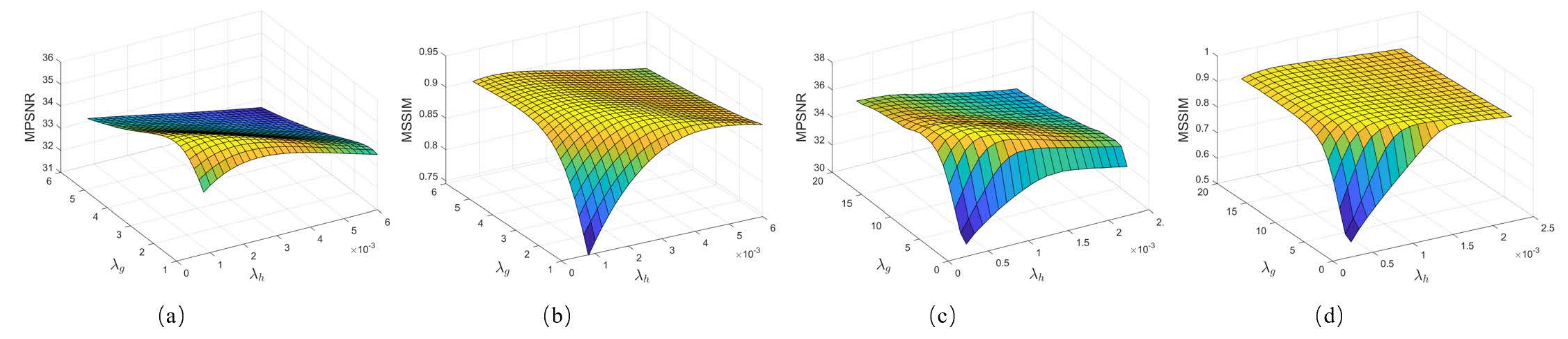}
\caption{Sensitivity analysis of regularization parameters $\lambda_g$ and $\lambda_h$ for Gaussian tensor completion model with a sampling rate of 0.15, where the multispectral image \textit{eygptian\_statue} is used. (a) MPSNR and (b) MSSIM for the TDTV$_{\infty}$ model; (c) MPSNR and (d) MSSIM for the TDTV$_{0.5}$ model. }
\label{fig2}
\end{figure}
\section{Conclusion}
In this paper, we develop a tensor nonconvex dual-TV regularization model (TDTV$_a$) for Exponential-family tensor completion,  which can be viewed as a unified extension of Gaussian, Poisson, and one-bit tensor completion.
The proposed TDTV$_a$ regularization encodes both sparsity and low-rankness in the gradient domain. 
For the estimator of the TDTV$_a$ model, we establish a rigorous upper bound using tools from high-dimensional statistics. Furthermore, for a special case of this upper bound, we derive a minimax lower bound, showing that the upper bound approaches the lower bound up to a logarithmic factor, with a gap of order $\mathcal{O}(\max_k s_k^2 / M)$.   Finally, various experiments are conducted on synthetic and real tensor data to validate the theoretical results and demonstrate the superior performance of the proposed method.

In future work, several research directions are worth exploring. First, we aim to characterize the complexity of the proposed model by estimating the covering number of its associated hypothesis space, thereby enabling a rigorous recovery error analysis. Second, we will extend the proposed framework to accommodate more general and realistic mixed noise settings.
Finally, we will systematically compare the proposed method with state-of-the-art deep learning-based completion approaches, and explore potential hybrid frameworks that combine the strengths of both paradigms.

\section*{Acknowledgments}
We are grateful to the  Associate Editor and the anonymous reviewers for  their thorough and constructive reviews. We also thank our colleague Xuguang Wei for his generous assistance with the high-performance computing.

\newpage
\begin{appendices}
\section{Some Auxiliary Lemmas}
For ease of  the statement,   $\| \cdot \|_{\circledast, \mathcal{L}}$ is abbreviated as $\| \cdot \|_{\circledast}$,  $\| \cdot \|_{\text{TL}_1, \mathcal{L}}$ is abbreviated as $\| \cdot \|_{\text{TL}_1}$,  and $\mathcal{A} \ast_{\mathcal{L}} \mathcal{B}$ is abbreviated as $\mathcal{A} \ast  \mathcal{B} $. Let us introduce more notations. $m:=\min\{n_1, n_2\}$, $M:=\max\{n_1, n_2\}$,  $N:=n_1n_2n_3$, $\widetilde{r} = r_1 + r_2 + \cdots + r_{n_3}$, $\widetilde{D}_k(\mathcal{X}) := \mathcal{X} \times_k \boldsymbol{D}_k~(k=1,2,3)$,  $\boldsymbol{D}_{k}^{\dagger}$  stands for the Moore-Penrose inverse of the 1-order difference matrix $\boldsymbol{D}_{k}$, and 
\begin{displaymath}
\| \widetilde{D}(\mathcal{X}) \|_{\circledast} : = \sum \limits_{k=1}^{3} \|   \widetilde{D}_k(\mathcal{X}) \|_{\circledast}, ~~~\text{and}~~~\| \widetilde{D}(\mathcal{X}) \|_{\ell_1} : = \sum \limits_{k=1}^{3} \|   \widetilde{D}_k(\mathcal{X}) \|_{\ell_1},
\end{displaymath}
\begin{displaymath}
\| \widetilde{D}(\mathcal{X}) \|_{\text{TL$_1$}} : = \sum \limits_{k=1}^{3} \|   \widetilde{D}_k(\mathcal{X}) \|_{\text{TL$_1$}}, ~~~\text{and}~~~\| \widetilde{D}(\mathcal{X}) \|_{\text{T}\ell_1} : = \sum \limits_{k=1}^{3} \|   \widetilde{D}_k(\mathcal{X}) \|_{\text{T}\ell_1}.
\end{displaymath}
Let us denote $\eta^{\mathcal{Y}} := \nabla \Phi_{\mathcal{Y}} \left (\mathcal{ X}^* \right )$ as the gradient of  function $\Phi_{\mathcal{Y}}$ at the point $\mathcal{ X}^* $,  and 
$$\Sigma_{\xi}:= \frac{1}{n} \sum_{o \in \Omega}  \xi_o \mathcal{E}_{o},$$
where $\{\xi_o\}_{ o \in \Omega} $ is an independent Rademecher random variable series and $\mathcal{E}_{o}$ is a basis tensor whose value is 1 at the index $o$ and  0 otherwise. 

The set $\Omega = \{ w_1, w_2, \cdots,w_n\} \in [n_1] \times [n_2] \times [n_3]$ is an observed index set of i.i.d. random variables with probability distribution $\Pi=\{ \pi_{i_1, i_2, i_3}\}$ on $[n_1] \times [n_2] \times [n_3]$.  
For a 3-order tensor $\mathcal{X} \in \mathbb{R}^{n_1 \times n_2 \times n_3}$,  its weighted sum of squares with respect to $\Pi$ is defined as:
\begin{displaymath}
\|\mathcal{X}  \|_{F(\pi)}^2 :=\sum_{i=1}^{n_1}\sum_{j=1}^{n_2}\sum_{k=1}^{n_3 } \pi_{i,j,k} \mathcal{X}_{i,j,k}^2.
\end{displaymath}

\begin{lemma}[Lemma 9 in Appendix E \cite{Wang2023a}]
\label{lem1}
For $\boldsymbol{x}\in \mathbb{R}^{N}$, denote $\nabla(x)$ as its gradient vector. Suppose $\nabla(\boldsymbol{x})$ is at most $s$-sparse, i.e., $\|\nabla(\boldsymbol{x})\|_{\ell_0} \le s$, and its sparse positions are ramdomly distributed, Then, with probability at least 
$1-2\exp(-c_0 \epsilon^2 \frac{s}{4N\| \nabla(\boldsymbol{x}) \|_{\infty}^{2}})$, it holds that 
\begin{equation}
\|\nabla(\boldsymbol{x})\|_{\ell_2} \le (1+\epsilon) \sqrt{\frac{s}{N}} \| \boldsymbol{x}\|_{2} \precsim \sqrt{\frac{s}{N}} \| \boldsymbol{x}\|_{2}.
\end{equation}
In other words,  $\|\nabla(\boldsymbol{x})\|_{\ell_2}$ is bounded by $ \sqrt{\frac{s}{N}} \| \boldsymbol{x}\|_{2}$ with high probability,
\end{lemma}

\begin{definition}[inverse scaling  factor] Let $(\boldsymbol{D}_{k}^{\dagger})^T= [\boldsymbol{t}_1, \boldsymbol{t_2}, \cdots, \boldsymbol{t_{n_k}}]$, $k=1,2,3$.  The inverse scaling factor of $(\boldsymbol{D}_{k}^{\dagger})^T$ is defined as 
\begin{equation*}
\rho_{k} = \max_{j \in [n_k] } \| \boldsymbol{t}_j \|.
\end{equation*}
The \it{total inverse scale factor} is defined as $\rho=\max_{k} \{ \rho_k\} $. 
\end{definition}

\begin{lemma}
Given 1D Total Variation matrix $\boldsymbol{D}_k \in \mathbb{R}^{(n_k -1) \times n_k}$, $\boldsymbol{D}_{k}^{T}\boldsymbol{D}_{k}$ is the so-called Laplacian matrix and satisfies the following eigen decomposition \cite{huetter16,ortelli2020, donnat2024one}:
\begin{equation*}
\boldsymbol{D}_{k}^{T}\boldsymbol{D}_{k} = V \Lambda V^T, 
\end{equation*}
where $V=[v_1, v_2,\cdots,v_{n_k} ]$, $\Lambda = \mathrm{diag}(\lambda_1, \lambda_2, \cdots, \lambda_{n_k})$,  and $0=\lambda_1 < \lambda_2 \le \cdots \le \lambda_{n_k}$ with $\lambda_j = 2 -2 \cos(\frac{(j-1)\pi}{n_k})$, $j=1,2,\cdots,n_k$. Then, it holds that 
\begin{align}
&(1)~\rho_k \le  \frac{1}{\sqrt{\lambda_2}} \le \frac{\sqrt{2} n_k}{\pi},\\
&(2)~\|\boldsymbol{D}_k\|  = \sqrt{\lambda_{n_k}} \le 2, \\
&(3)~\| \boldsymbol{D}_{k}^{\dagger}\| = \frac{1}{\sqrt{\lambda_2}} \le \frac{\sqrt{2} n_k}{\pi},~~~~~~~~~~~~~~~~~~~~~~~~~~~~~~~~~~~~~~~~~~~~~~~~~~~~~~~~~~~~~~~~~~~~~~~~~~~~~~~~~~~~~~~~~~~~~~~~~~~~\\
&(4)~\sigma_{\min}(\boldsymbol{D}_{k}^{\dagger}) = \frac{1}{\sqrt{\lambda_{n_k}}} \le \sqrt{10}
\end{align}
\end{lemma}
\begin{proof}
Noticing  that $\boldsymbol{D}_{k}^{\dagger} = (\boldsymbol{D}_{k}^{T}\boldsymbol{D}_{k})^{\dagger} \boldsymbol{D}_{k}^{T}$ and the fact that  $2-2\cos(x) \ge x^2/2$ for any $x\in [0, 1/2]$ and $2-2\cos(x) \ge 0.1$ for any $x \in [1/2, \pi]$,  we can easily prove for (2)-(3).

 Now, we shall give the proof for (1). Let $\boldsymbol{t}_j, j=1,2,\cdots,n_k$, be the $j$-column of $(\boldsymbol{D}_{k}^{\dagger})^T$.  Then, we have
\begin{align*}
\boldsymbol{t}_j = (\boldsymbol{D}_{k}^{\dagger})^T \boldsymbol{e}_j = \boldsymbol{D}_k (\boldsymbol{D}_{k}^{T}\boldsymbol{D}_{k})^{\dagger} \boldsymbol{e}_j.
\end{align*}
which shows that 
\begin{align*}
\| \boldsymbol{t}_j\|_{2}^{2}  &= \boldsymbol{t}_j^T \boldsymbol{t}_j = \boldsymbol{e}_j^{T}  (\boldsymbol{D}_{k}^{T}\boldsymbol{D}_{k})^{\dagger}  \boldsymbol{D}_k^{T} \boldsymbol{D}_k (\boldsymbol{D}_{k}^{T}\boldsymbol{D}_{k})^{\dagger} \boldsymbol{e}_j = \boldsymbol{e}_j^{T}  (\boldsymbol{D}_{k}^{T}\boldsymbol{D}_{k})^{\dagger}  \boldsymbol{e}_j\\
&= \boldsymbol{e}_j^{T} ( \sum_{k=2}^{n_k} \frac{1}{\lambda_k} v_k v_{k}^{T}) \boldsymbol{e}_j  = \sum_{k=2}^{n_k} \frac{1}{\lambda_k}  (v_{k}^{T} \boldsymbol{e}_j )^2 \le   \frac{1}{\lambda_2} \sum_{k=2}^{n_k}  (v_{k}^{T} \boldsymbol{e}_j )^2 \\
&\le \frac{1}{\lambda_2}  \|\boldsymbol{e}_j\|_{2}^{2}  = \frac{1}{\lambda_2} .
\end{align*}
So, we have that $\rho_k = \max_{j} \| \boldsymbol{t}_j\|_{2} \le \frac{1}{\sqrt{\lambda_2}}$.

(4) $\lambda_{n_k} = 2 - 2\cos( \frac{(n_k - 1)\pi}{n_k}) = 2 -2 \cos( 1- \frac{\pi}{n_k})$. Also, from $1 - \frac{\pi}{n_k} \ge \frac{1}{2}$, we get that $n_k \ge 2\pi$.  In many practical applications,  we can make an assumption that $n_k \ge 7$. Based on this assumption and the fact that $2-2\cos(x) \ge 0.1$ for any $x \in [1/2, \pi]$, we have 
\begin{equation}
\lambda_{n_k}  =  2 -2 \cos( 1- \frac{\pi}{n_k}) \ge 0.1,
\end{equation}
which implies that $\sigma_{\min}(\boldsymbol{D}_{k}^{\dagger}) = \frac{1}{ \sqrt{ \lambda_{n_k}} } \le \sqrt{10}$. \notag 
\end{proof}

\begin{proposition}[Hoeffding's inequality \cite{Vershynin2018}]
\label{hoeff1}
Let $X_1, \cdots, X_n$ be independent random variables such that  $a_i \le X_i \le b_i$ almost surely.  Consider the sum of these random variables $S_n = X_1 + \cdots + X_n$. Then, Hoeffding's theorem states that, for all $t>0$,
\begin{equation*}
\mathbb{P}\left( \lvert S_n -  \mathbb{E}(S_n) \rvert  \ge t \right) \le  2 \exp \left( - \frac{2t^2}{ \sum_{i=1}^{n} (b_i -a_i)^2} \right).
\end{equation*}

\end{proposition}
\begin{proposition}[Bernstein's inequality \cite{Vershynin2018} ] 
\label{prop1}
Let $X_1, \cdots, X_N$ be independent, mean zeros, sub-exponential random variables. Then, for  every $t\ge0$, we have 
\begin{equation*}
\mathbb{P}\left \{ \left  \lvert  \sum_{i=1}^{N} X_i  \right  \rvert  \ge t \right\} \le 2 \exp\left [  -c \min\left( \frac{t^2}{\sum_i \| X_i\|_{\psi_1}^2} , \frac{t}{ \max_i \| X_i\|_{\psi_1}} \right) \right],
\end{equation*}
where $c$ is an absolute constant. 
\end{proposition}
\begin{proposition}[\cite{Klopp2014a, Klopp2015a}]
		\label{Thm1:prop1}
		Consider a finite sequence of independent random matrices $(Z_i)_{1\leq i \leq n} \in \mathbb{R}^{n_1 \times n_2}$ satisfying $\mathbb{E}[Z_i]=0$ and for some $U>0$, $\| Z_i\| \leq U$ for all $i=1,2,\cdots,n$. Define
		\begin{displaymath}
			\sigma_Z^2:=\max \Big\{ \| \frac{1}{n}\sum_{i=1}^{n} \mathbb{E} [Z_i Z_{i}^{T}]\|,  ~~\|\frac{1}{n}\sum_{i=1}^{n} \mathbb{E} [ Z_{i}^{T}Z_i]\|\Big\}.
		\end{displaymath}
	Then, for any $n \geq (U^2 \log(d) )/(9\sigma_{Z}^{2})$, the following holds:
		\begin{displaymath}
			\mathbb{E}\big[\big\|\frac{1}{n} \sum_{i=1}^{n} Z_i \big\|\big] \leq  c^*\sigma_{Z} \sqrt{ \frac{2e\log(d)}{n}},
		\end{displaymath}
		with $c^*=1+\sqrt{3}$ and $d=n_1 + n_2$.
\end{proposition}
\begin{proposition}[\cite{Klopp2014a,Klopp2015a}]
		\label{Thm1:prop2}
		Consider a finite sequence of independent matrices $(Z_{i})_{1\leq i \leq n}\in \mathbb{R}^{n_1 \times n_2}$ satisfying $\mathbb{E}[Z_i]=0$. For some $U>0$, assume 
\begin{displaymath}
\inf\Big\{ \delta>0 : \mathbb{E}[\exp(\|Z_i\|/\delta)] \leq e \Big\} \leq U,~~\text{for  } i=1,\cdots,n
\end{displaymath}
and define $\sigma_{Z}$ as in Proposition~\ref{Thm1:prop1}. Then for any $t>0$, with probability at least $1-e^{-t}$, it holds: 
\begin{displaymath}
\| \frac{1}{n}\sum_{i} Z_i\| \leq c_U \max\Big\{ \sigma_Z \sqrt{\frac{t+\log(d)}{n}},  U \log(\frac{U}{\sigma_Z}) \frac{t+\log(d)}{n}\Big\},
\end{displaymath}
where $c_U$ is a constant which depends only on $U$.
	\end{proposition}
\begin{lemma}
\label{lem3} 
Assume that $\boldsymbol{H}_k =  \frac{1}{n_k}  \boldsymbol{1}_{n_k}   \boldsymbol{1}_{n_k} ^T$ and $\boldsymbol{1}_{n_k} \in \mathbb{R}^{n_k}$ is a vector  whose entries are all 1. Let $\mathrm{ker} (\widetilde{D}_{k} ) = \{  \mathcal{X}  \in \mathbb{R}^{n_1 \times n_2 \times n_3}: \widetilde{D}_{k} (\mathcal{X}) =0  \}$ and $\mathrm{ker}^{\perp }(\widetilde{D}_{k} )$ be its orthogonal complementary space. Denote $\Pi_{\mathbb{S}}$ as the projection on the space $\mathbb{S}$.  It follows that 
\begin{align}
&(1)~\Pi_{ \mathrm{ker} (\widetilde{D}_{k} )}(\mathcal{X}) = \mathcal{X} \times_k \boldsymbol{H}_k  ~~~\text{and}~~~\Pi_{ \mathrm{ker}^{\perp }( \widetilde{D}_{k} )}(\mathcal{X})  = \mathcal{X}\times_k \boldsymbol{D}_{k}^{\dagger} \boldsymbol{D}_{k},  ~k= 1, 2, 3. \nonumber \\
&(2)~\left \lvert \left \langle \mathcal{X}, \mathcal{Y} \right \rangle  \right \rvert \le  \sqrt{ \frac{ \pi_{\min}^{-1}}{n_k}  \cdot \sum_{i,j} \left |   \left \langle \mathcal{X}(id_k(i,j)), \boldsymbol{1}_{n_k} \right \rangle\right |^2} \cdot  \left \| \mathcal{Y}\right \|_{F(\pi)} +   \left \| \mathcal{X} \times_k (\boldsymbol{D}_{k}^{\dagger})^{T} \right \|   \left \|  \mathcal{Y} \times_k \boldsymbol{D}_{k} \right \|_{\circledast}, ~k=1,2,3, \nonumber~~~~~~~~~~~~~~~~~~~~~~~~~~~~~
\end{align}
where $id_1(i,j)=(:, i, j)$, $id_2(i,:,j)$ and $id_3(i,j)=(i,j,:)$. 
\end{lemma}
\begin{proof}  Without loss of generalization, let us give a proof for the $\widetilde{D}_1$ case corresponding to mode 1. 

 (1) For any tubal vector $\boldsymbol{x}_{ij} = \mathcal{X}(:,i,j)$,  we know that the projection on the kernel space $\text{ker}(\widetilde{D}_1)$ is 
$\Pi_{ \mathrm{ker}(\widetilde{D}_1)}(\boldsymbol{x}_{ij}) = \left \langle \boldsymbol{x}_{ij}, \boldsymbol{1}_{n_1}  \right \rangle \frac{\boldsymbol{1}_{n_1}} {n_1} = \frac{1}{n_1} \boldsymbol{1}_{n_1} \boldsymbol{1}_{n_1}^T  \boldsymbol{x}_{ij}$, and the projection on its orthogonal complementary space is $\Pi_{ \mathrm{ker}^{\perp }( \widetilde{D}_{1} )}(\mathcal{X})  =  \boldsymbol{D}_{1}^{\dagger} \boldsymbol{D}_{1}  \boldsymbol{x}_{ij}$ due to  the fact that $\boldsymbol{I}_{n_1} = \boldsymbol{H} + \boldsymbol{D}_{1}^{\dagger} \boldsymbol{D}_{1} $. This fact consequently yields the conclusion to be proved.

(2) Let $\boldsymbol{y}_{ij} = \mathcal{Y}(:,i,j)$, $\pi_{ij}= \pi(:,i,j)$ and $\pi_{\min} = \min \{ \pi_{i,j,k}\}$. From (1), we can have that 
\begin{align}
\left \langle \mathcal{X}, \mathcal{Y} \right \rangle  & = \left \langle \Pi_{ \mathrm{ker} (\widetilde{D}_{1})} (\mathcal{X}),  \mathcal{Y} \right \rangle +  \left \langle \Pi_{ \mathrm{ker}^{\perp}(\widetilde{D}_{1})} (\mathcal{X}),  \mathcal{Y} \right \rangle   = \left \langle \Pi_{ \mathrm{ker} (\widetilde{D}_{1})} (\mathcal{X}),   \mathcal{Y} \right \rangle +  \left \langle \mathcal{X},  \Pi_{ \mathrm{ker}^{\perp}(\widetilde{D}_{1})} ( \mathcal{Y}) \right \rangle  \nonumber  \\
& = \left \langle \Pi_{ \mathrm{ker} (\widetilde{D}_{1})} (\mathcal{X}),    \Pi_{ \mathrm{ker} (\widetilde{D}_{1})} (\mathcal{Y}) \right \rangle +  \left \langle \mathcal{X},  \mathcal{Y} \times_1 (\boldsymbol{D}_{1}^{\dagger} \boldsymbol{D}_{1})\right \rangle   \notag  \\
&= \left \langle \Pi_{ \mathrm{ker} (\widetilde{D}_{1})} (\mathcal{X}),   \Pi_{ \mathrm{ker} (\widetilde{D}_{1})} (\mathcal{Y}) \right \rangle +  \left \langle \mathcal{X} \times_1 (\boldsymbol{D}_{1}^{\dagger})^{T},  \mathcal{Y} \times_1 \boldsymbol{D}_{1} \right \rangle   \notag  \\
&\leq  \frac{1}{n_1} \sum_{i,j} \left\lvert  \left \langle  \boldsymbol{x}_{ij}, \boldsymbol{1}_{n_1} \right \rangle \right \rvert  \left\lvert  \left \langle  \boldsymbol{y}_{ij}, \boldsymbol{1}_{n_1} \right \rangle \right \rvert  +  \left \| \mathcal{X} \times_1 (\boldsymbol{D}_{1}^{\dagger})^{T} \right \|   \left \|  \mathcal{Y} \times_1 \boldsymbol{D}_{1} \right \|_{\circledast}  \notag  \\
&\leq  \frac{1}{n_1} \sum_{i,j} \left\lvert  \left \langle  \boldsymbol{x}_{ij}, \boldsymbol{1}_{n_1} \right \rangle \right \rvert  \left\| \mathcal{Y}(:,i,j)  \right \|_{\ell_1}  +  \left \| \mathcal{X} \times_1 (\boldsymbol{D}_{1}^{\dagger})^{T} \right \|   \left \|  \mathcal{Y} \times_1 \boldsymbol{D}_{1} \right \|_{\circledast}  \notag \\
&\leq  \frac{1}{n_1} \sum_{i,j} \left\lvert  \left \langle  \boldsymbol{x}_{ij}, \boldsymbol{1}_{n_1} \right \rangle \right \rvert 
\sqrt{\|\pi_{ij}^{-1}\|_{\ell_1}} \left\| \mathcal{Y}(:,i,j)  \right \|_{F(\pi_{ij})}  +  \left \| \mathcal{X} \times_1 (\boldsymbol{D}_{1}^{\dagger})^{T} \right \|   \left \|  \mathcal{Y} \times_1 \boldsymbol{D}_{1} \right \|_{\circledast}  \notag \\
&\leq  \frac{1}{n_1} \sum_{i,j} \left\lvert  \left \langle  \boldsymbol{x}_{ij}, \boldsymbol{1}_{n_1} \right \rangle \right \rvert 
\sqrt{n_1 \pi_{\min}^{-1} } \left\| \mathcal{Y}(:,i,j)  \right \|_{F(\pi_{ij})}  +  \left \| \mathcal{X} \times_1 (\boldsymbol{D}_{1}^{\dagger})^{T} \right \|   \left \|  \mathcal{Y} \times_1 \boldsymbol{D}_{1} \right \|_{\circledast}  \notag \\ 
&\leq  \sqrt{ \frac{\pi_{\min}^{-1} }   {n_1}}  \sum_{i,j} \left\lvert  \left \langle  \boldsymbol{x}_{ij}, \boldsymbol{1}_{n_1} \right \rangle \right \rvert 
\left\| \mathcal{Y}(:,i,j)  \right \|_{F(\pi_{ij})}  +  \left \| \mathcal{X} \times_1 (\boldsymbol{D}_{1}^{\dagger})^{T} \right \|   \left \|  \mathcal{Y} \times_1 \boldsymbol{D}_{1} \right \|_{\circledast}  \notag \\ 
&\leq  \sqrt{ \frac{\pi_{\min}^{-1} }   {n_1}} \cdot \sqrt{ \sum_{i,j} \left\lvert  \left \langle  \boldsymbol{x}_{ij}, \boldsymbol{1}_{n_1} \right \rangle \right \rvert^2 }  \cdot
\left\| \mathcal{Y} \right \|_{F(\pi)}  +  \left \| \mathcal{X} \times_1 (\boldsymbol{D}_{1}^{\dagger})^{T} \right \|   \left \|  \mathcal{Y} \times_1 \boldsymbol{D}_{1} \right \|_{\circledast}  \notag 
\end{align}
where the first inequality uses the H$\ddot{o}$lder inequality and the last inequality uses the Cauchy-Schwartz inequality.
\end{proof}

\begin{lemma}\label{lem4} Assume that  H4 holds.  It follows that  
\begin{align}
&(1)~\sum_{i,j} \lvert \left \langle  \eta^{\mathcal{Y} } (id_k(i,j)),   \boldsymbol{1}_{n_k}   \right \rangle \rvert^2  \le \frac{ \log(2N) \delta_{\alpha}^2}{n}, \\
&(2)~ \sum_{i,j} \mathbb{E}\left [  \lvert \left \langle  \eta^{\mathcal{Y} } (id_k(i,j)),   \boldsymbol{1}_{n_k}   \right \rangle \rvert^2   \right]  \precsim  \frac{\log(2N) \delta_{\alpha}^2  }{n}, ~~~~~~~~~~~~~~~~~~~~~~~~~~~~~~~~~~~~~~~~~~~~~~~~~~~~~~~~~~~~~~~~~~\\
&(3)~ \sum_{i,j} \mathbb{E}\left[ \left |   \left \langle  \Sigma_{\xi}(id_k(i,j)), \boldsymbol{1}_{n_k} \right \rangle\right |^2 \right] \precsim  \frac{2 log(2N) }{n}.
\end{align}
\begin{proof}
(1) Without loss of generalization, the proof is  given  for the $k=3$ case. Denote $\boldsymbol{a}_{1}^{ij} = \mathcal{Y}(i,j, o_1) - F'(\mathcal{X}^*(i,j, o_1)),\cdots, \boldsymbol{a}_{n_{ij}}^{ij} =  \mathcal{Y}(i,j, o_{n_{ij}}) - F'(\mathcal{X}^*(i,j, o_{n_{ij}}))$ where $n_{ij}$ is the cardinality of vector $\eta^{\mathcal{Y}}(i,j,:)$. We notice that $\boldsymbol{a}_{1}^{ij}, \cdots, \boldsymbol{a}_{n_{ij}}^{ij} $ are independent, mean zeros, sub-exponential variables.  By assumption H4, we get that $\|\boldsymbol{a}_{o_p}^{ij} \|_{\psi_1} = \delta_{\alpha}$. 
From Proposition~\ref{prop1},  we have with the probability at least $1-\frac{1}{N}$,
\begin{equation*}
\lvert \boldsymbol{a}_{1}^{ij}  + \cdots + \boldsymbol{a}_{n_{ij}}^{ij} \rvert  \le \delta_{\alpha} \sqrt{n_{ij }\log(2N) }.
\end{equation*}
Then, we have with high probability at least $1-\frac{1}{n_3}$,
\begin{align*}
\sum_{i,j} \lvert \left \langle  \eta^{\mathcal{Y} } (id_k(i,j)),   \boldsymbol{1}_{n_k}   \right \rangle \rvert^2  = \frac{  \sum_{i,j} \lvert \boldsymbol{a}_{1}^{ij}  + \cdots + \boldsymbol{a}_{n_{ij}}^{ij} \rvert^2 }{n^2} \le \frac{ \delta_{\alpha}^2  (\sum_{i,j} n_{ij })\log(2N)}{n^2} \le  \frac{ \delta_{\alpha}^2 \log(2N)}{n},
\end{align*}
where  the last inequality uses the fact that $\sum_{i,j} n_{ij} =n$. 

(2) Let $\zeta_{ij}^2 := \lvert \boldsymbol{a}_{1}^{ij}  + \cdots + \boldsymbol{a}_{n_{ij}}^{ij} \rvert^2$ and $t_{\alpha} := \delta_{\alpha} \sqrt{n_{ij }\log(2N) }$.  It can be easily seen that $\zeta_{ij}$ is a nonnegative random variable. Hence, we have 
\begin{align}
\mathbb{E}[\zeta_{ij}^2]  &= \int_{0}^{\infty} \mathbb{P}\left( \zeta_{ij}^2 > t \right) dt = \int_{0}^{t_{\alpha}^2} \mathbb{P}\left( \zeta_{ij}^2 > t \right) dt  +  \int_{t_{\alpha}^2}^{\infty} \mathbb{P}\left( \zeta_{ij}^2 > t \right) dt  \notag \\
&\le \int_{0}^{t_{\alpha}^2} 1 dt  +  \int_{0}^{\infty} \mathbb{P}\left(\zeta_{ij}> \sqrt{t} \right) dt  \le \int_{0}^{t_{\alpha}^2} 1 dt  + 2  \int_{0}^{\infty} \mathbb{P}\left( \zeta_{ij}> u  \right) u du   \notag \\
&\le \int_{0}^{t_{\alpha}^2} 1 dt  + 4  \int_{0}^{\infty}   e^{ - \frac{u^2}{n_{ij} \delta_{\alpha}^2}}u du  \le t_{\alpha}^2 + 2n_{ij} \delta_{\alpha}^2   \notag \\
&= n_{ij} \log(2N) \delta_{\alpha}^2 + 2n_{ij} \delta_{\alpha}^2   \precsim n_{ij} \log(2N) \delta_{\alpha}^2  \notag  
\end{align}
which shows that  
\begin{align*}
\sum_{i,j} \mathbb{E}\left [  \lvert \left \langle  \eta^{\mathcal{Y} } (id_k(i,j)),   \boldsymbol{1}_{n_k}   \right \rangle \rvert^2   \right]  = \sum_{i,j}  \frac{\zeta_{ij}^2 }{n^2} \precsim  \frac{\log(2N) \delta_{\alpha}^2  }{n}.
\end{align*} 

(3)  The proof is analogous  to that of (1) and (2) by using Heoffding's inequality in Proposition~\ref{hoeff1}, so we omit it here.  
\end{proof}
\end{lemma}

\begin{lemma}
	\label{lem5}
	Assume that $o_p=(i_p,j_p,k_p) \in [n_1]\times [n_2] \times [n_3]$ follows distribution $\mathbf{\pi} = \{ \pi_{i,j,k} \}$. 
Let $\mathcal{Z}_{o_p} =  \xi_{o_p} \cdot \mathcal{E}_{o_p}  \times_t \boldsymbol{T}^{(t)}  \times_3 \boldsymbol{L},  (t=1,2,3)$, and denote $\overline{Z}_{o_p} = \mathrm{bkdiag}\left(\mathcal{Z}_{o_p}^{(1)},  \cdots, \mathcal{Z}_{o_p}^{(n_3)}  \right)$. Under the assumption H3,  it follows that
\begin{align}
&(1)~\|\overline{Z}_{o_p} \| \le  \rho_t  \sqrt{\ell},  \\
&(2)~\| \mathbb{E} [ \overline{Z}_{o_p} \overline{Z}_{o_p}^{H}  ] \| \le \frac{ \nu\ell}{ m n_3} \| \boldsymbol{T}^{(t)}\|^2, \\
&(3)~\| \mathbb{E} [ \overline{Z}_{o_p}^{H}  \overline{Z}_{o_p}  ] \| \le    \frac{ \nu\ell}{ m n_3} \| \boldsymbol{T}^{(t)}\|^2.~~~~~~~~~~~~~~~~~~~~~~~~~~~~~~~~~~~~~~~~~~~~~~~~~~~~~~~~~~~~~~~~~~~~~~~~~~~~~~~~~~~~~~~~~~~~~~~~
\end{align}
\end{lemma}
\begin{proof} We only provide the proof for the $t=1$ case. Let $\boldsymbol{L} = [\boldsymbol{l}_{1}, \cdots, \boldsymbol{l}_{n_3}]$ and $\boldsymbol{T}^{(1)} =[ \boldsymbol{t}_{1}, \cdots, \boldsymbol{t}_{n_1}]$. By the definition of $\mathcal{Z}_{o_p}$ and noticing $o_p=(i_p,j_p,k_p)$ , we can compute 
\begin{equation*}
\overline{Z}_{o_p} =  \mathrm{bkdiag}\left(\mathcal{Z}_{o_p}^{(1)},  \cdots, \mathcal{Z}_{o_p}^{(n_3)}  \right) = \xi_{o_p}  \cdot \mathrm{bkdiag}\left(\boldsymbol{l}_{1,k_p}\boldsymbol{t}_{i_p} \boldsymbol{e}_{j_p}^{T}, \cdots,  \boldsymbol{l}_{n_3,k_p}\boldsymbol{t}_{i_p} \boldsymbol{e}_{j_p}^{T}\right),
\end{equation*}
which shows that 
\begin{equation*}
\| \overline{Z}_{o_p}\|  = \lvert \xi_{o_p} \rvert \cdot \max_{s} \|\boldsymbol{l}_{s,k_p}\boldsymbol{t}_{i_p} \boldsymbol{e}_{j_p}^{T}\| \le \max_{s} \lvert \boldsymbol{l}_{s,k_p} \rvert \max_{i_p} \| \boldsymbol{t}_{i_p}\| \le \sqrt{\ell} \rho_{1}.
\end{equation*}

Moveover, we can get that 
\begin{equation*}
\overline{Z}_{o_p} \overline{Z}_{o_p}^{T} =  \xi_{o_p}^2 \cdot \mathrm{bkdiag}\left ( \boldsymbol{l}_{1,k_p}^2 \boldsymbol{t}_{i_p}  \boldsymbol{t}_{i_p}^T, \cdots,   \boldsymbol{l}_{n_3 ,k_p}^2 \boldsymbol{t}_{i_p}  \boldsymbol{t}_{i_p}^T \right).
\end{equation*}
From the above equality, we can further have that 
\begin{align}
\mathbb{E}[ \overline{Z}_{o_p} \overline{Z}_{o_p}^{T}  ] &= \sum_{i,j,k} \pi_{i,j,k} \cdot \mathbb{E}[\xi_{i,j,k}^2] \cdot \mathrm{bkdiag}\left ( \boldsymbol{l}_{1,k}^2 \boldsymbol{t}_{i}  \boldsymbol{t}_{i}^T, \cdots,   \boldsymbol{l}_{n_3 ,k}^2 \boldsymbol{t}_{i}  \boldsymbol{t}_{i}^T \right) \notag \\
& = \sum_{i,j,k} \pi_{i,j,k} \cdot 1 \cdot \mathrm{bkdiag}\left ( \boldsymbol{l}_{1,k}^2 \boldsymbol{t}_{i}  \boldsymbol{t}_{i}^T, \cdots,   \boldsymbol{l}_{n_3 ,k}^2 \boldsymbol{t}_{i}  \boldsymbol{t}_{i}^T \right) \notag \\
& = \sum_{i,k} R_{ik}  \cdot \mathrm{bkdiag}\left ( \boldsymbol{l}_{1,k}^2 \boldsymbol{t}_{i}  \boldsymbol{t}_{i}^T, \cdots,   \boldsymbol{l}_{n_3 ,k}^2 \boldsymbol{t}_{i}  \boldsymbol{t}_{i}^T \right) \notag \\
&\preceq \max_{i,k} R_{ik} \cdot \sum_{i,k}  \mathrm{bkdiag}\left ( \boldsymbol{l}_{1,k}^2 \boldsymbol{t}_{i}  \boldsymbol{t}_{i}^T, \cdots,   \boldsymbol{l}_{n_3 ,k}^2 \boldsymbol{t}_{i}  \boldsymbol{t}_{i}^T \right) \notag \\
&\preceq \max_{i,k} R_{ik} \cdot \ell \cdot   \mathrm{bkdiag}\left ( \boldsymbol{T}^{(1)} (\boldsymbol{T}^{(1)})^{T}, \cdots,    \boldsymbol{T}^{(1)} (\boldsymbol{T}^{(1)})^{T}\right) \notag,
\end{align}
which implies that 
\begin{equation*}
\| \mathbb{E}[ \overline{Z}_{o_p} \overline{Z}_{o_p}^{T}  ]  \| \le   \max_{i,k} R_{ik}  \cdot  \ell \cdot \|  \boldsymbol{T}^{(1)}\|^2 \le  \frac{\nu \ell}{mn_3} \cdot \|  \boldsymbol{T}^{(1)}\|^2,
\end{equation*}
where the second inequality uses the assumption H3. 

Likewise, the similar argument can be applied such that 
$$
\| \mathbb{E}[ \overline{Z}_{o_p}^{T}  \overline{Z}_{o_p} ]  \| \le   \max_{i,k} C_{jk}  \cdot  \ell \cdot \|  \boldsymbol{T}^{(1)}\|^2 \le  \frac{\nu \ell}{mn_3} \cdot \|  \boldsymbol{T}^{(1)}\|^2. 
$$
We complete the proof.
\end{proof}
\begin{lemma}
\label{lem6}
Assume that $\Sigma_{\xi} = \frac{1}{n}\sum_{p=1}^{n} \xi_{o_p}  \mathcal{E}_{o_p}$ and H3 holds. Then, when $n \geq   \frac{m n_3 \log(d) \max_t \rho_{t}^{2} }{9\nu \max_t \| (\boldsymbol{D}_{t}^{\dagger})^T \|^2 }$, we have
\begin{displaymath}
\mathbb{E} \left [ \max_k \left \|   \Sigma_{\xi}  \times_k (\boldsymbol{D}_{k}^{\dagger})^T  \right \| \right ]   \le  c^* \cdot \max_{k} \| \boldsymbol{D}_{k}^{\dagger}\| \cdot \sqrt{ \frac{2e \nu\ell \log(d)}{n m n_3} },
\end{displaymath}
with $c^*=1+\sqrt{3}$  and $d=(n_1+n_2)n_3$.
\end{lemma}
\begin{proof}
Let $\mathcal{Z}_{o_p} =  \xi_{o_p} \cdot \mathcal{E}_{o_p}  \times_t \boldsymbol{T}^{(t)}  \times_3 \boldsymbol{L}~(t=1,2,3)$ with $\boldsymbol{T}^{(t)} = (\boldsymbol{D}_{t}^{\dagger})^T$, and denote $\overline{Z}_{o_p} = \mathrm{bkdiag}\left(\mathcal{Z}_{o_p}^{(1)},  \cdots, \mathcal{Z}_{o_p}^{(n_3)}  \right)$. 
From Lemma~\ref{lem5}, we can get that 
\begin{displaymath}
\| \overline{Z}_{o_p} \|  \leq \rho_{t} \sqrt{\ell} \le \sqrt{\ell} \max_t \rho_{t} 
\end{displaymath}
and 
\begin{displaymath}
\| \mathbb{E} [ \overline{Z}_{o_p} \overline{Z}_{o_p}^{H}  ] \| \le \frac{ \nu\ell}{ m n_3} \| \boldsymbol{T}^{(t)}\|^2\le \frac{ \nu\ell}{ m n_3} \max_t \| \boldsymbol{T}^{(t)}\|^2,~~\| \mathbb{E} [ \overline{Z}_{o_p}^{H}  \overline{Z}_{o_p}  ] \| \le    \frac{ \nu\ell}{ m n_3} \| \boldsymbol{T}^{(t)}\|^2\le \frac{ \nu\ell}{ m n_3} \max_t \| \boldsymbol{T}^{(t)}\|^2.
\end{displaymath}
Then, applying Proposition~\ref{Thm1:prop1} with $U=\sqrt{\ell} \max_{t} \rho_{t}$ and $\sigma_{Z}^2 =   \frac{ \nu\ell}{ m n_3} \max_t \|( \boldsymbol{D}_{t}^{\dagger})^T\|^2$, it holds
	\begin{displaymath}
		\mathbb{E}[ \max_t \| \Sigma_{\xi}  \times_t  (\boldsymbol{D}_{t}^{\dagger})^T \|]  = \mathbb{E}[ \max_t \| \frac{1}{n}\sum_{p=1}^{n} \overline{Z}_{o_p}  \|]  \leq c^* \cdot \max_t \| (\boldsymbol{D}_{t}^{\dagger})^T \| \cdot \sqrt{\frac{2e\nu\ell \log(d)}{n m n_3}},
	\end{displaymath}
	provided that $n\geq \frac{U^2 \log(d)}{9\sigma_{z}^{2}} = \frac{m n_3 \log(d) \max_t \rho_{t}^{2} }{9\nu \max_t \| (\boldsymbol{D}_{t}^{\dagger})^T \|^2 }$ with $d=(n_1 +n_2)n_3$.
\end{proof}
\begin{lemma}
	\label{lem7}
	Assume that $\Phi_{\mathcal{Y}}(\mathcal{X})=\frac{1}{n} \sum_{o=1}^{n} F(\mathcal{X}_{w_o}) - \mathcal{Y}_{w_o} \mathcal{X}_{w_o}$. Then, under the assumptions H1, H2, H3 and H4,  when $n > 2 \nu^{-1} (\frac{\rho\delta_{\alpha}}{\overline{\sigma}_{\alpha}   \| \boldsymbol{D}_{k}^{\dagger}\|  })^2 \log^2\big (\frac{\rho \delta_{\alpha}\sqrt{\mu mn_3}}{\underline{\sigma}_{\alpha}   \sigma_{\min}(\boldsymbol{D}_{k}^{\dagger}) } \big) \cdot mn_3 \log(d)$,  we have
	\begin{displaymath}
\| \nabla \Phi_{\mathcal{Y}}(\mathcal{X})  \times_k (\boldsymbol{D}_{k}^{\dagger})^T\| \leq  c_{\alpha}  \overline{\sigma}_{\alpha} \cdot \| \boldsymbol{D}_{k}^{\dagger}\|  \sqrt{\frac{2\nu \ell \log(d)}{nmn_3}},
	\end{displaymath}
	with probability at least $1-\frac{1}{d}$ with $d=(n_1 + n_2)n_3$, where $c_{\alpha}$ depends only on $\delta_{\alpha}$. 
\end{lemma}
\begin{proof}
	It is not hard to compute the gradient of $\Phi_{\mathcal{Y}}(\mathcal{X})$ as:
\begin{displaymath}
\nabla \Phi_{\mathcal{Y}}(\mathcal{X}) = \frac{1}{n} \sum_{o=1}^{n}\big ( F'(\mathcal{X}_{w_o} )- \mathcal{Y}_{w_o}\big)\mathcal{E}_{w_o}.
\end{displaymath}
Let $\mathcal{Z}_{o} = \big ( F'(\mathcal{X}_{w_o} )- \mathcal{Y}_{w_o}\big)\mathcal{E}_{w_o}$. We can get
\begin{displaymath}
\begin{aligned}
\mathbb{E}[\mathcal{Z}_o \times_t (\boldsymbol{D}_t^{\dagger})^T] &=\mathbb{E}_{w_o}\big[ \mathbb{E}_{\mathcal{Y}|w_o}[\big ( F'(\mathcal{X}_{w_o} )- \mathcal{Y}_{w_o}\big) \cdot \mathcal{E}_{w_o}\times_t (\boldsymbol{D}_t^{\dagger})^T]\big] \\
&= \mathbb{E}_{w_o}\big[ \big ( F'(\mathcal{X}_{w_o} )- \mathbb{E}_{\mathcal{Y}|w_o}[\mathcal{Y}_{w_o}]\big) \cdot \mathcal{E}_{w_o}\times_t \boldsymbol{D}^T] \\
			&= \mathbb{E}_{w_o}\big[ 0 \cdot \mathcal{E}_{w_o}\times_t (\boldsymbol{D}_t^{\dagger})^T] =0. 
\end{aligned}
\end{displaymath}

Let $\boldsymbol{T}^{(t)} = (\boldsymbol{D}_t^{\dagger})^T$. In the following, we give the  proof for the $t=1$ case. Denote  $\overline{Z}_o = \mathrm{bkdiag}( \mathcal{Z}_{o} \times_1 \boldsymbol{T}^{(1)} \times_3 \boldsymbol{L} )$ with $o$ corresponding to $(p,q,k)$. Through some algebra operations,  we have 
\begin{displaymath}
\overline{Z}_o = \text{diag}\big\{ \boldsymbol{l}_{1,k}\boldsymbol{t}_p \boldsymbol{e}_{q}^{T} , \cdots,  \boldsymbol{l}_{n_3,k}\boldsymbol{t}_p \boldsymbol{e}_{q}^{T} \big\} \cdot \big(F'(\mathcal{X}_{p,q,k})-\mathcal{Y}_{p,q,k}\big).
\end{displaymath}
Then, it follows that 
\begin{displaymath}
\| \overline{Z}_o \| \le | \big(F'(\mathcal{X}_{p,q,k})-\mathcal{Y}_{p,q,k}\big)| \cdot \max_k |\boldsymbol{l}_{1,k}| \cdot \max_{p} \|\boldsymbol{t}_p \| \le  | \big(F'(\mathcal{X}_{p,q,k})-\mathcal{Y}_{p,q,k}\big)| \cdot \sqrt{\ell}  \cdot \rho_1 \le  | \big(F'(\mathcal{X}_{p,q,k})-\mathcal{Y}_{p,q,k}\big)| \cdot \sqrt{\ell}  \cdot \rho
\end{displaymath}
From the assumption H4, it follows that 
\begin{equation}
\mathbb{E} \Big [ \exp \big(\frac{\|\overline{Z}_o  \|}{\rho \sqrt{\ell} \delta_{\alpha}} \big) \Big ] \leq  \mathbb{E} \Big [ \exp \big(\frac{|F'(\mathcal{X}_{p,q,k})-\mathcal{Y}_{p,q,k}  |}{ \delta_{\alpha} } \big) \Big ] \leq e.
\end{equation}

We can also obtain that 
\begin{displaymath}
\overline{Z}_{o}\overline{Z}_{o}^{H} = \text{diag}\big\{ |\boldsymbol{l}_{1,k}|^2 \boldsymbol{t}_{p}\boldsymbol{t}_{p}^T, \cdots,  |\boldsymbol{l}_{n_3,k}|^2 \boldsymbol{t}_{p}\boldsymbol{t}_{p}^T \big\} \cdot \big(F'(\mathcal{X}_{p,q,k})-\mathcal{Y}_{p,q,k}\big)^2.
\end{displaymath}
Now, its expectation can be computed as follows:
\begin{align}
\mathbb{E}[\overline{Z}_o \overline{Z}_{o}^{H}] &= \sum_p \sum_q \sum_k \pi_{p,q,k}\text{diag}\big\{ |\boldsymbol{l}_{1,k}|^2 \boldsymbol{t}_{p}\boldsymbol{t}_{p}^T, \cdots,  |\boldsymbol{l}_{n_3,k}|^2 \boldsymbol{t}_{p}\boldsymbol{t}_{p}^T \big\}\cdot \mathbb{E}\big(F'(\mathcal{X}_{p,q,k})-\mathcal{Y}_{p,q,k}\big)^2  \notag \\
&=\sum_p \sum_q \sum_k \pi_{p,q,k} \text{diag}\big\{ |\boldsymbol{l}_{1,k}|^2 \boldsymbol{t}_{p}\boldsymbol{t}_{p}^T, \cdots,  |\boldsymbol{l}_{n_3,k}|^2 \boldsymbol{t}_{p}\boldsymbol{t}_{p}^T \big\} \cdot F''(\mathcal{X}_{p,q,k})  \label{eq12} \\
&\preceq \overline{\sigma}_{\alpha}^2 \cdot \text{diag}\big\{\sum_p\sum_k \pi_{p,\cdot,k}|\boldsymbol{l}_{1,k}|^2 \boldsymbol{t}_{p}\boldsymbol{t}_{p}^T, \cdots,  \sum_p\sum_k \pi_{p,\cdot,k}|\boldsymbol{l}_{n_3,k}|^2 \boldsymbol{t}_{p}\boldsymbol{t}_{p}^T \big\}, \notag \\
&\preceq \overline{\sigma}_{\alpha}^2  \cdot \max_{p,k}( \pi_{p,\cdot,k}) \cdot \ell \cdot \text{diag}\big\{\boldsymbol{T}^{(1)}(\boldsymbol{T}^{(1)})^T, \cdots, \boldsymbol{T}^{(1)}(\boldsymbol{T}^{(1)})^T\big\},  \notag \\
&\preceq \overline{\sigma}_{\alpha}^2 \cdot \frac{\nu \ell}{mn_3}\cdot \| \boldsymbol{T}^{(1)}\|^2, \notag
\end{align}
where the above deduction uses the assumption H1 ($F''(x)\leq \overline{\sigma}_{\alpha}^2$), and H3 ($\max_{p,k} \{R_{p,k}\} \leq \frac{\nu}{mn_3}$ with $R_{p,k}=\pi_{p,\cdot,k}$). In addition, considering the fact that the sequence $\{\overline{Z}_o\}_{1\leq o \leq n}$ owns independent and identical distribution, we have  
\begin{displaymath}
	\begin{aligned}		
\frac{1}{n} \sum_{o=1}^{n} \mathbb{E} [\overline{Z}_o \overline{Z}_{o}^{H}] = \mathbb{E} [\overline{Z}_{1} \overline{Z}_{1}^{H}],
\end{aligned}
\end{displaymath}
which implies that 
\begin{equation}
\label{zzH}
\| \frac{1}{n} \sum_{o=1}^{n} \mathbb{E} [\overline{Z}_o \overline{Z}_{o}^{H}] \| \leq  \frac{\overline{\sigma}_{\alpha}^2\nu}{m n_3}\|  \boldsymbol{T}^{(1)}\|^2.
\end{equation}
Similarly,  the following inequality can be established:
\begin{equation}
\label{zHz}
\| \frac{1}{n} \sum_{o=1}^{n} \mathbb{E} [\overline{Z}_o^{H} \overline{Z}_{o}] \| \leq  \frac{\overline{\sigma}_{\alpha}^2\nu \ell}{m n_3}\| \boldsymbol{T}^{(1)}\|^2.
\end{equation}
Denote $\sigma_{Z}^{2}:=\max\big\{\| \frac{1}{n} \sum_{o=1}^{n} \mathbb{E} [\overline{Z}_o \overline{Z}_{o}^{H}] \| ,\| \frac{1}{n} \sum_{o=1}^{n} \mathbb{E} [\overline{Z}_o^{H} \overline{Z}_{o}] \|\big\}$. The combination of \eqref{zzH} and \eqref{zHz} implies $\sigma_{Z}^{2} \leq \frac{\overline{\sigma}_{\alpha}^2\nu \ell}{m n_3}\| \boldsymbol{T}^{(1)}\|^2$. 

On the other hand,  using the assumption H1 ($F''(x)\geq \underline{\sigma}_{\alpha}^2$) and the fact $\min_{p,q,k} \{R_{p,k}, C_{q,k}\} \geq \frac{1}{\mu mn_3}$,  a similar argument  gives $\sigma_{Z}^{2} \geq \underline{\sigma}_{\alpha}^2\ell \sigma_{\min}(\boldsymbol{T}^{(1)})^2/(\mu mn_3)$ based on \eqref{eq12}, where $\sigma_{\min}(\boldsymbol{T}^{(1)})$ represents the smallest non-zero singular value of $T^{(1)}$.

Applying Proposition \ref{Thm1:prop2}  for $t=\log(d)$, $U= \delta_{\alpha} (\rho \sqrt{\ell})$, and $(\underline{\sigma}_{\alpha}^2\ell\sigma_{\min}(\boldsymbol{T}^{(1)})^2)/(\mu mn_3) \leq \sigma_{Z}^{2}  \leq \frac{\overline{\sigma}_{\alpha}^2\nu\ell}{m n_3} \| \boldsymbol{T}^{(1)}\|^2 $, with probability at least $1-d^{-1}$, we have
\begin{equation*}
\|\nabla \Phi_{\mathcal{Y}}(\mathcal{X}) \times_1 \boldsymbol{T}^{(1)} \| = \|\frac{1}{n}\sum_{o=1}^{n} \overline{Z}_o\| \leq c_{\alpha} \cdot \max\Big\{\overline{\sigma}_{\alpha} \| \boldsymbol{T}^{(1)}\|  \sqrt{\frac{v \ell}{mn_3}}\sqrt{\frac{2\log(d)}{n}} ,  \delta_{\alpha}\rho\sqrt{\ell}\log\big(\frac{\rho\delta_{\alpha}\sqrt{\mu mn_3}}{\underline{\sigma}_{\alpha} \sigma_{\min}(\boldsymbol{T}^{(1)})} \big)\frac{2\log(d)}{n}\Big\} 
\end{equation*}
which $c_{\alpha}$ depends only on $\delta_{\alpha}$. This inequality combined with the condition of $n$ completes the proof.
\end{proof}

\newpage 
\section{The Proof of Thoerem 1}
In this section, we  first provide some important lemmas, and then present the proof details for Theorem 1 in the paper.  In the following, let us denote $\Pi_k : =\Pi_{ \mathrm{ker} (\widetilde{D}_{k})} $ and $\Pi_{k}^{\perp}:= \Pi_{ \mathrm{ker}^{\perp}(\widetilde{D}_{k})} $ for ease of statement. Denote  the support set $T_k:=\mathrm{supp}(\widetilde{D}_{k}(\mathcal{X}^{*}))$. Suppose that $\mathcal{X}^{*} = \mathcal{U}^{*} \ast \mathcal{S}^{*} \ast (\mathcal{V}^{*})^{H}$. Then, the space spanned by the singular tensors of $\mathcal{X}^{*}$ is expressed as
\begin{equation*}
\text{span}(\mathcal{X}^{*}) = \left \{  \mathcal{U}^* \ast \mathcal{Y}^H +   \mathcal{X}^{H} \ast (\mathcal{V}^{*})^{H} :  \mathcal{X} \in \mathbb{R}^{n_1 \times r \times n_3} , \mathcal{Y} \in \mathbb{R}^{ n_2 \times r \times n_3}    \right \}.
\end{equation*}
For any tensor $\mathcal{X} \in \mathbb{R}^{n_1 \times n_2 \times n_3}$,  the orthogonal projections \cite{Wang2023a}  on the space $\text{span}(\mathcal{X}^{*}) $ and its orthogonal complementary  space can be respectively expressed as 
\begin{equation*}
\Pi_{\mathcal{X}^{*}} (\mathcal{X})   =  \mathcal{X} -   \Pi_{\mathcal{X}^{*}}^{\perp}(\mathcal{X}),~~~\text{and}~~~\Pi_{\mathcal{X}^{*}}^{\perp}(\mathcal{X}) = \left(\mathcal{I}_{n_1} -\mathcal{U}^{*} \ast (\mathcal{U}^{*})^H \right )\ast \mathcal{X} \left ( \mathcal{I}_{n_2}  -\mathcal{V}^{*} \ast (\mathcal{V}^{*})^H \right).
\end{equation*}
Based on the projection operators above, we can deduce a useful lemma  as follows. 
\begin{lemma}
\label{lem8}
For two tensors $\mathcal{X}, \mathcal{X}^{*} \in \mathbb{R}^{n_1 \times n_2 \times n_3}$,  assume that the support set of  the true tensor $\mathcal{X}^{*}$ is $T$ and its complement set is $T^C$.  It follows that 
\begin{align}
&(1)~\left \| \mathcal{X}^* \right \|_{\ell_1} -  \left \| \mathcal{X}  \right \|_{\ell_1} \le \left \| \left( \mathcal{X} - \mathcal{X}^{*} \right )_{T} \right \|_{\ell_1} - \left \| \left( \mathcal{X} - \mathcal{X}^* \right)_{T^C} \right \|_{\ell_1}, \notag \\
&(2)~\left \| \mathcal{X}^* \right \|_{\circledast} -  \left \| \mathcal{X}  \right \|_{\circledast} \le \left \|  \Pi_{\mathcal{X}^{*}} \left( \mathcal{X} - \mathcal{X}^{*} \right ) \right \|_{\circledast} - \left \| \Pi^{\perp}_{\mathcal{X}^{*}}\left( \mathcal{X} - \mathcal{X}^* \right)\right \|_{\circledast}. \notag ~~~~~~~~~~~~~~~~~~~~~~~~~~~~~~~~~~~~~~~~~~~~~~~~~~~~~~~~~~~~~~~~~~~~~~
\end{align}
\end{lemma}
\begin{proof}
We only provide the proof for (2) here, and the proof of (1) can be imitated.  We have 
\begin{align}
\| \mathcal{X}\|_{\circledast} &= \| \mathcal{X}^{*} + (\mathcal{X} - \mathcal{X}^{*}) \|_{\circledast} \nonumber  \\
& = \| \mathcal{X}^* +   \Pi_{\mathcal{X}^{*}}(\mathcal{X} - \mathcal{X}^{*}) + \Pi_{\mathcal{X}^{*}}^{\perp}(\mathcal{X} - \mathcal{X}^{*}) \|_{\circledast} \nonumber \\ 
&\ge  \| \mathcal{X}^* +    \Pi_{\mathcal{X}^{*}}^{\perp}(\mathcal{X} - \mathcal{X}^{*}) \|_{\circledast}  - \| \Pi_{\mathcal{X}^{*}}(\mathcal{X} - \mathcal{X}^{*}) \|_{\circledast}\nonumber \\ 
&= \| \mathcal{X}^* \|_{\circledast} +    \| \Pi_{\mathcal{X}^{*}}^{\perp}(\mathcal{X} - \mathcal{X}^{*}) \|_{\circledast}  - \| \Pi_{\mathcal{X}^{*}}(\mathcal{X} - \mathcal{X}^{*}) \|_{\circledast}\nonumber, 
\end{align}
which consequently leads to the conclusion. 
\end{proof}
\begin{lemma}
\label{lem8a}
For two tensors $\mathcal{X}, \mathcal{X}^{*} \in \mathbb{R}^{n_1 \times n_2 \times n_3}$,  assume that the support set of  the true tensor $\mathcal{X}^{*}$ is $T$ and its complement set is $T^C$.  It follows that 
\begin{align}
&(1)~\left \| \mathcal{X}^* \right \|_{\text{T}\ell_1} -  \left \| \mathcal{X}  \right \|_{\text{T}\ell_1} \le \frac{1+a}{a}\left \| \left( \mathcal{X} - \mathcal{X}^{*} \right )_{T} \right \|_{\ell_1}  - \frac{1+a}{a + \alpha} \left \| \left( \mathcal{X} - \mathcal{X}^* \right)_{T^C} \right \|_{\ell_1}, \notag \\
&(2)~\left \| \mathcal{X}^* \right \|_{\text{TL$_1$}} -  \left \| \mathcal{X}  \right \|_{\text{TL$_1$}}  \le  \frac{1+a}{a}\left\|  \Pi_{\mathcal{X}^{*}} \left( \mathcal{X} - \mathcal{X}^{*} \right ) \right \|_{\circledast}   -\frac{1+a}{ a + \sqrt{N} \alpha}  \left \| \Pi^{\perp}_{\mathcal{X}^{*}}\left( \mathcal{X} - \mathcal{X}^* \right)\right \|_{\circledast}. \notag ~~~~~~~~~~~~~~~~~~~~~~~~~~~~~~~~~~~~~~~~~~~~~~~~~~~~~~~~~~~~~~~~~~~~~~
\end{align}
where $N=n_1n_2n_3$ and $\alpha$ is the maximum of the absolute values over all entries.
\end{lemma}
\begin{proof}
We only provide the proof for (2) here, and the proof of (1) can be imitated.  To begin with,  let us compute the derivative of $\left \| \mathcal{X}  \right \|_{\text{TL$_1$}}$ with respect to the tensor $\mathcal{X}$. By t-SVD,  we have $\mathcal{X} = \mathcal{U} \ast \mathcal{S} \ast \mathcal{V}^{H}$. In the transformation domain, we  have $\overline{\boldsymbol{X}} = \overline{\boldsymbol{U}} \overline{\boldsymbol{S}} \overline{\boldsymbol{V}}^{H}$, where $\overline{\boldsymbol{X}}$ indicates the block-diagonal of $\mathcal{L}(\mathcal{X})$. Then, we have $\overline{\boldsymbol{X}} = \sum_{j=1}^{mn_3} \sigma_{j}( \overline{\boldsymbol{X}}) \overline{ \boldsymbol{u}}_j \overline{\boldsymbol{v}}_{j}^{H}$, which indicates that $\sigma_j( \overline{\boldsymbol{X}} )= \overline{\boldsymbol{u}}_{j}^{H} \overline{\boldsymbol{X}} \overline{\boldsymbol{v}}_j$.  From $\left \| \mathcal{X}  \right \|_{\text{TL$_1$}}=\left \| \overline{\boldsymbol{X}} \right \|_{\text{TL$_1$}}$, we have 
\begin{align}
\frac{\partial \left \| \mathcal{X}  \right \|_{\text{TL$_1$}}}{\partial \overline{\boldsymbol{X}}} &= 
\frac{\partial \left \| \overline{\boldsymbol{X}}  \right \|_{\text{TL$_1$}}}{\partial \overline{\boldsymbol{X}}} = \frac{\partial}{\partial \overline{\boldsymbol{X}}} \sum_{j} \frac{(1+a)\sigma_j( \overline{\boldsymbol{X}} )}{a + \sigma_j( \overline{\boldsymbol{X}} ) } =\sum_{j}  \frac{\partial}{\partial \overline{\boldsymbol{X}}}  \frac{(1+a)\sigma_j( \overline{\boldsymbol{X}} )}{a + \sigma_j( \overline{\boldsymbol{X}} ) }=\sum_{j}  \frac{a(1+a)}{(a + \sigma_j( \overline{\boldsymbol{X}} ))^2 }  \frac{\partial \sigma_j( \overline{\boldsymbol{X}}) }{\partial \overline{\boldsymbol{X}}}    \notag \\
&= \sum_{j}  \frac{a(1+a)}{(a + \sigma_j( \overline{\boldsymbol{X}} ))^2 }  \frac{\partial (\overline{\boldsymbol{u}}_{j}^{H} \overline{\boldsymbol{X}} \overline{\boldsymbol{v}}_j) }{\partial \overline{\boldsymbol{X}}} = \sum_{j}  \frac{a(1+a)}{(a + \sigma_j( \overline{\boldsymbol{X}} ))^2 }  \frac{\partial ( tr\{ \overline{\boldsymbol{X}} \overline{\boldsymbol{v}}_j\overline{\boldsymbol{u}}_{j}^{H}\}) }{\partial \overline{\boldsymbol{X}}} \notag \\
&= \sum_{j}  \frac{a(1+a)}{(a + \sigma_j( \overline{\boldsymbol{X}} ))^2 }  (\overline{\boldsymbol{u}}_{j}\overline{\boldsymbol{v}}_j^{H})\frac{  \partial \overline{\boldsymbol{X}}}{\partial \overline{\boldsymbol{X}}} =  \sum_{j}  \frac{a(1+a)}{(a + \sigma_j( \overline{\boldsymbol{X}} ))^2 }  (\overline{\boldsymbol{u}}_{j}\overline{\boldsymbol{v}}_j^{H}),
\end{align}
which shows that 
\begin{equation*}
\frac{\partial \left \| \mathcal{X}  \right \|_{\text{TL$_1$}}}{\partial \mathcal{X}}  =  \sum_{j}  \frac{a(1+a)}{(a + \sigma_j( \overline{\boldsymbol{X}} ))^2 }  \mathcal{U}_{j} \ast_{\mathcal{L}} \mathcal{V}_{j}^{H}.
\end{equation*}

By the mean value theorem, there exists a  a tensor $\widetilde{\mathcal{X}}$ between $\mathcal{X}$ and $\mathcal{X}^{*} + \Pi^{\perp}_{\mathcal{X}^{*}}\left( \mathcal{X} - \mathcal{X}^* \right)$ such that 
\begin{align}
\left \| \mathcal{X}  \right \|_{\text{TL$_1$}} &= \left \| \mathcal{X}^* + (\mathcal{X} - \mathcal{X}^*) \right \|_{\text{TL$_1$}}= \left \| \mathcal{X}^* + \Pi^{\perp}_{\mathcal{X}^{*}}(\mathcal{X} - \mathcal{X}^*) + \Pi_{\mathcal{X}^{*}}(\mathcal{X} - \mathcal{X}^*)  \right \|_{\text{TL$_1$}} \notag \\
&= \left \| \mathcal{X}^* + \Pi^{\perp}_{\mathcal{X}^{*}}(\mathcal{X} - \mathcal{X}^*)  \right \|_{\text{TL$_1$}} + \left \langle \frac{\partial \left \| \mathcal{X}  \right \|_{\text{TL$_1$}}}{\partial \mathcal{X}} \Big \vert_{\mathcal{X} =\widetilde{\mathcal{X}}} ,  ~ \Pi_{\mathcal{X}^{*}}(\mathcal{X} - \mathcal{X}^*) \right  \rangle \notag \\
&\ge  \left \| \mathcal{X}^* \right \|_{\text{TL$_1$}}  + \left \| \Pi^{\perp}_{\mathcal{X}^{*}}(\mathcal{X} - \mathcal{X}^*)  \right \|_{\text{TL$_1$}}  - \left\|  \frac{\partial \left \| \mathcal{X}  \right \|_{\text{TL$_1$}}}{\partial \mathcal{X}} \Big \vert_{\mathcal{X} =\widetilde{\mathcal{X}}} \right\|  \left\| \Pi_{\mathcal{X}^{*}}(\mathcal{X} - \mathcal{X}^*)   \right\|_{\circledast}  \notag \\
&\ge  \left \| \mathcal{X}^* \right \|_{\text{TL$_1$}}  + \left \| \Pi^{\perp}_{\mathcal{X}^{*}}(\mathcal{X} - \mathcal{X}^*)  \right \|_{\text{TL$_1$}}  -  \frac{a(1+a)}{a^2} \left\| \Pi_{\mathcal{X}^{*}}(\mathcal{X} - \mathcal{X}^*)   \right\|_{\circledast} \label{lem9:eq1}  
\end{align}
On the other hand,   from the fact that $ \left \| \mathcal{X} \right \|_{\text{TL$_1$}}  \ge   \frac{1 + a }{ a + \sigma_{max}(\mathcal{X})}\left \| \mathcal{X}^* \right \|_{\circledast} \ge  \frac{1 + a }{ a  + \| \mathcal{X}\|_{F} }\left \| \mathcal{X}^* \right \|_{\circledast} \ge\frac{1 + a }{ a  + \sqrt{N} \alpha }\left \| \mathcal{X}^* \right \|_{\circledast}$, we have 
\begin{equation}
\left \| \Pi^{\perp}_{\mathcal{X}^{*}}(\mathcal{X} - \mathcal{X}^*)  \right \|_{\text{TL$_1$}} \ge \frac{1+a}{ a + \sqrt{N} \alpha} 
\left \| \Pi^{\perp}_{\mathcal{X}^{*}}(\mathcal{X} - \mathcal{X}^*)  \right \|_{\circledast}.  \label{lem9:eq2}  
\end{equation}
Combining \eqref{lem9:eq1} and \eqref{lem9:eq2} yields the conclusion.
\end{proof}

\begin{lemma}[Proposition 1 \cite{Feng2024a}]
\label{lem8b}
For any tensor $\mathcal{X} \in \mathbb{C}^{n_1 \times n_2 \times n_3}$, the gradient tensor $\mathcal{G}_k=\widetilde{D}_k(\mathcal{X})$ is obtained along its $k$-th mode. Based on the transformed t-SVD, the following inequality is estbalished:
$$\widetilde{r} - n_3 \le \mathrm{rank}_{s}(\mathcal{G}_k) \le \widetilde{r},$$
where $\widetilde{r} = \sum_{i=1}^{n_3} r_i$ and $\mathrm{rank}_{s}(\mathcal{G}_k) = \sum_{i=1}^{n_3} \mathrm{rank}(\mathcal{G}_{k}^{(i)})$.
\end{lemma}

\begin{lemma}[TDTV$_a$ ($0<a<\infty$) model]
\label{lem9b}
Assume H1 holds, $ \lambda_g \ge \max_{k} \frac{2}{3} \frac{a+\alpha \sqrt{N}}{1+a} \left \| \nabla \Phi_{\mathcal{Y}} \left (\mathcal{ X}^* \right ) \times_k  (D_{k}^{\dagger})^T \right \|$ and $\lambda_h \ge (\frac{a+\alpha}{a+\alpha \sqrt{N}} + \frac{a+\alpha}{a}) \lambda_g \sqrt{\frac{\widetilde{r}}{\ell}}$. Let $\widetilde{r}$ be the sum of multi-ranks of $\mathcal{X}^*$ and  $T_k=\mathrm{supp}\left(\widetilde{D}_{k}(\mathcal{X}^{*})\right)$.  Then, it holds that 
\begin{flalign}
&(1)~\frac{ \underline{\sigma}_{\alpha}^{2}}{2n} \sum \limits_{o=1}^{n}    \left( \mathcal{X}_o - \mathcal{X}_{o}^{*} \right)^2  \le   \Xi_1^a \cdot  \left\| \mathcal{X} -\mathcal{X}^{*} \right \|_{F(\pi)}, \notag \\
&(2)~\sum_{k=1}^{3}   \left \|  \widetilde{D}_{k}\left( \mathcal{X} - \mathcal{X}^*  \right) \right\|_{\circledast}   \le  \frac{2(a+\sqrt{N} \alpha) }{ (1+a)\lambda_g} \Xi_1^a \cdot   \left \|  \mathcal{X} - \mathcal{X}^*  \right\|_{F(\pi)}~~~~~~~~~~~~~~~~~~~~~~~~~~~~~~~~~~~~~~~~~~~~~~~~~~~~~~~~~~~~~~~~~~~~~~~\notag 
\end{flalign}
where  $\Xi_1^a$  has the following form: $\Xi_1^a =$\\
\begin{align}
 \sum_{k=1}^{3}\left(  \sqrt{ \frac{\pi_{\min}^{-1}}{9n_k}  \sum_{i,j}  \lvert \left \langle  \eta^{\mathcal{Y} } (id_k(i,j)),   \boldsymbol{1}_{n_k}   \right \rangle \rvert^2  } +(1+\epsilon)( \frac{1+a}{a + \sqrt{N} \alpha} + \frac{1+a}{a} ) \lambda_g\sqrt{\frac{\mu \widetilde{r}|T_k|}{\ell}}+ \frac{(1+\epsilon)(1+a)}{a} \lambda_h\sqrt{\mu}|T_k|  \right)  \label{Xi1a} 
\end{align}
\end{lemma}
\begin{proof}
Based on Lemmas~\ref{lem1}, \ref{lem3},  \ref{lem8a}, and~\ref{lem8b}, we present the proof as follows.

~~Denote $\Psi_{\mathcal{Y}}(\mathcal{X}) :=  \Phi_{\mathcal{Y}}(\mathcal{X})   +\lambda_g \text{TCTV}_{a} (\mathcal{X})+ \lambda_{h}\text{TATV}_{a}( \mathcal{X})$,  $\eta^{\mathcal{Y}} := \nabla \Phi_{\mathcal{Y}} \left (\mathcal{ X}^* \right ) $,
 $\Delta_{\mathcal{X}} = \mathcal{X} - \mathcal{X}^{*}$,  $\Xi_{\eta} := \left \| \eta^{\mathcal{Y}} \times_k  (\boldsymbol{D}_{k}^{\dagger})^T \right \|$, and 
$$\triangle_{\lambda_g, \lambda_h} :=  \lambda_g \left [  \left \|  \widetilde{D} \left(  \mathcal{X}^* \right) \right \|_{\text{TL$_1$}} -  \left \|  \widetilde{D} \left(  \mathcal{X} \right) \right \|_{\text{TL$_1$}}  \right ] + \lambda_h \left [  \left \|  \widetilde{D} \left(  \mathcal{X}^* \right) \right \|_{\text{T}\ell_1 }  -  \left \|  \widetilde{D} \left(  \mathcal{X} \right) \right \|_{\text{T}\ell_1 }  \right ].$$ 
Let $\mathcal{X}$ be the minimum solution of $\Psi_{\mathcal{Y}}(\mathcal{X})$. 
By $\Psi_{\mathcal{Y}}(\mathcal{X}) \leq \Psi_{\mathcal{Y}}(\mathcal{X}^*)$, it follows that 
\begin{align}
& \frac{1}{n} \sum \limits_{o=1}^{n} \left [ F(\mathcal{X}_0) - F(\mathcal{X}_{o}^{*}) - F'(\mathcal{X}_{o}^{*})(\mathcal{X}_o - \mathcal{X}_{o}^{*} ) \right ]  \notag \\
&\le  \frac{1}{n} \sum\limits_{o=1}^{n} \left  [ \mathcal{Y}_o \left( \mathcal{X}_o - \mathcal{X}_{o}^{*} \right) - F'(\mathcal{X}_{o}^{*})\left(\mathcal{X}_o - \mathcal{X}_{o}^{*} \right ) \right ] + \triangle_{\lambda_g, \lambda_h}   \notag  \\
&\le  \left \langle \eta^\mathcal{Y},   \Delta_{\mathcal{X}}  \right \rangle +   \triangle_{\lambda_g, \lambda_h}   \le \frac{1}{3} \sum_{k=1}^{3}\left(  \sqrt{ \frac{\pi_{\min}^{-1}}{n_k}  \sum_{i,j}  \lvert \left \langle  \eta^{\mathcal{Y} } (id_k(i,j)),   \boldsymbol{1}_{n_k}   \right \rangle \rvert^2  } \cdot \left \| \Delta_{\mathcal{X}}  \right\|_{F(\pi)}  \right. \notag  \\
&\left. + \left \| \eta^{\mathcal{Y}} \times_k  (D_{k}^{\dagger})^T \right \| \left \| \widetilde{D}_{k}\left( \Delta_{\mathcal{X}}   \right) \right\|_{\circledast} \right)  +  \triangle_{\lambda_g, \lambda_h}   \notag \\
&\le \sum_{k=1}^{3}\left(  \sqrt{ \frac{\pi_{\min}^{-1}}{9n_k}  \sum_{i,j}  \lvert \left \langle  \eta^{\mathcal{Y} } (id_k(i,j)),   \boldsymbol{1}_{n_k}   \right \rangle \rvert^2  } \cdot \left \|  \Delta_{\mathcal{X}} \right\|_{F(\pi)} + \frac{ \Xi_{\eta}}{3}\left \|
\Pi_{\widetilde{D}_{k}\left(\mathcal{X}^{*} \right)} \left( \widetilde{D}_{k}\left(\Delta_{\mathcal{X}}   \right) \right) \right\|_{\circledast}   \right.  \notag \\ 
&\left. +  \frac{ \Xi_{\eta}}{3}\left \| 
\Pi^{\perp}_{\widetilde{D}_{k}\left(\mathcal{X}^{*} \right)} \left( \widetilde{D}_{k}\left( \Delta_{\mathcal{X}} \right) \right) \right\|_{\circledast} \right) + \triangle_{\lambda_g, \lambda_h},    \label{lem3:1}
\end{align}
where the third inequality uses  Lemma~\ref{lem3} (2).  In addition, from Lemma~\ref{lem8a}, we can obtain

\begin{equation}
\begin{split}
 \triangle_{\lambda_g, \lambda_h} &\le  \lambda_g  \sum_{k=1}^{3}  \left(  \frac{1+a}{a} \left \|
\Pi_{\widetilde{D}_{k}\left(\mathcal{X}^{*} \right)} \left( \widetilde{D}_{k}\left( \Delta_{\mathcal{X}}  \right) \right) \right\|_{\circledast}   - \frac{1+a}{ a + \sqrt{N} \alpha} \left \|
\Pi^{\perp}_{\widetilde{D}_{k}\left(\mathcal{X}^{*} \right)} \left( \widetilde{D}_{k}\left( \Delta_{\mathcal{X}}  \right) \right) \right\|_{\circledast}  \right)  \\ \notag 
&+\lambda_h  \sum_{k=1}^{3}  \left( \frac{1+a}{a} \left \|
\left( \widetilde{D}_{k}\left( \Delta_{\mathcal{X}}  \right) \right)_{T_k} \right\|_{\ell_1 }   -   \frac{1+a}{ a + \alpha}\left \|
\left( \widetilde{D}_{k}\left(\Delta_{\mathcal{X}}   \right) \right)_{T_k^C}  \right\|_{\ell_1}  \right) .
\end{split}
\end{equation} 
Substituting the above  inequality into \eqref{lem3:1} and combining the following inequality induced by assumption H1,
 $$F(\mathcal{X}_0) - F(\mathcal{X}_{o}^{*})  - F'(\mathcal{X}_{o}^{*})(\mathcal{X}_o - \mathcal{X}_{o}^{*} ) \ge  \frac{ \underline{\sigma}_{\alpha}^{2} }{2} \left( \mathcal{X}_o - \mathcal{X}_{o}^{*} \right)^2 ,$$
we have that 
\begin{align}
&\frac{ \underline{\sigma}_{\alpha}^{2}}{2n} \sum \limits_{o=1}^{n}    \left( \mathcal{X}_o - \mathcal{X}_{o}^{*} \right)^2  \le \sum_{k=1}^{3}  \left(  \sqrt{ \frac{\pi_{\min}^{-1}}{9n_k}  \sum_{i,j}  \lvert \left \langle  \eta^{\mathcal{Y} } (id_k(i,j)),   \boldsymbol{1}_{n_k}   \right \rangle \rvert^2  } \cdot \left \| \Delta_{\mathcal{X}}  \right\|_{F(\pi)}     \right.   \notag  \\
&+\left.  \Big(  \frac{ \Xi_{\eta} }{3}+  \frac{\lambda_g a}{1+a} \Big ) \left \|
\Pi_{\widetilde{D}_{k}\left(\mathcal{X}^{*} \right)} \left( \widetilde{D}_{k}\left(\Delta_{\mathcal{X}}   \right) \right) \right\|_{\circledast}   + \Big(  \frac{ \Xi_{\eta}}{3}-  \frac{\lambda_g(1+a)}{a + \sqrt{N} \alpha} \Big )   \left \|
\Pi^{\perp}_{\widetilde{D}_{k}\left(\mathcal{X}^{*} \right)} \left( \widetilde{D}_{k}\left( \Delta_{\mathcal{X}}   \right) \right) \right\|_{\circledast} \right.  \notag \\
&\left. +\frac{\lambda_h(1+a)}{a}  \left \|
\left( \widetilde{D}_{k}\left(\Delta_{\mathcal{X}}  \right) \right)_{T_k} \right\|_{\ell_1}  - \frac{\lambda_h(1+a)}{a+\alpha}  \left \|
\left( \widetilde{D}_{k}\left( \Delta_{\mathcal{X}}   \right) \right)_{T_k^{C}} \right\|_{\ell_1} \right)   \notag  \\
& \le \sum_{k=1}^{3}  \left(  \sqrt{ \frac{\pi_{\min}^{-1}}{9n_k}  \sum_{i,j}  \lvert \left \langle  \eta^{\mathcal{Y} } (id_k(i,j)),   \boldsymbol{1}_{n_k}   \right \rangle \rvert^2  } \cdot \left \| \Delta_{\mathcal{X}}  \right\|_{F(\pi)} + \Big (\frac{1+a}{2(a+\sqrt{N}\alpha)}  + \frac{1+a}{a}\Big)  \lambda_g \left \|
\Pi_{\widetilde{D}_{k}\left(\mathcal{X}^{*} \right)} \left( \widetilde{D}_{k}\left(\Delta_{\mathcal{X}}   \right) \right) \right\|_{\circledast} -    \right.  \notag \\  
&\left.\frac{\lambda_g (1+a)}{2(a+\sqrt{N}\alpha)}   \left \|
\Pi^{\perp}_{\widetilde{D}_{k}\left(\mathcal{X}^{*} \right)} \left( \widetilde{D}_{k}\left( \Delta_{\mathcal{X}}   \right) \right) \right\|_{\circledast} +\frac{\lambda_h(1+a)}{a}   \left \|
\left( \widetilde{D}_{k}\left(\Delta_{\mathcal{X}}  \right) \right)_{T_k} \right\|_{\ell_1}  -\frac{\lambda_h(1+a)}{a+\alpha}   \left \|
\left( \widetilde{D}_{k}\left( \Delta_{\mathcal{X}}   \right) \right)_{T_k^{C}} \right\|_{\ell_1} \right)  \label{lem3:2} \\
&\le  \sum_{k=1}^{3} \left(   \sqrt{ \frac{\pi_{\min}^{-1}}{9n_k}  \sum_{i,j}  \lvert \left \langle  \eta^{\mathcal{Y} } (id_k(i,j)),   \boldsymbol{1}_{n_k}   \right \rangle \rvert^2  } \cdot \left \|  \Delta_{\mathcal{X}}  \right\|_{F(\pi)} +   \Big (\frac{1+a}{a+\sqrt{N}\alpha}  + \frac{1+a}{a}\Big) \lambda_g   
\left \| \Pi_{\widetilde{D}_{k}\left(\mathcal{X}^{*} \right)} \left( \widetilde{D}_{k}\left( \Delta_{\mathcal{X}}  \right) \right) \right\|_{\circledast}  \right. \notag \\
 & \left. + \frac{\lambda_h (1+a)}{a} \left \| \left( \widetilde{D}_{k}\left( \Delta_{\mathcal{X}}  \right) \right)_{T_k} \right\|_{\ell_1 } - \frac{\lambda_h(1+a)}{a+ \alpha}  \left \|
\left( \widetilde{D}_{k}\left( \Delta_{\mathcal{X}}   \right) \right)_{T_k^{C}} \right\|_{\ell_1} \right)  \notag \\
&\le  \sum_{k=1}^{3}\left( \sqrt{ \frac{\pi_{\min}^{-1}}{9n_k}  \sum_{i,j}  \lvert \left \langle  \eta^{\mathcal{Y} } (id_k(i,j)),   \boldsymbol{1}_{n_k}   \right \rangle \rvert^2  } \cdot \left \|  \Delta_{\mathcal{X}}   \right\|_{F(\pi)} +    \Big (\frac{1+a}{a+\sqrt{N}\alpha}  + \frac{1+a}{a}\Big)\lambda_g   
\sqrt{\frac{\widetilde{r} }{\ell} }\left \|  \left( \widetilde{D}_{k}\left( \Delta_{\mathcal{X}} \right) \right) \right\|_{F}  \right. \notag \\
 & \left.+\frac{\lambda_h(1+a)}{a}  \left \| \left( \widetilde{D}_{k}\left( \Delta_{\mathcal{X}}  \right) \right)_{T_k} \right\|_{\ell_1 } - \frac{\lambda_h(1+a) }{a+\alpha} \left \|
\left( \widetilde{D}_{k}\left( \Delta_{\mathcal{X}}  \right) \right)_{T_k^{C}} \right\|_{\ell_1} \right)  \notag \\
&\le  \sum_{k=1}^{3}\left( \sqrt{ \frac{\pi_{\min}^{-1}}{9n_k}  \sum_{i,j}  \lvert \left \langle  \eta^{\mathcal{Y} } (id_k(i,j)),   \boldsymbol{1}_{n_k}   \right \rangle \rvert^2  } \cdot \left \| \Delta_{\mathcal{X}}  \right\|_{F(\pi)} +    \Big (\frac{1+a}{a+\sqrt{N}\alpha}  + \frac{1+a}{a}\Big) \lambda_g   
\sqrt{\frac{\widetilde{r} }{\ell} }\left \|  \left( \widetilde{D}_{k}\left( \Delta_{\mathcal{X}}  \right) \right)_{T_k}\right\|_{F}  \right. \notag \\
 & \left. + \frac{\lambda_h(1+a)}{a}  \left \| \left( \widetilde{D}_{k}\left( \Delta_{\mathcal{X}}  \right) \right)_{T_k} \right\|_{\ell_1 } +    \right.\\
 & \left. \Big (\frac{1+a}{a+\sqrt{N}\alpha}  + \frac{1+a}{a}\Big)\lambda_g 
\sqrt{\frac{\widetilde{r} }{\ell} }\left \|  \left( \widetilde{D}_{k}\left( \Delta_{\mathcal{X}}  \right) \right)_{T_k^{C}}\right\|_{F} -\frac{\lambda_h(1+a)}{a+\alpha}   \left \|
\left( \widetilde{D}_{k}\left( \Delta_{\mathcal{X}}   \right) \right)_{T_k^{C}} \right\|_{F} \right)  \notag \\
&\le  \sum_{k=1}^{3}\left( \sqrt{ \frac{\pi_{\min}^{-1}}{9n_k}  \sum_{i,j}  \lvert \left \langle  \eta^{\mathcal{Y} } (id_k(i,j)),   \boldsymbol{1}_{n_k}   \right \rangle \rvert^2  } \cdot \left \|\Delta_{\mathcal{X}} \right\|_{F(\pi)} +   (1+\epsilon)  \Big (\frac{1+a}{a+\sqrt{N}\alpha}  + \frac{1+a}{a}\Big)\lambda_g   
\sqrt{\frac{\widetilde{r}|T_k| }{\ell N} }\left \|   \Delta_{\mathcal{X}}    \right\|_{F}  \right. \notag \\
 & \left. + (1+\epsilon) \frac{\lambda_h(1+a)}{a}\frac{|T_k|}{\sqrt{N}} \left \| \Delta_{\mathcal{X}}  \right\|_{F} \right)  \label{lem10:eq1}  \le  \sum_{k=1}^{3}\left(  \sqrt{ \frac{\pi_{\min}^{-1}}{9n_k}  \sum_{i,j}  \lvert \left \langle  \eta^{\mathcal{Y} } (id_k(i,j)),   \boldsymbol{1}_{n_k}   \right \rangle \rvert^2  }  \right.\\
&\left. +(1+\epsilon)  \Big (\frac{1+a}{a+\sqrt{N}\alpha}  + \frac{1+a}{a}\Big)\lambda_g \sqrt{\frac{\mu \widetilde{r}|T_k|}{\ell}}+ \frac{(1+\epsilon) (1+a) \lambda_h\sqrt{\mu}|T_k|}{a}  \right) \left\| \Delta_{\mathcal{X}}  \right \|_{F(\pi)} \\
&:=\Xi_1^{a} \cdot  \left\| \Delta_{\mathcal{X}}  \right \|_{F(\pi)}, \notag 
\end{align}
where the inequality \eqref{lem10:eq1} uses Lemma~\ref{lem1} and the last inequality uses the fact $\| \Delta_{\mathcal{X}} \|_{F} \le \sqrt{N\mu} \| \Delta_{\mathcal{X}} \|_{F(\pi)}$. 

From the inequality \eqref{lem3:2} and the fact $\frac{ \underline{\sigma}_{\alpha}^{2}}{2n} \sum \limits_{o=1}^{n}    \left( \mathcal{X}_o - \mathcal{X}_{o}^{*} \right)^2 \ge 0$,  we can also have 
\begin{align}
&\sum_{k=1}^{3}   \left \|  \widetilde{D}_{k}\left( \Delta_{\mathcal{X}}  \right) \right\|_{\circledast}   \le \sum_{k=1}^{3}   \left(  \left \| \Pi_{\widetilde{D}_{k}\left(\mathcal{X}^{*} \right)} \left( \widetilde{D}_{k}\left( \Delta_{\mathcal{X}}   \right) \right) \right\|_{\circledast} +  \left \|
\Pi^{\perp}_{\widetilde{D}_{k}\left(\mathcal{X}^{*} \right)} \left( \widetilde{D}_{k}\left( \Delta_{\mathcal{X}}  \right) \right) \right\|_{\circledast}  \right) \\
&\le \sum_{k=1}^{3}   \left(  2 (\frac{2a +\sqrt{N} \alpha}{a}) \left \|
\Pi_{\widetilde{D}_{k}\left(\mathcal{X}^{*} \right)} \left( \widetilde{D}_{k}\left( \Delta_{\mathcal{X}}   \right) \right) \right\|_{\circledast}   +    \frac{2(a+\sqrt{N}\alpha) \lambda_h} {a \lambda_g } \left \|
\left( \widetilde{D}_{k}\left( \Delta_{\mathcal{X}}  \right) \right)_{T_k} \right\|_{\ell_1 }   \right. \notag \\
&\left. - \frac{2(a+\sqrt{N}\alpha) \lambda_h} {a \lambda_g }  \left \|
\left( \widetilde{D}_{k}\left(\Delta_{\mathcal{X}}  \right) \right)_{T_k^C} \right\|_{\ell_1 }  + \frac{2(a + \sqrt{N}\alpha)}{(1+a)\lambda_g}\sqrt{ \frac{4\pi_{\min}^{-1}}{9n_k }  \sum_{i,j}  \lvert \left \langle  \eta^{\mathcal{Y} } (id_k(i,j)),   \boldsymbol{1}_{n_k}   \right \rangle \rvert^2  } \cdot \left \| \Delta_{\mathcal{X}}  \right\|_{F(\pi)} \right)  \notag \\
&\le \sum_{k=1}^{3}   \left(  2 (\frac{2a +\sqrt{N} \alpha}{a}) \sqrt{\frac{\widetilde{r}}{\ell}}  \left \|
\widetilde{D}_{k}\left( \Delta_{\mathcal{X}}   \right)  \right\|_{F}   +    \frac{2(a+\sqrt{N}\alpha) \lambda_h} {a \lambda_g } \left \|
\left( \widetilde{D}_{k}\left( \Delta_{\mathcal{X}}  \right) \right)_{T_k} \right\|_{\ell_1 }   \right. \notag \\
&\left. - \frac{2(a+\sqrt{N}\alpha) \lambda_h} {a \lambda_g }  \left \|
\left( \widetilde{D}_{k}\left(\Delta_{\mathcal{X}}  \right) \right)_{T_k^C} \right\|_{F }  + \frac{2(a + \sqrt{N}\alpha)}{(1+a)\lambda_g}\sqrt{ \frac{4\pi_{\min}^{-1}}{9n_k }  \sum_{i,j}  \lvert \left \langle  \eta^{\mathcal{Y} } (id_k(i,j)),   \boldsymbol{1}_{n_k}   \right \rangle \rvert^2  } \cdot \left \| \Delta_{\mathcal{X}}  \right\|_{F(\pi)} \right)  \notag \\
&\le \sum_{k=1}^{3}   \left( 2 (\frac{a + \sqrt{N} \alpha}{a}) \sqrt{\frac{\widetilde{r}}{\ell}} \left \| \left(
 \widetilde{D}_{k}\left( \Delta_{\mathcal{X}}   \right)  \right)_{T_k}\right\|_{F}   + \frac{2(a+\sqrt{N}\alpha)\lambda_h}{a\lambda_g} \sqrt{|T_k|}\left \|
\left( \widetilde{D}_{k}\left( \Delta_{\mathcal{X}}  \right) \right)_{T_k} \right\|_{F }  \right. \notag \\
&\left. + \frac{2(a + \sqrt{N}\alpha)}{(1+a)\lambda_g}\sqrt{ \frac{4\pi_{\min}^{-1}}{9n_k }  \sum_{i,j}  \lvert \left \langle  \eta^{\mathcal{Y} } (id_k(i,j)),   \boldsymbol{1}_{n_k}   \right \rangle \rvert^2  } \cdot \left \| \Delta_{\mathcal{X}}  \right\|_{F(\pi)} \right)  \notag \\
&\le \sum_{k=1}^{3}   \left( 2(1+\epsilon) (\frac{a + \sqrt{N} \alpha}{a}) \sqrt{\frac{\mu \widetilde{r}|T_k|}{\ell}}    + \frac{2(1+\epsilon)(a+\sqrt{N}\alpha)\lambda_h|T_k|\sqrt{\mu}}{a\lambda_g}  \right. \notag \\
&\left. + \frac{2(a + \sqrt{N}\alpha)}{(1+a)\lambda_g}\sqrt{ \frac{4\pi_{\min}^{-1}}{9n_k }  \sum_{i,j}  \lvert \left \langle  \eta^{\mathcal{Y} } (id_k(i,j)),   \boldsymbol{1}_{n_k}   \right \rangle \rvert^2  }  \right) \cdot \left \| \Delta_{\mathcal{X}}  \right\|_{F(\pi)} \notag \\
&\le \frac{2(a+\sqrt{N}\alpha)}{(1+a)\lambda_g} \Xi_1^a  \cdot \left \|  \Delta_{\mathcal{X}}   \right\|_{F(\pi)}.
\end{align}
which completes the  proof.
\end{proof}

Denote $\frac{1}{n}\left \|  \Omega (\mathcal{X} - \mathcal{X}^{*}) \right \|_{F}^{2} := \frac{1}{n} \sum_{o=1}^{n} ( \mathcal{X}_{o} - \mathcal{X}_{o}^{*})^2$, and $\Delta_{\mathcal{X}} = \mathcal{X} - \mathcal{X}^{*}$. It is easy to show that 
$$\mathbb{E}\left (\frac{1}{n}\left \|  \Omega (\Delta_{\mathcal{X}}) \right \|_{F}^{2} \right) = \left \|  \Delta_{\mathcal{X}} \right \|_{F(\pi)}^{2}.$$
Let us introduce two sets:
\begin{equation*}
\begin{split} 
\mathcal{C}_{\alpha}\left( \widetilde{r}, \boldsymbol{T}, \beta \right):= \left\{ \mathcal{X} \in \mathbb{B}(\alpha) : 
\sum_{k=1}^{3}   \left \|\widetilde{D}_{k}\left( \Delta_{\mathcal{X}}  \right) \right\|_{\circledast}   \le 
\frac{2(a+\sqrt{N}\alpha)}{(1+a)\lambda_g} \Xi_1^a  \cdot   \left \|  \Delta_{\mathcal{X}}   \right\|_{F(\pi)},~~\left\| \Delta_{\mathcal{X}}\right \|_{F(\pi)}^2 \ge \beta
\right\}.
\end{split}
\end{equation*}
\begin{equation*}
\mathcal{C}_{\alpha}\left(\widetilde{r}, \boldsymbol{T}, \beta, \theta \right) := \left\{ \mathcal{X} \in \mathcal{C}_{\alpha}\left( \widetilde{r}, \boldsymbol{T}, \beta \right) :  \left\| \Delta_{\mathcal{X}}\right\|_{F(\pi)}^2  \le \theta \right\}.
\end{equation*}
Then, we define the following random variable:
$$\boldsymbol{Z}_{\theta}= \underset{ \mathcal{X} \in \mathcal{C}_{\alpha}\left(\widetilde{r}, \boldsymbol{T}, \beta, \theta \right)   }{\sup} \left \lvert  \frac{1}{n} \left \| \Omega(\Delta_{\mathcal{X}})  \right\|_{F}^{2} -  \left \| \Delta_{\mathcal{X}}\right \|_{F(\pi)}^2 \right \rvert.$$

\begin{lemma}[TDTV$_a$ model]
\label{lem10b}
Let $\Omega$ be the sampling set with distribution $\pi$, which satisfies the assumptions H1 and H2. Assume $\lambda_g \ge \max_{k} \frac{2}{3} \frac{a+\alpha \sqrt{N}}{1+a} \left \| \nabla \Phi_{\mathcal{Y}} \left (\mathcal{ X}^* \right ) \times_k  (\boldsymbol{D}_{k}^{\dagger})^T \right \|$ and $\lambda_h \ge (\frac{a+\alpha}{a+\alpha \sqrt{N}} + \frac{a+\alpha}{a}) \lambda_g \sqrt{\frac{\widetilde{r}}{\ell}}$. 
Denote $\Sigma_{\xi}:= \frac{1}{n} \sum_{o=1}^{n}  \xi_o \mathcal{E}_{o} $. Then,  it holds that
$$\mathbb{P}\left( \boldsymbol{Z}_{\theta} > \frac{5}{12} \theta +   \Xi^a \right) \le \exp( - c n \theta^2),$$
where $c=\frac{1}{128}$ and 
\begin{equation}
\label{Xia}
\Xi^a = 44\alpha^2\left[  \sum_{k=1}^{3}   \sqrt{ \frac{ \pi_{\min}^{-1}}{9 n_k}  \cdot \sum_{i,j} \mathbb{E}\left[ \left |   \left \langle  \Sigma_{\xi}(id_k(i,j)), \boldsymbol{1}_{n_k} \right \rangle\right |^2 \right]}  
 +  \frac{\mathbb{E} \left[  \max_k  \left \|  \Sigma_{\xi}  \times_k( \boldsymbol{D}_{k}^{\dagger})^{T}\right \|  \right] }{3}  \frac{2(a+\sqrt{N} \alpha) }{ (1+a)\lambda_g} \Xi_1^a \right ]^2.
\end{equation}
\end{lemma}
\begin{proof}
Let us first control the upper bound of $\mathbb{E}(\boldsymbol{Z}_{\theta})$.
\begin{align}
\mathbb{E}(\boldsymbol{Z}_{\theta}) &= \mathbb{E}\left( \underset{ \mathcal{X} \in \mathcal{C}_{\alpha}\left(\widetilde{r}, \boldsymbol{T}, \beta, \theta \right)   }{\sup} \left \lvert  \frac{1}{n} \left \| \Omega(\Delta_{\mathcal{X}})  \right\|_{F}^{2} -  \left \| \Delta_{\mathcal{X}} \right \|_{F(\pi)}^2 \right \rvert \right) \notag \\
& \leq 2 \mathbb{E} \left( \underset{ \mathcal{X} \in \mathcal{C}_{\alpha}\left(\widetilde{r}, \boldsymbol{T}, \beta, \theta \right)   }{\sup} \left \lvert  \frac{1}{n} \sum_{o=1}^{n}  \xi_o (\mathcal{X}_{o}- \mathcal{X}_{o}^{*}) ^{2}   \right \rvert \right) \notag \\
&  =  8 \alpha^2  \mathbb{E} \left( \underset{ \mathcal{X} \in \mathcal{C}_{\alpha}\left(\widetilde{r}, \boldsymbol{T}, \beta, \theta \right)   }{\sup} \left \lvert  \frac{1}{n} \sum_{o=1}^{n}  \xi_o ( \frac{\mathcal{X}_{o}- \mathcal{X}_{o}^{*}}{2\alpha})^{2}   \right \rvert \right), \notag 
\end{align}
where $\{ \xi_o\}_{o=1}^{n}$ is an i.i.d. Rademacher sequence. The box constraint $\| \mathcal{X} \|_{\infty} \le \alpha $ and $\| \mathcal{X}^*\|_{\infty} \le \alpha $ implies that $\lvert \mathcal{X}_o - \mathcal{X}_{o}^{*} \rvert  \le 2 \alpha$, i.e., $\lvert  \frac{\mathcal{X}_o - \mathcal{X}_{o}^{*} }{2\alpha}\rvert \le 1$. Then, using the contraction inequality (\cite{ledoux2013probability}, Theorem 4.12) yields
\begin{align}
\mathbb{E}(\boldsymbol{Z}_{\theta})  &\le 8\alpha \mathbb{E} \left ( \underset{ \mathcal{X} \in \mathcal{C}_{\alpha}\left(\widetilde{r}, \boldsymbol{T}, \beta, \theta \right)   }{\sup} \left \lvert  \frac{1}{n} \sum_{o=1}^{n}  \xi_o ( \mathcal{X}_{o}- \mathcal{X}_{o}^{*} )\right \rvert  \right) \notag \\  
&=8\alpha \mathbb{E} \left ( \underset{ \mathcal{X} \in \mathcal{C}_{\alpha}\left(\widetilde{r}, \boldsymbol{T}, \beta, \theta \right)   }{\sup} \left \lvert  \left \langle  \frac{1}{n} \sum_{o=1}^{n}  \xi_o \mathcal{E}_{o},  \mathcal{X} -\mathcal{X}^* \right \rangle\right \rvert  \right) \notag  \\ 
&=8\alpha \mathbb{E} \left ( \underset{ \mathcal{X} \in \mathcal{C}_{\alpha}\left(\widetilde{r}, \boldsymbol{T}, \beta, \theta \right)   }{\sup} \left \lvert  \left \langle  \frac{1}{n} \sum_{o=1}^{n}  \xi_o \mathcal{E}_{o},  \Delta_{\mathcal{X}} \right \rangle\right \rvert  \right)  \notag \\
&=8\alpha \mathbb{E} \left ( \underset{ \mathcal{X} \in \mathcal{C}_{\alpha}\left(\widetilde{r}, \boldsymbol{T}, \beta, \theta \right)   }{\sup} \left \lvert  \left \langle \Sigma_{\xi},  \Delta_{\mathcal{X}} \right \rangle\right \rvert  \right)   \label{lem4:1}
\end{align}
 From Lemma~\ref{lem3}, we have that
\begin{align}
\left \langle \Sigma_{\xi},  \Delta_{\mathcal{X}} \right \rangle  &\le \frac{1}{3} \sum_{k=1}^{3}  \left(   \sqrt{ \frac{ \pi_{\min}^{-1}}{n_k}  \cdot \sum_{i,j} \left |   \left \langle  \Sigma_{\xi}(id_k(i,j)), \boldsymbol{1}_{n_k} \right \rangle\right |^2} \cdot  \left \|  \Delta_{\mathcal{X}} \right \|_{F(\pi)} +   \left \|  \Sigma_{\xi} \times_k (\boldsymbol{D}_{k}^{\dagger})^{T} \right \|   \left \|   \Delta_{\mathcal{X}} \times_k \boldsymbol{D}_{k} \right \|_{\circledast}  \right) \notag \\
& =  \sum_{k=1}^{3}   \sqrt{ \frac{ \pi_{\min}^{-1}}{9 n_k}  \cdot \sum_{i,j} \left |   \left \langle  \Sigma_{\xi}(id_k(i,j)), \boldsymbol{1}_{n_k} \right \rangle\right |^2} \cdot  \left \|  \Delta_{\mathcal{X}} \right \|_{F(\pi)} +   \sum_{k=1}^{3} \frac{1}{3} \left \|  \Sigma_{\xi} \times_k (\boldsymbol{D}_{k}^{\dagger})^{T} \right \|   \left \|   \Delta_{\mathcal{X}} \times_k \boldsymbol{D}_{k} \right \|_{\circledast}  \notag \\
&\le  \sum_{k=1}^{3}   \sqrt{ \frac{ \pi_{\min}^{-1}}{9 n_k}  \cdot \sum_{i,j} \left |   \left \langle  \Sigma_{\xi}(id_k(i,j)), \boldsymbol{1}_{n_k} \right \rangle\right |^2} \cdot  \left \|  \Delta_{\mathcal{X}} \right \|_{F(\pi)} +  \frac{\max_k \left \|  \Sigma_{\xi} \times_k (\boldsymbol{D}_{k}^{\dagger})^{T} \right \|  }{3}    \sum_{k=1}^{3}   \left \|   \Delta_{\mathcal{X}} \times_k \boldsymbol{D}_{k} \right \|_{\circledast}   \label{lem4:2}
\end{align}
Since $\mathcal{X} \in \mathcal{C}_{\alpha}\left(\widetilde{r}, \boldsymbol{T}, \beta, \theta \right) $,  we have 
\begin{align*}
\sum_{k=1}^{3}   \left \|   \Delta_{\mathcal{X}} \times_k \boldsymbol{D}_{k} \right \|_{\circledast}= \sum_{k=1}^{3}  \left \|  \widetilde{D}_{k}\left(  \Delta_{\mathcal{X}}  \right)  \right \|_{\circledast} \le  \frac{2(a+\sqrt{N}\alpha)}{(1+a)\lambda_g} \Xi_1^a  \cdot   \left \|  \Delta_{\mathcal{X}}   \right\|_{F(\pi)}
\end{align*}
Substituting the above inequality into \eqref{lem4:2}, we can obtain
\begin{align}
\left \langle \Sigma_{\xi},  \Delta_{\mathcal{X}} \right \rangle  \le \left[  \sum_{k=1}^{3}   \sqrt{ \frac{ \pi_{\min}^{-1}}{9 n_k}  \cdot \sum_{i,j} \left |   \left \langle  \Sigma_{\xi}(id_k(i,j)), \boldsymbol{1}_{n_k} \right \rangle\right |^2} +  \frac{\max_k \left \|   \Sigma_{\xi} \times_k (\boldsymbol{D}_{k}^{\dagger})^{T} \right \| }{3} \frac{2(a+\sqrt{N}\alpha)}{(1+a)\lambda_g} \Xi_1^a    \right ]   \left \|  \Delta_{\mathcal{X}}   \right\|_{F(\pi)} \label{lem4:3}
\end{align}
Now,  combining \eqref{lem4:1} and \eqref{lem4:3}, we can have 
\begin{align}
\mathbb{E} \left (  \boldsymbol{Z}_{\theta} \right)  & \le 8\alpha \left[  \sum_{k=1}^{3}   \sqrt{ \frac{ \pi_{\min}^{-1}}{9 n_k}  \cdot \sum_{i,j} \mathbb{E}\left[ \left |   \left \langle  \Sigma_{\xi}(id_k(i,j)), \boldsymbol{1}_{n_k} \right \rangle\right |^2 \right]}  
+  \frac{\mathbb{E} \left[  \max_k  \left \|  \Sigma_{\xi}  \times_k (\boldsymbol{D}_{k}^{\dagger})^{T}\right \|  \right] }{3}   \frac{2(a+\sqrt{N}\alpha)}{(1+a)\lambda_g} \Xi_1^a  \right ]   \left \|  \Delta_{\mathcal{X}}   \right\|_{F(\pi)}  \notag   \\
& =: 8\alpha  \cdot \Xi_2^a \cdot \sqrt{\theta}. \label{lem4:4}
\end{align}

Then, recalling the Massart's concentration inequality (see, e.g.\cite{buhlmann2011statistics}, Theorem 14,2) 
\begin{equation*}
\mathbb{P}\left( \boldsymbol{Z}_{\theta} \ge \mathbb{E}(\boldsymbol{Z}_{\theta} )+ \frac{1}{9} \cdot \left (\frac{5}{12}\theta \right ) \right)  \le \exp(- \frac{1}{128} n \theta^2), 
\end{equation*}
it follows that 
\begin{equation*}
\begin{split}
 \mathbb{E}(\boldsymbol{Z}_{\theta} ) +  \frac{1}{9} \cdot \left (\frac{5}{12}\theta \right ) &\le 8\alpha \Xi_2^a \sqrt{\theta} +    \frac{1}{9} \cdot \left (\frac{5}{12}\theta \right ) \le  (\frac{1}{9} +\frac{8}{9}  )\cdot \frac{5}{12}\theta  + 44 \alpha^2 (\Xi_2^a)^2  \\
& = \frac{5}{12}\theta   +  \Xi^a , 
\end{split}
\end{equation*}
which completes the proof.
\end{proof}

\begin{lemma}[TDTV$_a$ model]
\label{lem11b}
Let $\Omega$ be the sampling set with distribution $\pi$, which satisfies the assumptions H1 and H2. 
Assume that $\lambda_g \ge \max_{k} \frac{2}{3} \frac{a+\alpha \sqrt{N}}{1+a} \left \| \nabla \Phi_{\mathcal{Y}} \left (\mathcal{ X}^* \right ) \times_k  (D_{k}^{\dagger})^T \right \|$ and $\lambda_h \ge (\frac{a+\alpha}{a+\alpha \sqrt{N}} + \frac{a+\alpha}{a}) \lambda_g \sqrt{\frac{\widetilde{r}}{\ell}}$.   Denote $\Sigma_{\xi}:= \frac{1}{n} \sum_{o=1}^{n}  \xi_o \mathcal{E}_{o}$.  For all $\mathcal{X} \in \mathcal{C}_{\alpha}\left( \widetilde{r}, \boldsymbol{T}, \beta \right)$, it holds that 
$$ \frac{1}{n} \left \| \Omega\left(  \Delta_{\mathcal{X} } \right) \right \|_{F}^{2} \ge \frac{1}{2} \left \|  \Delta_{\mathcal{X} } \right \|_{F(\pi)}^{2}  - \Xi^a,$$
with probability at least $1-\frac{2}{d}$ with $d=(n_1+n_2)n_3$, and $\Xi^a$ is defined in $\eqref{Xia}$.
\end{lemma}
\begin{proof}
Let $\nu = \frac{6}{5}$, $\beta=\sqrt{\frac{64 \log(d)}{n \log(6/5)}}$. Consider the following event:
\begin{equation}
\mathcal{B} = \left \{  \exists \mathcal{X} \in \mathcal{C}_{\alpha} \left( \widetilde{r}, \boldsymbol{T}, \beta \right),  ~\mathrm{s.t.}~\left \lvert \frac{1}{n} \left \| \Omega(\Delta_{\mathcal{X}} ) \right \|_{F}^{2}  - \left \| \Delta_{\mathcal{X}} \right \|_{F(\pi)}^{2}   \right  \rvert > \frac{1}{2} \left \| \Delta_{\mathcal{X}} \right \|_{F(\pi)}^2 + \Xi^a \right\}.
\end{equation}
For $t \in \mathbb{N}$, let 
\begin{equation*}
\mathcal{S}_{t} = \left \{ \mathcal{X}  \in  \mathcal{C}_{\alpha} \left( \widetilde{r}, \boldsymbol{T}, \beta \right)~: \nu^{t-1} \beta \le \| \Delta_{\mathcal{X}} \|_{F(\pi)}^{2} \le \nu^{t}\beta  \right\}.
\end{equation*}
If the event $\mathcal{B}$ holds for some tensor $\mathcal{X} \in  \mathcal{C}_{\alpha} \left( \widetilde{r}, \boldsymbol{T}, \beta \right)$, then $\mathcal{X}$ belongs to some $\mathcal{S}_t$ and 
\begin{align}
\left \lvert \frac{1}{n} \left \| \Omega(\Delta_{\mathcal{X}} \right \|_{F}^{2} - \left \| \Delta_{\mathcal{X}} \right \|_{F(\pi)}^2 \right \rvert &> \frac{1}{2} \left \| \Delta_{\mathcal{X}} \right \|_{F(\pi)}^{2} + \Xi^a  \notag \\
&> \frac{1}{2} \nu^{t-1} \beta + \Xi^a\notag \\
&> \frac{5}{12} \nu^{t} \beta + \Xi^a.  \label{lem5:1}
\end{align}
For each $\theta >\beta$, consider the set  $\mathcal{C}_{\alpha}\left(\widetilde{r}, \boldsymbol{T}, \beta, \theta \right) $ and the following event:
\begin{equation*}
\mathcal{B}_{t} = \left \{  \exists \mathcal{X} \in \mathcal{C}_{\alpha}\left(\widetilde{r}, \boldsymbol{T}, \beta, \nu^t\beta \right): ~\left \lvert \frac{1}{n} \left \| \Omega(\Delta_{\mathcal{X}}) \right \|_{F}^{2} - \left \| \Delta_{\mathcal{X}} \right \|_{F(\pi)}^2 \right \rvert > \frac{5}{12} \nu^t \beta + \Xi^a  \right\}.
\end{equation*}
Note that $\mathcal{X} \in \mathcal{S}_t$ implies that $\mathcal{X} \in \mathcal{C}_{\alpha}\left(\widetilde{r}, \boldsymbol{T}, \beta, \nu^t \beta \right)$. Then, $\eqref{lem5:1}$ implies that $\mathcal{B}_t$ holds, and thus  we get $\mathcal{B}  \subset  \mathop{\cup} \limits_t \mathcal{B}_t$. Thus, it is enough to estimate the probability of the simpler event $\mathcal{B}_t$ and then apply the union bound.

Lemma~\ref{lem10b} implies that $\mathbb{P}(\mathcal{B}_t ) \le \exp(-c n \nu^{2t}\beta^2)$. Using the union bound, we obtain
\begin{align*}
\mathbb{P}(\mathcal{B}) \le \sum \limits_{t=1}^{\infty} \mathbb{P}(\mathcal{B}_t) \le \sum_{t=1}^{\infty} \exp(-cn\nu^{2t}\beta^2) \le \sum_{t=1}^{\infty} \exp(-2 c n \cdot \log(v) \beta^2 \cdot t) ,
\end{align*}
where the third inequality uses $\exp(x) \ge x$. We finally compute for $\beta = \sqrt{\frac{64 \log(d) }{ \log(6/5) n}}$:
\begin{align*}
\mathbb{P}(\mathcal{B}) \le \frac{\exp\left(-2cn\log(\nu) \beta^2 \right)}{ 1 - \exp\left(-2cn\log(\nu)\beta^2\right)} = \frac{\exp(-\log(d))}{1 - \exp(-\log(d))} = \frac{d^{-1}}{1- d^{-1}} \le \frac{2}{d}.
\end{align*}
Hence, we complete the proof of Lemma~\ref{lem11b}.
\end{proof}

\begin{lemma}[TDTV$_a$ model]
\label{lem12b}
The assumptions are same as Lemma~\ref{lem11b}.  It holds that 
\begin{equation}
\label{lem12b:eq1}
\left \| \mathcal{X} - \mathcal{X}^{*} \right \|_{F(\pi)}^{2} \le 4 \cdot \Xi^a + \frac{16  \cdot (\Xi_1^a)^2 }{ \underline{\sigma}_{\alpha}^{4} },
\end{equation}
 with probability $1-\frac{2}{d}$, where $\Xi^a$  and $\Xi_1^a$ are defined in $\eqref{Xia}$ and  $\eqref{Xi1a}$, respectively. 
\end{lemma}
\begin{proof}
From lemma~\ref{lem9b} and Lemma~\ref{lem11b}, we can get that 
\begin{align*}
\frac{1}{2} \left \|   \mathcal{X} - \mathcal{X}^{*} \right \|_{F(\pi)}^2 & \le \Xi^a+ \frac{1}{n} \left \| \Omega(\mathcal{X} -\mathcal{X}^{*}) \right \|_{F}^2 \\
&\le \Xi^a + \Xi_1^a \cdot \frac{2}{\underline{\sigma}_{\alpha}^2 } \cdot  \left \| \mathcal{X} -\mathcal{X}^{*} \right \|_{F(\pi)} \\
&\le \Xi^a + \frac{1}{4} \left \| \mathcal{X} -\mathcal{X}^{*} \right \|_{F(\pi)}^2 + \frac{4 \cdot  (\Xi_1^a)^2 }{\underline{\sigma}_{\alpha}^4},
\end{align*}
which implies that 
\begin{equation*}
\left\| \mathcal{X} - \mathcal{X}^{*} \right \|_{F(\pi)}^{2} \le 4 \cdot  \Xi^a +  \frac{16 \cdot  (\Xi_1^a)^2 }{\underline{\sigma}_{\alpha}^4}.
\end{equation*}
Hence, we completed the proof.
\end{proof}

\begin{theorem}[ $\text{TDTV}_{a}(0<a\le \infty)$ model]
\label{thm3}
Assume that H1, H2, H3 and H4 hold, and the sampling number $n$ and the regularization parameters $\lambda_g$, $\lambda_h$  satisfy: 
\begin{align}
&n>\max\left\{  \frac{m n_3 \log(d) \max_t \rho_{t}^{2} }{9\nu \max_t \| \boldsymbol{D}_{t}^{\dagger} \|^2 },  ~~2 \nu^{-1} (\frac{\rho\delta_{\alpha}}{\overline{\sigma}_{\alpha}   \| \boldsymbol{D}_{k}^{\dagger}\|  })^2 \log^2\big (\frac{\rho \delta_{\alpha}\sqrt{\mu mn_3}}{\underline{\sigma}_{\alpha}   \sigma_{\min}(\boldsymbol{D}_{k}^{\dagger}) } \big) \cdot mn_3 \log(d)\right\},   \\
&\lambda_g \ge \max_{k} \frac{2}{3} \frac{a+\alpha \sqrt{N}}{1+a} \left \| \nabla \Phi_{\mathcal{Y}} \left (\mathcal{ X}^* \right ) \times_k  (\boldsymbol{D}_{k}^{\dagger})^T \right \|, ~\text{and}~\lambda_h \ge (\frac{a+\alpha}{a+\alpha \sqrt{N}} + \frac{a+\alpha}{a}) \lambda_g \sqrt{\frac{\widetilde{r}}{\ell}}. 
\end{align}
Then,  the mean squared error (MSE) of etimator $\hat{\mathcal{X}}$ has the following upper bounds:
\begin{equation}\label{thm3:eq1} 
\begin{split}
&\frac{\| \hat{\mathcal{X}} - \mathcal{X}^{*}\|_{F}^{2}}{n_1n_2n_3}   \le \\
&C_1 \left ( \alpha^2 + \frac{\delta_{\alpha}^2}{\sigma_{\alpha}^4} + \frac{\alpha^2 \delta_{\alpha}^2(a + \alpha \sqrt{N})^2}{(1+a)^2}  \frac{\mathbb{E}\left[  \max_k  \left \|   \Sigma_{\xi} \times_k (\boldsymbol{D}_{k}^{\dagger})^{T}\right \|  \right]^2}{\lambda_g^2}  \right) \cdot \frac{\mu^2 (n_1n_2 +n_2n_3 + n_1n_3) \log(2n_1n_2n_3)}{n}  \\
&+ C_2 \left(  \frac{\lambda_g^2 }{\sigma_{\alpha}^{4}} \Big( \frac{1+a}{a +\alpha \sqrt{N}}+ \frac{1+a}{a} \Big)^2+ \alpha^2   \Big( \frac{2a + \alpha \sqrt{N}}{a}\Big)^2 \mathbb{E} \left[  \max_k  \left \|  \Sigma_{\xi} \times_k(\boldsymbol{D}_{k}^{\dagger})^{T} \right \|  \right]^2 \right) \cdot \frac{ \mu^2 \widetilde{r} \max_{k} |T_k|}{\ell} \\
&+C_3\left(   \frac{1}{\sigma_{\alpha}^4} \Big( \frac{1+a}{a}\Big)^2+\frac{\alpha^2 (a + \alpha \sqrt{N})^2}{a^2} \frac{ \mathbb{E} \left[  \max_k  \left \|  \Sigma_{\xi} \times_k (\boldsymbol{D}_{k}^{\dagger})^{T}\right \|  \right]^2}{\lambda_g^2}\right) \cdot \mu^2 \lambda_{h}^2 \max_{k} |T_k|^2,
\end{split}
\end{equation}
Moreover, if $\lambda_g$,  $\lambda_h$ and $a$ are specified as:
\begin{align} 
\lambda_g = \frac{2 c_{\alpha}}{3}  \frac{a+\alpha \sqrt{N}}{1+a} \overline{\sigma}_{\alpha} \cdot \max_{k} \| \boldsymbol{D}_{k}^{\dagger}\|  \cdot  \sqrt{\frac{2\nu \ell \log(d)}{nmn_3}}, ~\lambda_h =  (\frac{a+\alpha}{a+\alpha \sqrt{N}} + \frac{a+\alpha}{a})\lambda_g \sqrt{\frac{\widetilde{r}}{\ell}},~ a^{-1} = \mathcal{O}\left( (\alpha \sqrt{N})^{-1} \right), \notag
\end{align}
where  $c_{\alpha}$ depends only on $\delta_{\alpha}$.  Then, the upper bound reduces to 
\begin{equation}
\label{thm3:eq2} 
\begin{split}
\frac{\| \hat{\mathcal{X}} - \mathcal{X}^{*}\|_{F}^{2}}{n_1n_2n_3}  &\le  C_1(\alpha^2+ \frac{\alpha^2\delta_{\alpha}^2}{c_{\alpha}^{2}\overline{\sigma}_{\alpha}^{2} } + \frac{\delta_{\alpha}^{2} }{\underline{\sigma}_{\alpha}^{4} } )  \cdot \frac{\mu^2 (n_1n_2 +n_2n_3 + n_1n_3) \log(2n_1n_2n_3)}{n}  \\
&+C_2 (\alpha^2 + \frac{c_{\alpha}^{2}\overline{\sigma}_{\alpha}^{2}}{\underline{\sigma}_{\alpha}^{4}}) c_h^2  \cdot \frac{\nu \mu^2  \left( \widetilde{r} \max_{k} |T_k|^2  \log((n_1+n_2)n_3)\right)}{n} \cdot\frac{ \max_k \| \boldsymbol{D}_{k}^{\dagger} \|^2 }{mn_3} ,
\end{split}
\end{equation}
where $C_1$ and $C_2$  are both absolute constants. 
\end{theorem}
\begin{proof} Case (i): For  TDTV$_a(0<a<\infty)$ model,  substituting $\Xi^{a}$ $\eqref{Xia}$ and $\Xi_1^a$  $\eqref{Xi1a}$ into \eqref{lem12b:eq1} in Lemma~\ref{lem12b} and further combining Lemmas~\eqref{lem4}, \eqref{lem6} and \eqref{lem7},  we can deduce the conclusion \eqref{thm3:eq1} and \eqref{thm3:eq2}. 

Case (ii): For TDTV$_{\infty}$ (or TDTV) model, based on the relationship between Lemma~\ref{lem8} and Lemma \ref{lem8a},   we can deduce the conclusion \eqref{thm3:eq1} and \eqref{thm3:eq2} with $a \rightarrow \infty$. 
\end{proof}

\newpage 
\section{Proof of Theorem~2}
The following lemma is necessary to construct a suitable packing set for the proof of Theorem~2.
\begin{lemma}
\label{thm2:lem1}.
	Without loss of generalization, assume that $n_1 \leq n_2$. Let $\gamma \in (0, 1/8)$, the tubal-rank $r_t \le r$, and define
    \begin{equation}
    \label{kappa}
    \kappa := \min\left( 1/2,  \frac{\sqrt{\gamma r_t n_2 n_3}}{2\sqrt{2n} \overline{\sigma}_{\alpha} \alpha} \right ).
    \end{equation}
There exists a subset $\mathbb{T} \subseteq  \mathbb{K} (r,  \alpha)$ with cardinality:
	\begin{equation}
	\label{T_set}
		|\mathbb{T}| \geq 2^{\frac{r_t n_2 n_3}{16}+1}
	\end{equation}
	that satisfies the following properties:\\
	(i)~For any $\mathcal{X} \in \mathbb{T}$, each entry $x_{i,j,k}=\kappa\alpha$ or $0$, \\
	(ii)~For any two distinct $\mathcal{X}^{j}$, $\mathcal{X}^{k}\in \mathbb{T},~j\neq k$, it holds that
	\begin{displaymath}
		\| \mathcal{X}^{j} - \mathcal{X}^{k} \|_{F}^{2} \geq \frac{\kappa^2 \alpha^2 n_1n_2n_3}{8}.
	\end{displaymath}
\end{lemma}
\begin{proof}
A probabilistic argument~\cite{Davenport2014a, Hou2021a} is utilized for the proof process.  We consider a tensor of size $r_t\times n_2 \times n_3$, denotes as $\mathcal{S}=\big(\kappa \alpha \cdot\varepsilon_{i,j,k}\big)$,  where all $\varepsilon_{i,j,k}s$ are i.i.d. Bernoulli random variables taking value $1$ or $0$.  
Let us denote 
$$
\text{Pad}(\mathcal{U}, n_1) := 
\begin{bmatrix}
\mathcal{U} \\
0 
\end{bmatrix} \in \mathbb{R}^{n_1 \times n_2 \times n_3},
$$
which means that padding tensor $\mathcal{U}$ with zeros such that the padded tensor  is a large tensor of size $n_1 \times n_2 \times n_3$. 
Let us  introduce a new notation $\text{Copy}(\mathcal{A}, k)$, which means that it will copy the tensor $\mathcal{A} \in \mathbb{R}^{r_t \times n_2 \times n_3}$ with $k$ times to form a new tensor of size $k r_t  \times n_2 \times n_3$.
Then, a tensor set is constructed as:
\begin{displaymath}
\mathbb{T} = \Big\{ \mathcal{X} = \text{Pad}\big(\text{Copy}(\mathcal{S}, [\frac{n_1}{r_t}]), n_1\big): \mathcal{S} \in \mathbb{R}^{r_t \times n_2 \times n_3}~\text{is a random tensor} \Big\},
\end{displaymath}
where each entry $\mathcal{X} \in \mathbb{R}^{n_1 \times n_2 \times n_3}$. By generating many random tensors $\mathcal{S}$ by drawing the samples from Bernoulli distribution,  we can generate the desired set $\mathbb{T}$  such that $|\mathbb{T}| = \text{ceil}(2^{\frac{r_t n_2 n_3}{8}+1})$. 

It is easy to check that, with non-zero probability, this set $\mathbb{T}$  will have the two desired properties. In fact, (i) for any $\mathcal{X} \in \mathbb{T}$,  $\|\mathcal{X}\|_{\infty} = \kappa \alpha \leq \alpha$.  Then, combining the fact that $rank_{t}(\mathcal{X}) \leq r$  leads to  the consequence  $\mathbb{T} \subseteq  \mathbb{K}(r, \alpha)$.

It remains to verify that $\mathcal{X}$ satisfies the requirement (ii). For any $\mathcal{X}^1,\mathcal{X}^2 \in \mathbb{T}$, we get
	\begin{displaymath}
		\begin{aligned}
			\| \mathcal{X}^1 - \mathcal{X}^2 \|_{F}^{2} &= \sum_{i,j,k}(\mathcal{X}_{i,j,k}^{1} - \mathcal{X}_{i,j,k}^{2})^2 \geq \lfloor \frac{n_1}{r_t} \rfloor \sum_{i\in [r_t]} \sum_{j \in [n_2]} \sum_{k \in [n_3]} (\mathcal{X}_{i,j,k}^{1} - \mathcal{X}_{i,j,k}^{2})^2 \\
			&=  \kappa^2\alpha^2  \lfloor \frac{n_1}{r_t} \rfloor  \sum_{i\in [r_t]} \sum_{j \in [n_2]} \sum_{k\in [n_3]} \delta_{i,j,k} =: \kappa^2 \alpha^2  \lfloor \frac{n_1}{r_t} \rfloor \mathcal{Z}(\mathcal{X}^1, \mathcal{X}^2),
		\end{aligned}
	\end{displaymath}
	where $\delta_{i,j,k}$ above are independent $0/1$ Bernoulli random variables with mean $\frac{1}{2}$. Applying Hoeffding's inequality and a union bound leads to
	\begin{displaymath}
		\begin{aligned}
			&\mathbb{P}\Big(\underset{\mathcal{X}^1 \neq \mathcal{X}^2 \in \mathbb{T}}{\min} \mathcal{Z}(\mathcal{X}^1, \mathcal{X}^2) \leq \frac{r_t n_2n_3}{4}\Big)   \leq \left(\begin{array}{c}|\mathbb{T}|\\2\end{array}\right)\exp(-\frac{r_t n_2n_3}{8}) \\
			&\leq \frac{1}{2}  \cdot 2^{\frac{r_t n_2n_3}{16}+1}  \cdot  2^{\frac{r_t n_2n_3}{16}}  \cdot 2^{-\frac{r_tn_2n_3}{8}}\\
			&\leq  2^{\frac{r_t n_2n_3}{8}}  \cdot 2^{-\frac{r_tn_2n_3}{8}} \\
			& = 1.
		\end{aligned}
	\end{displaymath}
	Thus, the left-hand side of the above inequality is less than 1, which indicates that  the event that $\mathcal{Z}(\mathcal{X}^1, \mathcal{X}^2) > \frac{r_tn_2n_3}{4}$ for all $\mathcal{X}^1 \neq \mathcal{X}^2 \in \mathbb{T}$ occurs with a non-zero probability. In this event, we have
	\begin{displaymath}
		\begin{aligned}
			\| \mathcal{X}^1 - \mathcal{X}^2\|_{F}^{2} &\geq \kappa^2 \alpha^2  \lfloor \frac{n_1}{r_t} \rfloor \mathcal{Z}(\mathcal{X}^1, \mathcal{X}^2) \geq \frac{1}{4} \kappa^2\alpha^2  \lfloor \frac{n_1}{r_t} \rfloor r_t n_2n_3\geq \frac{\kappa^2\alpha^2  n_1n_2n_3 }{8},
		\end{aligned}
	\end{displaymath}
	where we use the inequality $\lfloor x\rfloor > \frac{x}{2}$ and the convention $n_1 \leq n_2$. That is to say $\frac{\|\mathcal{X}^1 - \mathcal{X}^2\|_{F}^{2}}{n_1n_2n_3} \geq \frac{\kappa^2\alpha^2 }{8}$, which completed the proof.
\end{proof}
In the following, based on the above lemma, we give a detailed proof for Theorem 2. First, we present the content of Theorem 2. Then, the proof  details are provided.

\begin{theorem}
For  $\alpha>0$, $\gamma \in (0, 1/8)$ and $1 \le r_t \le r$, there exist two constants $c>0$ and $\theta_{\alpha,r}>0$ such that, 
\begin{equation}
\underset{\mathcal{\hat{X}}}{\inf} \underset{\dot{\mathcal{X}} \in\mathbb{K} (r,  \alpha) }{\sup}\mathbb{P}\left( \frac{\| \mathcal{\hat{X}} - \dot{\mathcal{X}}\|_{F}^{2} }{n_1n_2n_3} > c \min\{ \alpha^2, \frac{ \gamma r_tn_2n_3 }{\overline{\sigma}_{\alpha}^{2} n}  \}\right) \ge \theta_{\alpha,r}
\end{equation}
\end{theorem}
\begin{proof}

From Lemma~\ref{thm2:lem1},  we can know that there exists a subset $\mathbb{T} \subseteq  \mathbb{K} (r,  \alpha)$ with cardinality $|\mathbb{T}| \geq 2^{\frac{r_t n_2 n_3}{16}+1}$ containing the null tensor $\mathcal{X}^{0}$, 	that satisfies the following properties:\\
	(i)~For any $\mathcal{X} \in \mathbb{T}$, each entry $x_{i,j,k}=\kappa\alpha$ or $0$, \\
	(ii)~For any two distinct $\mathcal{X}^{1}$, $\mathcal{X}^{2}\in \mathbb{T}$, it holds that
	\begin{equation}
          \label{thm2:eq1}
		\| \mathcal{X}^{1} - \mathcal{X}^{2} \|_{F}^{2} \geq \frac{\kappa^2 \alpha^2 n_1n_2n_3}{8}.
	\end{equation}

For some $\mathcal{X} \in \mathbb{T}$, we now estimate the Kullback-Leiber divergence $D(\mathbb{P}_{\mathcal{X}} || \mathbb{P}_{\mathcal{X}^{0}})$ between probability measures $\mathbb{P}_{\mathcal{X}}$ and $\mathbb{P}_{\mathcal{X}_0}$. By independence of the observations $(\mathcal{Y}_{w_o}, w_{o})_{o=1}^{n}$ and since the distribution of $\mathcal{Y}_{w_o}|w_o$ belongs to the exponential family, one obtains: 
\begin{equation*}
D(\mathbb{P}_{\mathcal{X}} || \mathbb{P}_{\mathcal{X}^{0}}) = n \mathbb{E}_{w_1}[ \Phi'(\mathcal{X}_{w_1})( \mathcal{X}_{w_1} - \mathcal{X}_{w_1}^{0}) + \Phi(\mathcal{X}_{w_1}) - \Phi(\mathcal{X}_{w_1}^{0}) ].
\end{equation*}
Since $\mathcal{X}_{w_1}^{0}=0$ and either $\mathcal{X}_{w_1}=0$ or $\kappa \alpha$, by strong convexity and by definition of $\kappa$, one gets
\begin{equation*}
D(\mathbb{P}_{\mathcal{X}} || \mathbb{P}_{\mathcal{X}^{0}}) \le n \frac{\overline{\sigma}_{\alpha}^2}{2} \kappa^2 \alpha^2 \le \frac{\gamma r_t n_2n_3}{16} \le \gamma \log_{2} ( |\mathbb{T}| -1), 
\end{equation*}
which implies 
\begin{equation}
 \label{thm2:eq2}
\frac{1}{|\mathbb{T}| -1} \sum_{\mathcal{X} \in \mathbb{T} } D(\mathbb{P}_{\mathcal{X}} || \mathbb{P}_{\mathcal{X}^{0}})  \le   \gamma \log_{2} ( |\mathbb{T}| -1).
\end{equation}
Using \eqref{thm2:eq1}, \eqref{thm2:eq2}, \eqref{kappa} and  (Tsybakov  \cite{tsybakov2009}, 2009, Theorem 2.5) together gives
\begin{equation*}
\underset{\mathcal{\hat{X}}}{\inf} \underset{\dot{\mathcal{X}} \in\mathbb{K} (r,  \alpha) }{\sup}\mathbb{P}\left( \frac{\| \mathcal{\hat{X}} - \dot{\mathcal{X}}\|_{F}^{2} }{n_1n_2n_3} > \tilde{c} \min \left \{ \alpha^2, \frac{\gamma  r_t  n_2n_3}{n\overline{\sigma}_{\alpha}^{2}} \right \}\right) \ge \theta_{\alpha,r},
\end{equation*}
where $$\theta_{\alpha,r} = \frac{1}{1+2^{- \frac{rn_2n_3}{32}}} \left(1- 2\gamma -4 \sqrt{\frac{\gamma}{rn_2n_3}}\right),$$
and $ \tilde{c}$ is a numerical constant. Since we are free to choose $\gamma$ as small as possible, this achieves the proof.
\end{proof}

\newpage
\section{Details for operators $\text{Shrink}_{\lambda}^{a}(\cdot)$ and $\text{t-SVT}_{\lambda}^{a}(\cdot)$}
In this section, we will provide the details for the operators $\text{Shrink}_{\lambda}^{a}(\cdot)$ and $\text{t-SVT}_{\lambda}^{a}(\cdot)$.
\begin{proposition}
For a fixed tensor $\mathcal{B}$, consider the following $\text{T}\ell_1$ regularization  problem:
\begin{equation*}
\hat{\mathcal{X}}_{\lambda} =\arg\min \limits_{ \mathcal{X} } \left \{ \frac{1}{2} \| \mathcal{X} - \mathcal{B} \|_F^2  + \lambda \| \mathcal{X} \|_{\text{T}\ell_1} \right \},
\end{equation*}
where $\lambda$ is a regularization parameter.  The optimal solution $\hat{\mathcal{X}}_{\lambda}$ can be efficiently sought by $\text{T}\ell_1$ shrinkage operator:
\begin{align*}
\hat{\mathcal{X}}_{\lambda} = \text{Shrink}_{\lambda}^{a}( \mathcal{B} ),
\end{align*}
where $ \text{Shrink}_{\lambda}^{a}(\cdot)$ is the closed-form thresholding  function (see Theorem III.1. in \cite{Zhangshuai2017a}\cite{zhang2018aTL1}).
\end{proposition}

\begin{proposition}
For a fixed tensor $\mathcal{B}$ with transformed t-SVD $\mathcal{B} = \mathcal{U} *_{\mathcal{L}} \mathcal{S} *_{\mathcal{L}} \mathcal{V}^{H}$, consider tensor $\text{TL}_1$ spectral regularization  problem with respect to inverse transformation $\mathcal{L}$:
\begin{equation*}
\hat{\mathcal{X}}_{\lambda, \mathcal{L}} =\arg\min \limits_{ \mathcal{X} } \left \{ \frac{1}{2} \| \mathcal{X} - \mathcal{B} \|_F^2  + \lambda \| \mathcal{X} \|_{\text{TL}_1,\mathcal{L}} \right \},
\end{equation*}
where $\lambda$ is a regularization parameter.  The optimal solution $\hat{\mathcal{X}}_{\lambda, \mathcal{L}}$ can be efficiently solved by the following generalized tensor-singular-value shrinkage operator:
\begin{align*}
 \hat{\mathcal{X}}_{\lambda, \mathcal{L}} = \text{t-SVT}_{\lambda}^{a}(\mathcal{B}):= \mathcal{U} *_{\mathcal{L}} \mathcal{S}_{\lambda}^{a}  *_{\mathcal{L}} \mathcal{V}^{H},
  \end{align*}
where $\mathcal{S}_{\lambda}^{a}  = \mathcal{L}^{-1} \left( \text{Shrink}_{\lambda}^{a} \left( \mathcal{L}(\mathcal{S} ) \right )  \right)$. 
\end{proposition}
\end{appendices}

\bibliographystyle{IEEEtran}
\bibliography{bibTC}

\begin{thebibliography}{10}
\providecommand{\url}[1]{#1}
\csname url@samestyle\endcsname
\providecommand{\newblock}{\relax}
\providecommand{\bibinfo}[2]{#2}
\providecommand{\BIBentrySTDinterwordspacing}{\spaceskip=0pt\relax}
\providecommand{\BIBentryALTinterwordstretchfactor}{4}
\providecommand{\BIBentryALTinterwordspacing}{\spaceskip=\fontdimen2\font plus
\BIBentryALTinterwordstretchfactor\fontdimen3\font minus
  \fontdimen4\font\relax}
\providecommand{\BIBforeignlanguage}[2]{{%
\expandafter\ifx\csname l@#1\endcsname\relax
\typeout{** WARNING: IEEEtran.bst: No hyphenation pattern has been}%
\typeout{** loaded for the language `#1'. Using the pattern for}%
\typeout{** the default language instead.}%
\else
\language=\csname l@#1\endcsname
\fi
#2}}
\providecommand{\BIBdecl}{\relax}
\BIBdecl

\bibitem{Bi2018}
X.~Bi, A.~Qu, and X.~Shen, ``Multilayer tensor factorization with applications
  to recommender systems,'' \emph{Ann. of Statist.}, vol.~46, no.~6B, pp.
  3308--3333, 2018.

\bibitem{Nickel2011}
M.~Nickel, V.~Tresp, and H.~P. Kriegel, ``A three-way model for collective
  learning on multi-relational data,'' \emph{Int. Conf. Machine Learning},
  vol.~11, pp. 809--816, 2011.

\bibitem{Papalexakis2016}
\BIBentryALTinterwordspacing
E.~E. Papalexakis, C.~Faloutsos, and N.~D. Sidiropoulos, ``Tensors for data
  mining and data fusion: Models, applications, and scalable algorithms,''
  \emph{ACM Trans. Intell. Syst. Technol.}, vol.~8, no.~2, pp. 1--44, 2016.
  [Online]. Available: \url{https://doi.org/10.1145/2915921}
\BIBentrySTDinterwordspacing

\bibitem{Kolda2009}
T.~G. Kolda and B.~W. Bader, ``Tensor decompositions and applications,''
  \emph{SIAM Rev.}, vol.~51, no.~3, pp. 455--500, 2009.

\bibitem{Cichocki2016}
A.~Cichocki, N.~Lee, I.~Oseledets, A.-H. Phan, Q.~Zhao, and D.~P. Mandic,
  ``Tensor networks for dimensionality reduction and large-scale optimization:
  Part 1 low-rank tensor decompositions,'' \emph{Found. Trends Mach. Learn.},
  vol.~9, no. 4--5, p. 249–429, 2016.

\bibitem{Panagakis2021}
Y.~Panagakis, J.~Kossaifi, S.~Clara, G.~G. Chrysos, J.~Oldfield, M.~A.
  Nicolaou, A.~Anandkumar, and S.~Zafeiriou, ``Tensor methods in computer
  vision and deep learning,'' \emph{Proceedings of the IEEE}, vol. 109, no.~5,
  pp. 863--890, 2021.

\bibitem{Song2019}
Q.~Song, H.~Ge, J.~Caverlee, and X.~Hu, ``Tensor completion algorithms in big
  data analytics,'' \emph{ACM Trans. Knowl. Discov. Data}, vol.~13, no.~6, pp.
  1--48, 2019.

\bibitem{rudin1992nonlinear}
L.~I. Rudin, S.~Osher, and E.~Fatemi, ``Nonlinear total variation based noise
  removal algorithms,'' \emph{Physica D: nonlinear phenomena}, vol.~60, no.
  1-4, pp. 259--268, 1992.

\bibitem{getreuer2012total}
P.~Getreuer, ``Total variation inpainting using split bregman,'' \emph{Image
  Processing On Line}, vol.~2, pp. 147--157, 2012.

\bibitem{chambolle2016a}
A.~Chambolle and T.~Pock, ``An introduction to continuous optimization for
  imaging,'' \emph{Acta Numerica}, vol.~25, pp. 161--319, 2016.

\bibitem{cai2022approx}
J.-F. Cai, J.~K. Choi, and K.~Wei, ``Approximation theory of total variation
  minimization for data completion,'' \emph{arXiv preprint arXiv:2207.07473},
  2022.

\bibitem{Li2017a}
X.~Li, Y.~Ye, and X.~Xu, ``Low-rank tensor completion with total variation for
  visual data inpainting,'' in \emph{Proceedings of the AAAI Conference on
  Artificial Intelligence}, vol.~31, no.~1, 2017.

\bibitem{Qiu2021b}
D.~Qiu, M.~Bai, M.~K.~P. Ng, and X.~Zhang, ``Robust low-rank tensor completion
  via transformed tensor nuclear norm with total variation regularization,''
  \emph{Neurocomputing}, vol. 435, pp. 197--215, 2021.

\bibitem{Wang2023a}
H.~Wang, J.~Peng, W.~Qin, J.~Wang, and D.~Meng, ``Guaranteed tensor recovery
  fused low-rankness and smoothness,'' \emph{IEEE Trans. Pattern Anal. Mach.
  Intell.}, vol.~45, no.~9, pp. 10\,990--11\,007, 2023.

\bibitem{Feng2024a}
Q.~Feng, J.~Hou, W.~Kong, C.~Xu, and J.~Wang, ``Poisson tensor completion with
  transformed correlated total variation regularization,'' \emph{Pattern
  Recognition}, vol. 156, p. 110735, 2024.

\bibitem{Zeng2024a}
H.~Zeng, S.~Huang, Y.~Chen, S.~Liu, H.~Q. Luong, and W.~Philips, ``Tensor
  completion using bilayer multimode low-rank prior and total variation,''
  \emph{IEEE Trans. Neural Networks Learn. Syst.}, vol.~35, no.~10, pp.
  13\,297--13\,311, 2024.

\bibitem{Yokota2016a}
T.~Yokota, Q.~Zhao, and A.~Cichocki, ``Smooth parafac decomposition for tensor
  completion,'' \emph{IEEE Trans. on Signal Process.}, vol.~64, no.~20, pp.
  5423--5436, 2016.

\bibitem{Ko2020a}
C.-Y. Ko, K.~Batselier, L.~Daniel, W.~Yu, and N.~Wong, ``Fast and accurate
  tensor completion with total variation regularized tensor trains,''
  \emph{IEEE Trans. on Image Process.}, vol.~29, pp. 6918--6931, 2020.

\bibitem{Liu2012}
J.~Liu, P.~Musialski, P.~Wonka, and J.~Ye, ``Tensor completion for estimating
  missing values in visual data,'' \emph{IEEE Trans. Pattern Anal. Mach.
  Intell.}, vol.~35, no.~1, pp. 208--220, 2012.

\bibitem{Kilmer2011}
M.~E. Kilmer and C.~D. Martin, ``Factorization strategies for third-order
  tensors,'' \emph{Linear Algebra and Its Applications}, vol. 435, no.~3, pp.
  641--658, 2011.

\bibitem{Zhang2016a}
Z.~Zhang and S.~Aeron, ``Exact tensor completion using t-svd,'' \emph{IEEE
  Trans. Signal Process.}, vol.~65, no.~6, pp. 1511--1526, 2016.

\bibitem{Kernfeld2015}
E.~Kernfeld, M.~E. Kilmer, and S.~Aeron, ``Tensor--tensor products with
  invertible linear transforms,'' \emph{Linear Algebra Appl.}, vol. 485, pp.
  545--570, 2015.

\bibitem{Song2023a}
G.~Song, K.~N. Michael, and X.~Zhang, ``Tensor completion by multi-rank via
  unitary transformation,'' \emph{Appl. Comput. Harmon. Anal.}, vol.~65, pp.
  348--373, 2023.

\bibitem{Li2022a}
B.-Z. Li, X.-L. Zhao, T.-Y. Ji, X.-J. Zhang, and T.-Z. Huang, ``Nonlinear
  transform induced tensor nuclear norm for tensor completion,'' \emph{J. Sci.
  Comput.}, vol.~92, no.~83, p.~83, 2022.

\bibitem{Ashraphijuo2017}
M.~Ashraphijuo and X.~Wang, ``Fundamental conditions for low-cp-rank tensor
  completion,'' \emph{J. Mach. Learn. Res.}, vol.~18, no.~63, pp. 1--29, 2017.

\bibitem{Xia2018a}
A.~Zhang and D.~Xia, ``Tensor svd: statistical and computational limits,''
  \emph{IEEE Trans. Inf. Theory}, vol.~64, no.~11, pp. 7311--7338, 2018.

\bibitem{Zhou2017a}
P.~Zhou, C.~Lu, Z.~Lin, and C.~Zhang, ``Tensor factorization for low-rank
  tensor completion,'' \emph{IEEE Trans. Image Process.}, vol.~27, no.~3, pp.
  1152--1163, 2017.

\bibitem{Qiu2024a}
Y.~Qiu, G.~Zhou, Q.~Zhao, and S.~Xie, ``Noisy tensor completion via low-rank
  tensor ring,'' \emph{IEEE Trans. Neural Networks Learn. Syst.}, vol.~35,
  no.~1, pp. 1127 -- 1141, 2024.

\bibitem{Wu2024a}
T.~Wu, B.~Gao, J.~Fan, J.~Xue, and W.~L. Woo, ``Low-rank tensor completion
  based on self-adaptive learnable transforms,'' \emph{IEEE Trans. Neural
  Networks Learn. Syst.}, vol.~35, no.~7, pp. 8826--8838, 2024.

\bibitem{Liu2024a}
S.~Liu, J.~Leng, X.-L. Zhao, H.~Zeng, Y.~Wang, and J.-H. Yang, ``Learnable
  spatial-spectral transform-based tensor nuclear norm for multi-dimensional
  visual data recovery,'' \emph{IEEE Trans. Circuits Syst. Video Technol.},
  vol.~34, no.~5, pp. 3633--3646, 2024.

\bibitem{Li2023a}
B.-Z. Li, X.-L. Zhao, X.~Zhang, T.-Y. Ji, X.~Chen, and M.~K. Ng, ``A learnable
  group-tube transform induced tensor nuclear norm and its application for
  tensor completion,'' \emph{SIAM J. Imaging Sci.}, vol.~16, no.~3, pp.
  1370--1397, 2023.

\bibitem{Mai2024a}
T.~T.~N. Mai, E.~Y. Lam, and C.~Lee, ``Attention-guided low-rank tensor
  completion,'' \emph{IEEE Transactions on Pattern Analysis and Machine
  Intelligence}, vol.~46, no.~12, pp. 9818--9833, 2024.

\bibitem{Aidini2018a}
A.~Aidini, G.~Tsagkatakis, and P.~Tsakalides, ``1-bit tensor completion,''
  \emph{Electronic Imaging}, vol.~13, pp. 1--6, 2018.

\bibitem{Ghadermarzy2019a}
N.~Ghadermarzy, Y.~Plan, and O.~Yilmaz, ``Learning tensors from partial binary
  measurements,'' \emph{IEEE Trans. Signal Process.}, vol.~67, no.~1, pp.
  29--40, 2019.

\bibitem{Hou2021a}
J.~Hou, F.~Zhang, and J.~Wang, ``One-bit tensor completion via transformed
  tensor singular value decomposition,'' \emph{Appl. Math. Model.}, vol.~95,
  no.~1, pp. 760--782, 2021.

\bibitem{Cao2024a}
W.~Cao, X.~Chen, S.~Yan, Z.~Zhou, and A.~Cichocki, ``1-bit tensor completion
  via max-and-nuclear-norm composite optimization,'' \emph{{IEEE} Trans. Signal
  Process.}, vol.~72, pp. 3487--3501, 2024.

\bibitem{Wang2020a}
M.~Wang and L.~Li, ``Learning from binary multiway data: probabilistic tensor
  decomposition and its statistical optimality,'' \emph{J. Mach. Learn. Res.},
  vol.~21, no. 154, pp. 1–--38, 2020.

\bibitem{Zhang2022a}
X.~Zhang and M.~K. Ng, ``Low rank tensor completion with poisson
  observations,'' \emph{IEEE Trans. Pattern Anal. Mach. Intell.}, vol.~44,
  no.~8, pp. 4239--4251, 2022.

\bibitem{zhang2018aTL1}
S.~Zhang and J.~Xin, ``Minimization of transformed l1 penalty: theory,
  difference of convex function algorithm, and robust application in compressed
  sensing,'' \emph{Mathematical Programming}, vol. 169, no.~1, pp. 307--336,
  2018.

\bibitem{Qin2022a}
W.~Qin, H.~Wang, F.~Zhang, J.~Wang, X.~Luo, and T.~Huang, ``Low-rank high-order
  tensor completion with applications in visual data,'' \emph{{IEEE} Trans.
  Image Process.}, vol.~31, pp. 2433--2448, 2022.

\bibitem{Liu2024b}
X.~Liu, J.~Hou, J.~Peng, H.~Wang, D.~Meng, and J.~Wang, ``Tensor compressive
  sensing fused low-rankness and local-smoothness,'' \emph{Proceedings of the
  AAAI Conference on Artificial Intelligence}, vol.~37, no.~7, pp. 8879--8887,
  2024.

\bibitem{Jean2015}
J.~Lafond, ``Low rank matrix completion with exponential family noise,'' in
  \emph{Proceedings of The 28th Conference on Learning Theory}, vol.~40, 2015,
  pp. 1224--1243.

\bibitem{Alaya2019}
\BIBentryALTinterwordspacing
M.~Z. Alaya and O.~Klopp, ``Collective matrix completion,'' \emph{Journal of
  Machine Learning Research}, vol.~20, no. 148, pp. 1--43, 2019. [Online].
  Available: \url{http://jmlr.org/papers/v20/18-483.html}
\BIBentrySTDinterwordspacing

\bibitem{Liu2024c}
C.~Liu, S.~Li, D.~Hu, J.~Wang, W.~Qin, C.~Liu, and P.~Zhang, ``Nonlocal tensor
  decomposition with joint low rankness and smoothness for spectral ct image
  reconstruction,'' \emph{IEEE Trans. Comput. Imaging}, vol.~10, pp. 613--627,
  2024.

\bibitem{Hou2024a}
J.~Hou, X.~Liu, H.~Wang, and K.~Guo, ``Tensor recovery from binary measurements
  fused low-rankness and smoothness,'' \emph{Signal Processing}, vol. 221, p.
  109480, 2024.

\bibitem{KaiHuang2024}
K.~Huang, W.~Kong, M.~Zhou, W.~Qin, F.~Zhang, and J.~Wang, ``Enhanced low-rank
  tensor recovery fusing reweighted tensor correlated total variation
  regularization for image denoising,'' \emph{J. Sci. Comput.}, vol.~99, no.~3,
  p.~69, 2024.

\bibitem{buhlmann2011statistics}
P.~B{\"u}hlmann and S.~Van De~Geer, \emph{Statistics for high-dimensional data:
  methods, theory and applications}.\hskip 1em plus 0.5em minus 0.4em\relax
  Springer Science \& Business Media, 2011.

\bibitem{huetter16}
\BIBentryALTinterwordspacing
J.-C. H{\"u}tter and P.~Rigollet, ``Optimal rates for total variation
  denoising,'' in \emph{29th Annual Conference on Learning Theory}, ser.
  Proceedings of Machine Learning Research, vol.~49, Columbia University, New
  York, USA, 23--26 Jun 2016, pp. 1115--1146. [Online]. Available:
  \url{https://proceedings.mlr.press/v49/huetter16.html}
\BIBentrySTDinterwordspacing

\bibitem{Boyd2011a}
S.~Boyd, N.~Parikh, E.~Chu, B.~Peleato, and J.~Eckstein, ``Distributed
  optimization and statistical learning via the alternating direction method of
  multipliers,'' \emph{Found. Trends Mach. Learn.}, vol.~3, no.~1, pp. 1--122,
  2011.

\bibitem{Cai2013a}
T.~Cai and W.-X. Zhou, ``A max-norm constrained minimization approach to 1-bit
  matrix completion,'' \emph{J. Mach. Learn. Res.}, vol.~14, pp. 3619--3647,
  2013.

\bibitem{ortelli2020}
\BIBentryALTinterwordspacing
F.~Ortelli and S.~van~de Geer, ``Adaptive rates for total variation image
  denoising,'' \emph{Journal of Machine Learning Research}, vol.~21, no. 247,
  pp. 1--38, 2020. [Online]. Available:
  \url{http://jmlr.org/papers/v21/20-301.html}
\BIBentrySTDinterwordspacing

\bibitem{donnat2024one}
C.~Donnat, O.~Klopp, and N.~Verzelen, ``One-bit total variation denoising over
  networks with applications to partially observed epidemics,'' \emph{arXiv
  preprint arXiv:2405.00619}, 2024.

\bibitem{Vershynin2018}
R.~Vershynin, \emph{High-Dimensional Probability: An Introduction with
  Applications in Data Science}, ser. Cambridge Series in Statistical and
  Probabilistic Mathematics.\hskip 1em plus 0.5em minus 0.4em\relax Cambridge
  University Press, 2018.

\bibitem{Klopp2014a}
\BIBentryALTinterwordspacing
O.~Klopp, ``Noisy low-rank matrix completion with general sampling
  distribution,'' \emph{Bernoulli}, vol.~20, no.~1, pp. 282--303, 2014.
  [Online]. Available: \url{http://www.jstor.org/stable/42919393}
\BIBentrySTDinterwordspacing

\bibitem{Klopp2015a}
O.~Klopp, J.~Lafond, {\'E}.~Moulines, and J.~Salmon, ``Adaptive multinomial
  matrix completion,'' \emph{Electron. J. Stat.}, vol.~9, no.~2, pp.
  2950--2975, 2015.

\bibitem{ledoux2013probability}
M.~Ledoux and M.~Talagrand, \emph{Probability in Banach Spaces: isoperimetry
  and processes}.\hskip 1em plus 0.5em minus 0.4em\relax Springer Science \&
  Business Media, 2013.

\bibitem{Davenport2014a}
M.~A. Davenport, Y.~Plan, E.~van~den Berg, and M.~Wootters, ``1-bit matrix
  completion,'' \emph{Information and Inference: A Journal of the IMA}, vol.~3,
  no.~3, pp. 189 -- 223, 2014.

\bibitem{tsybakov2009}
A.~B. Tsybakov, \emph{Nonparametric estimators}.\hskip 1em plus 0.5em minus
  0.4em\relax Springer, 2009.

\bibitem{Zhangshuai2017a}
S.~Zhang and J.~Xin, ``Minimization of transformed penalty: Closed form
  representation and iterative thresholding algorithms,'' \emph{Communications
  in Mathematical Sciences}, vol.~15, no.~2, pp. 511--537, 2017.

\end{thebibliography}

\end{document}